\documentclass[11pt]{article}
\usepackage{latexsym}
\usepackage{amssymb,amsmath}
\usepackage{amsthm}
\usepackage{graphicx}
\usepackage{sgame}
\usepackage{color}
\usepackage[dvipsnames]{xcolor}
\usepackage[numbers,sort&compress]{natbib}
\usepackage[titletoc,toc]{appendix}
\usepackage{authblk}

\usepackage{empheq}
\usepackage{blkarray}
\usepackage{cancel}
\usepackage{breqn}
\usepackage{comment}
\usepackage[caption=false]{subfig} 
\usepackage{appendix}

\usepackage{tocbasic}
\DeclareTOCStyleEntry[
  beforeskip=.2em plus 1pt,% default is 1em plus 1pt
  pagenumberformat=\textbf
]{tocline}{section}

\usepackage{accents}

\usepackage{setspace}

\usepackage{enumerate}
\usepackage[colorlinks]{hyperref}
\usepackage[colorinlistoftodos]{todonotes}

\numberwithin{equation}{section}

\newcommand{\NN}{\mathbb{N}}

\newcommand{\RR}{\mathbb{R}}

\newcommand{\ds}{\displaystyle}
\newcommand{\mc}{\mathcal}
\newcommand{\ol}{\overline}

\newcommand{\bbm}{\begin{bmatrix}}
\newcommand{\bpm}{\begin{pmatrix}}
\newcommand{\ebm}{\end{bmatrix}}
\newcommand{\epm}{\end{pmatrix}}

 \newcommand{\del}[2]{\frac{\partial #1}{\partial #2}}
 \newcommand{\dsdel}[2]{\displaystyle\frac{\partial #1}{\partial #2}}

\newcommand{\myindent}{\hspace{10mm}}

\newcommand{\dsddx}[2]{\displaystyle\frac{d #1}{d #2}}
\newcommand{\dsddt}[1]{\displaystyle\frac{d #1}{dt}}

\newcommand{\holder}{H{\"o}lder\:}

\newtheorem*{assumption*}{\assumptionnumber}
\providecommand{\assumptionnumber}{}


\newcommand{\supholderupperbound}[1]{B(#1)}
\newcommand{\supholderlowerbound}[1]{b(#1)}
\newcommand{\infholderlowerbound}[1]{d(#1)}
\newcommand{\infholderupperbound}[1]{D(#1)}

\newcommand{\paren}[1]{\left(#1\right)}

\newcommand{\D}[2]{\frac{d#1}{d#2}}
\newcommand{\PD}[2]{\frac{\partial#1}{\partial#2}}

\newcommand{\abs}[1]{\left\lvert #1 \right\rvert}
\newcommand{\norm}[1]{\left\lVert #1 \right\rVert}

\newcommand{\ov}[1]{\overline{#1}}

\title{Long-Time Behavior of a PDE Replicator Equation for Multilevel Selection in Group-Structured Populations}

\author[1,2]{Daniel B. Cooney}
\author[1,2,3]{Yoichiro Mori}

\affil[1]{\small{Center for Mathematical Biology, University of Pennsylvania, Philadelphia, PA 19104, USA}}
\affil[2]{Department of Mathematics, University of Pennsylvania, Philadelphia, PA 19104, USA}
\affil[3]{Department of Biology, University of Pennsylvania, Philadelphia, PA 19104, USA}

\begin{document}

\newtheorem{theorem}{Theorem}[section]
\newtheorem{definition}[theorem]{Definition}
\newtheorem{lemma}[theorem]{Lemma}
\newtheorem{corollary}[theorem]{Corollary}
\newtheorem{claim}[theorem]{Claim}
\newtheorem{fact}[theorem]{Fact}
\newtheorem{proposition}[theorem]{Proposition}
\newtheorem{remark}[theorem]{Remark}
\newtheorem{observation}[theorem]{Observation}
\newtheorem{example}[theorem]{Example}

\renewcommand{\baselinestretch}{1.1}

% to get nice proofs ...
\newcommand{\qedsymb}{\mbox{ }~\hfill~{\rule{2mm}{2mm}}}

\maketitle
\singlespacing
\begin{abstract}
    In many biological systems, natural selection acts simultaneously on multiple levels of  organization. This scenario typically presents an evolutionary conflict between the incentive of individuals to cheat and the collective incentive to establish cooperation within a group. Generalizing previous work on multilevel selection in evolutionary game theory, we consider a hyperbolic PDE model of a group-structured population, in which members within a single group compete with each other for individual-level replication; while the group also competes against other groups for group-level replication. We derive a threshold level of the relative strength of between-group competition such that defectors take over the population below the threshold while cooperation persists in the long-time population above the threshold. Under stronger assumptions on the initial distribution of group compositions, we further prove that the population converges to a steady state density supporting cooperation for between-group selection strength above the threshold. We further establish long-time bounds on the time-average of the collective payoff of the population, showing that the long-run population cannot outperform the payoff of a full-cooperator group even in the limit of infinitely-strong between-group competition. When the group replication rate is maximized by an intermediate level of within-group cooperation, individual-level selection casts a long shadow on the dynamics of multilevel selection: no level of between-group competition can erase the effects of the individual incentive to defect. We further extend our model to study the case of multiple types of groups, showing how the games that groups play can coevolve with the level of cooperation.
\end{abstract}

{\hypersetup{linkbordercolor=black, linkcolor = black}
\begin{spacing}{0.01}
\renewcommand{\baselinestretch}{0.1}\normalsize
\tableofcontents
\addtocontents{toc}{\protect\setcounter{tocdepth}{2}}
%\addtocontents{toc}{~\vspace{-3\baselineskip}}
\end{spacing}
%\clearpage
\singlespacing

\section{Introduction}

\myindent Across a variety of biological and social systems, population structure often induces selective forces operating at multiple levels of organization. Of particular interest are hierarchical structures in which there is a tug-of-war between the interests at a one level of organization and the interests at a larger level. For problems of cooperation or collective behavior, the incentives of an individual to be a free-rider are often misaligned with the incentives of its group to produce a collective benefit for all of its members \cite{levin2010crossing}. Considering the effects of selection at multiple levels of organization is particularly important in systems on the cusp of undergoing a transition to a higher order of complexity, establishing a collective unit that can compete or replicate as a single unit. Multilevel selection has been invoked to describe major evolutionary transitions, with examples ranging from the evolution of multicellularity \cite{rainey2003evolution,tarnita2013evolutionary,rose2020meta,pichugin2015modes,pichugin2018reproduction,staps2019emergence}, to the evolution of social group structure \cite{nowak2010evolution,fu2015risk}. The transition to higher levels of biological complexity can be understood as a triumph of cooperative behavior via multilevel selection, in which groups can form a cooperative population structure overcoming individual competition within the group \cite{szathmary1995major,nowak2006five,taylor2007transforming}.

\myindent Questions about multilevel selection have ranged widely across scales. On one end of the spectrum, ideas of conflict between individuality and collective behavior have been considered for the evolution of multicellularity \cite{tarnita2013evolutionary,pichugin2018reproduction}, replication control of plasmids \cite{paulsson2002multileveled}, and the evolution of mutualism in the microbiome \cite{van2019role}. At the other end, the alignment of the individual-level and group-level incentives have been studied for problems ranging from collective hunting in animal groups \cite{boza2010beneficial} and the eusocial structure of insect colonies \cite{nowak2010evolution,fu2015risk} to the establishment of cooperative institutions for the management of common-pool resources \cite{ostrom2010beyond,schluter2016robustness,tavoni2012survival} and within-group cooperation coevolving with warfare \cite{henriques2019acculturation,turchin2010warfare} in human societies. Experimental and field work has addressed problems ranging from the cooperative cofounding of ant colonies \cite{shaffer2016foundress}, to the establishment of multicellularity in biofilms \cite{rose2020meta}, and the artificial selection for nonaggression in chickens \cite{muir1996group,bijma2007multilevel1,bijma2007multilevel2, wade2010group}. A natural tension between selective forces at different levels of selection arises in the evolution of virulence in infectious disease dynamics, in which competition for pathogen replication within an individual host promotes selection for more virulent pathogens, while increased virulence can also harm the host and prevent onward transmission within the host population \cite{levin1981selection,dwyer1990simulation,gilchrist2004optimizing,gilchrist2006evolution,blackstonevariation,boni2013virulence}.

\myindent One framework that has often been used to study multilevel selection is evolutionary game theory, which provides stylized models for the evolution of cooperative behavior in which individuals can maximize payoff by cheating, while groups achieve higher collective payoffs when at least some of their members cooperate. Traulsen and coauthors have studied the evolution of cooperation in the presence of multilevel selection, showing that group-level competition for replication could help to promote the fixation of cooperators over defectors in finite populations \cite{traulsen2005stochastic,traulsen2006evolution,traulsen2008analytical}.
The work of Simon and coauthors has further explored how more realistic mechanisms like group-level fission and fusion events and the possibility of non-constant group size can help to facilitate cooperation win out over the defection that is favored by individual-level selection \cite{simon2010dynamical,simon2012numerical,simon2013towards,simon2016group}. Further work on stochastic multilevel selection models in evolutionary games have explored the role of group-level extinctions \cite{bottcher2016promotion}, the role of spatial structure on between-group competition \cite{akdeniz2020cancellation}, and asymptotic formulas for fixation probabilities in the limit of large population size \cite{mcloone2018stochasticity}.

\myindent Luo introduced a stochastic model of two-level selection featuring two types of individuals: one with a constant reproductive advantage at the individual level (i.e. defectors), and the other that confers a selective advantage to its group (i.e. cooperators) \cite{luo2014unifying,van2014simple}. In the limit of infinitely many groups of infinite size, Luo derived a non-local hyperbolic PDE describing the simultaneous competition within and between groups \cite{luo2014unifying}. Luo and Mattingly characterized the long-time behavior of this PDE based upon the relative strengths of selection at the two levels and the \holder exponent of the initial condition near the full-cooperator group \cite{luo2017scaling}. They showed that there was a threshold level of between-group selection strength such that defectors would fix in the population when between-group competition below the threshold, while the population converges to a steady state density supporting positive levels of cooperation for between-group selection above the threshold. Further work on related nested birth-death models for multilevel selection has explored application to host-pathogen dynamics \cite{pokalyuk2019diversity,pokalyuk2019maintenance,osorio2020two}, as well as mathematical aspects of behavior in alternate infinite-population scaling limits, including fixation probabilities in stochastic Fleming-Viot models \cite{luo2017scaling,dawson2013multilevel,meizis2020convergence} and quasi-stationary distributions in a Wright-Fisher diffusion equation with multilevel selection \cite{velleret2019two, velleret2020individual}.

 \myindent This model of two-level selection was later extended to include individual-level and group-level birth rates that depended on the personal and collective payoffs obtained from a two-strategy games played between members of the groups \cite{cooney2019replicator}. Results analogous to those of Luo and Mattingly were demonstrated for special cases of the Prisoners' Dilemma (PD) and Hawk-Dove (HD) games in which the within-group dynamics were exactly solvable, and further work explored the multilevel dynamics all two-player, two-strategy social dilemmas \cite{cooney2020analysis} and in the presence of within-group mechanisms of assortment or reciprocity \cite{cooney2019assortment}. For the PD and HD games, it was conjectured that, for sufficiently strong between-group competition, the population would converge to the unique steady state with the same \holder exponent near $x=1$ as that of the initial distribution. These steady state densities displayed a surprising property, called the ``shadow of lower-level selection", in which the payoff of the modal group composition at steady state and the average payoff of the steady state population were limited by the payoff of the full-cooperator group. As a result, for games in which group payoff was maximized by intermediate levels of cooperation, the population always features less cooperation that optimal, even in the limit of infinitely strong between-group competition.

\myindent In this paper, we characterize the long-time behavior for a broad class of models for multilevel selection. While previous work had shown convergence to steady state densities for several one-parameter families of models for multilevel selection arising from special cases of PD games \cite{luo2017scaling,cooney2019replicator}, the techniques used in those cases relied on the ability to obtain explicit solutions for the characteristic curves describing the within-group dynamics,  making it difficult to extend the results to the more general situation. Here, we obtain careful estimates on the solutions along characteristics and extract a principal growth rate for the multilevel dynamics, allowing us to prove convergence to steady state for multilevel PDEs with continuously differentiable individual-level and group-level replication rates in which defectors have an individual-level advantage over cooperators and all-cooperator groups have a collective advantage over all-defector groups. This result confirms a previous conjecture for the long-time dynamics multilevel replicator equations arising from any PD game, and extends the scope of that conjecture to include applications to a range of topics including protocell evolution and the origin of chromosomes \cite{cooney2021pde}. Through this more general formulation of our model for multilevel selection, we are able to formalize previous intuition to understand how the possibility of achieving cooperation at steady state relies on the ability for the collective advantage of full-cooperator groups over full-defector groups to overcome the individual-level advantage of defecting in a group with many cooperators. %

\myindent We also extend our analysis our generalization of the multilevel Prisoners' Dilemma dynamics to study long-time behavior for initial populations beyond the class of measures with a well-defined H{\"o}lder exponent near full-cooperation considered in previous work \cite{luo2017scaling,cooney2019replicator,cooney2020analysis}. By characterizing the tail behavior of the initial measure through quantities that we call the supremum and infimum H{\"o}lder exponent, we can obtain upper and lower bounds for the principal growth rate for solutions along characteristics for any initial measure. Using these estimates, we find that the population will not converge to a density steady for any initial measure without a well-defined H{\"o}lder exponent near full-cooperation. However, we show that, for any initial measure, defectors will take over the population when between-group competition is sufficiently weak, while cooperation will survive in the long-time limit in the sense of weak persistence when between-group selection exceeds a threshold value that depends on the supremum H{\"o}lder exponent of the initial measure. We also use our estimates on the principal growth rate to derive long-time upper and lower bounds on the time-average of the average group-level replication rate of the population, showing that the long-time collective outcome cannot exceed the group-level replication rate of the all-cooperator group. This observation serves as a dynamical analogue of the ``shadow of lower-level selection", showing how this limitation on collective outcome stems from the tug-of-war between the collective incentive to cooperate and the individual incentive to defect.

\myindent We also characterize the multilevel dynamics for a generalization of the Prisoners' Delight game, in which cooperation is favored at both levels of selection, extracting a principal growth rate to show how full-cooperation is achieve via multilevel selection. Combining this with the results mentioned above for our generalization of the multilevel Prisoners' Dilemma dynamics, we fully characterize the dynamics of our two-level replicator equation for scenarios in which the within-group dynamics features no interior equilibria. This approach of extracting a principal growth rate for a population can be further applied to understand a generalization of the multilevel dynamics to a case in which our group-structured population may feature different possible individual-level and group-level replication rates.  We provide a sufficient condition for the long-time concentration of the population upon a single type of group under our multi-type two-level birth-death dynamics, showing that the long-time dynamics will favor the dominance of the group type featuring the maximal principal growth rate. This result can be used study the coevolution of group features and the strategic composition of groups, showing how multilevel competition can help to select the games played within groups. As one application of this framework with multiple group types, we show that this concentration result can be used to study the dynamics of generalized versions of Hawk-Dove  and Stag-Hunt games, which each feature an interior within-group equilibrium. The results for these cases confirm and extend existing conjectures on convergence to steady state for the multilevel dynamics of those games \cite{cooney2020analysis}. 

\myindent In Section \ref{sec:model}, we describe the mathematical formulation of our model of multilevel selection  with comparisons to previous work on multilevel selection in evolutionary games. In Section \ref{sec:summary}, we summarize our main results for the long-time dynamics of our PDE model of multilevel selection. Section \ref{sec:outline} provides an outline for the remainder of the paper.

\subsection{Model of Multilevel Selection} \label{sec:model}

For our model of multilevel selection, we consider a population with $m$ groups that is each composed on $n$ members. The within-group selection follow a frequency-dependent Moran process  replacing a randomly chosen member of the same group.
In a group with $i$ cooperators, cooperators and defectors give birth with rates $1 + w_I \pi_C(\frac{i}{n})$ and $1 + w_I \pi_D(\frac{i}{n})$, respectively, where $\pi_C(\cdot)$ and $\pi_D(\cdot)$ are $C^1$ functions on $[0,1]$ and $w_I$ is the intensity of selection for within-group competition. We can further consider the advantage of defectors over cooperators under within-group competition in an $i$-cooperator group through the quantity \begin{equation} \label{eq:piofi}  \pi\left(\tfrac{i}{n}\right) :=  \pi_D\left(\tfrac{i}{n}\right) - \pi_C\left(\tfrac{i}{n}\right). \end{equation} 
Between-group competition takes place through a group-level birth-death process in which a group with $i$ cooperators produces a copy of itself and replaces a randomly chosen group with rate $\Lambda(1 + w_G G(\frac{i}{n}))$, where $w_G$ is the selection intensity of between-group competition and $\Lambda$ describes the relative rate of within-group and between-group replication events. We note that the choice of $\pi(x) = s$ and $G(x) = x$ recovers the functions for within-group and between-group competition for the Luo-Mattingly model \cite{luo2014unifying,luo2017scaling}.

In the limit as the number of groups and group size tend to infinity ($m \to \infty, n \to \infty$), we can describe the composition of strategies in the group-structured population by $u(t,x)$, the probability density of groups composed $x$ cooperators and $1-x$ defectors at time $t$. Using either a heuristic derivation \cite{luo2014unifying,van2014simple,cooney2019replicator} or a weak convergence argument \cite{luo2017scaling}, we can show that the large-population limit of the stochastic ball-and-urn process can be describe by the following partial differential equation for the evolution of $u(t,x)$ 
\begin{equation}\label{multiselect}
\PD{u(t,x)}{t}=\PD{}{x}\paren{x(1-x)\pi(x)u(t,x)} +\lambda \paren{G(x)-\int_0^1 G(y)u(t,y)dy}u(t,x),
\end{equation}
where $\lambda := \frac{\Lambda w_G}{w_I}$ describes the relative strength of within-group and between-group competition. The first term on the right-hand side of Equation \ref{multiselect} describes the dynamics of within-group competition, in which defectors (respectively cooperators) increase in frequency within groups when $\pi(x) > 0$ (respectively $\pi(x) < 0$). The second term in Equation \ref{multiselect} describes the impact of between-group competition, and groups with composition $x$ increase in frequency when their replication rate $G(x)$ exceeds the average group-replication rate in the population $\int_0^1 G(y) u(t,y) dy$.

Equation \eqref{multiselect} is paired with initial data given by
\begin{equation}\label{ic}
u(0,x)=u_0(x)\geq 0, \quad \int_0^1 u_0(x)dx=1.
\end{equation}
We can check that, if $u(x,t)$ is a solution to Equation \ref{multiselect}, it will be of unit mass for all $t$. Equation \ref{multiselect} is a hyperbolic PDE, whose characteristic curves are given by solutions of the following ODE
\begin{equation} \label{eq:replicatorcharacteristics}
    \dsddt{x(t)} = - x (1 - x) \pi(x) \: \:, \: \: x(0) = x_0,
\end{equation}
which we note is the well-known replicator equation for individual-level selection within a given group \cite{hofbauer1998evolutionary}. The function $\pi(x)$, sometimes called the gain function \cite{bach2006evolution,kaznatcheev2017cancer} describes the relative advantage of defectors over cooperators under within-group competition in an $x$-cooperator group.

We can also consider a measure-valued formulation corresponding to the multilevel dynamics described by Equation \eqref{multiselect}. For an initial Borel probability measure $\mu_0(dx)$ and any $C^1([0,1])$ test-function $v(x)$, the $\mu_t(dx)$ evolves according to
\begin{equation} \label{eq:measuremultiselect}
\dsdel{}{t} \int_0^1 v(x) \ol{\mu}_t(dx) = -\int_0^1 \dsdel{v(x)}{x} x(1-x) \pi(x) \ol{\mu}_t(dx) + \lambda \int_0^1 v(x) \left[ G(x) - \int_0^1 G(y) \ol{\mu}_t(dy) \right] \ol{\mu_t}(dx).    
\end{equation}
To study solutions for Equation \ref{eq:measuremultiselect}, we can introduce an auxiliary linear equation given by
\begin{equation} \label{eq:linearmeasuremultiselect}
\dsdel{}{t} \int_0^1 v(x) \mu_t(dx) = -\int_0^1 \dsdel{v(x)}{x} x(1-x) \pi(x) \mu_t(dx) + \lambda \int_0^1 v(x)  G(x) \mu_t(dx),     
\end{equation}
paired with initial the measure $\mu_0(dx)$. We can check that solutions $\ol{\mu}_t(dx)$ to the measure-valued multilevel dynamics of Equation \eqref{eq:measuremultiselect} can be related to solutions $\mu_t(dx)$ of Equation \eqref{eq:linearmeasuremultiselect} through the normalization given by
\begin{equation} \label{eq:olmutnormalized}
\ol{\mu}_t(dx) = \frac{\mu_t(dx)}{\int_0^1 \mu_t(dy)}.
\end{equation}

The dynamics of Equation \eqref{eq:linearmeasuremultiselect} also have independent biological interest for studying multilevel selection in which $x$-cooperator reproduce with rate $G(x)$ and no groups are removed from the population. Such models of an expanding group-structured population may be relevant in applications in which the group-level reproduction corresponds to cell division \cite{fontanari2006coexistence,fontanari2013solvable,fontanari2014effect,fontanari2014nonlinear} or fission of social groups \cite{simon2016group,gueron1995dynamics}.

We will also consider an extension of the multilevel dynamics in which groups belong to one of $N$ possible subpopulations, with each subpopulation featuring its own reproduction rates $\pi_j(\cdot)$ and $G_j(\cdot)$. Within-group competition proceeds according to $\pi_j(\cdot)$, while between-group competition consists of group replicating with rate proportion to $G_j(\cdot)$ and replacing a randomly-chosen group from any of the $N$ subpopulations. For example, each subpopulation could be defined by a different two-strategy game played within its groups, and then the corresponding multilevel dynamics describe the coevolution of cooperation and the fraction of groups playing each game. 
 
 Denoting the set of subpopulations by $\mathcal{N} = \{1,\cdots,N\}$, we describe the composition of $x$-cooperator groups in the subpopulation $j \in \mathcal{N}$ at time $t$ by the density $u^j(t,x)$. This family of densities evolves according  to PDEs of the form
 \begin{dmath} \label{eq:subpopulationdensity}
\dsdel{u^j(t,x)}{t} =  \dsdel{}{x}\left[x(1-x) \pi_j(x) u^j(t,x) \right] + \lambda u^j(t,x) \left[G_j(x) - \ds\sum_{k=1}^N \int_0^1 G_k(y) u^k(t,y) dy \right], 
 \end{dmath}
for each $j \in \mathcal{N}$, and this system is paired with initial data satisfying
\begin{equation}
    u^j(t,x) = u^j_0(x) \geq 0, \: \: \ds\sum_{k=1}^N \int_0^1 u^j_0(x) dx = 1.
\end{equation}
 We can also consider a measure-valued analogue of our multipopulation model by describing the strategic composition of groups in the $j$th subpopulation by the measure $\mu^j_t(dx)$. For a $C^1([0,1])$ test-function $v_j(x)$, this measure evolves according to the following equation 
 \begin{dmath} \label{eq:measurevaluedsubpopulation}
     \dsdel{}{t} \int_0^1 v_j(x) \mu^j_t(dx) = -\int_0^1 \dsdel{v_j(x)}{x} x (1-x) \pi_j(x) \mu^j_t(dx) + \lambda \int_0^1 v_j(x)  \left[ G_j(x) - \ds\sum_{i=1}^n \left(\int_0^1 G_i(y) \mu_t(dy) \right) \right] \mu_t(dx), 
 \end{dmath}
where the subpopulations are coupled through the nonlocal regulation term describing between-group competition. The system described by Equation \eqref{eq:measurevaluedsubpopulation} is paired with initial data given by the measures $\mu_0^j(dx)$ for $j \in \mathcal{N}$, which together satisfy the normalization condition given by
 \begin{equation} \label{eq:multipopinitial}
     \ds\sum_{j = 1}^{N} \mu_0^j\left([0,1]\right) = \ds\sum_{j = 1}^{N} \int_0^1 \mu^j_0(dx) 
     = 1.
 \end{equation}
 
 We can also associate with Equation \eqref{eq:measurevaluedsubpopulation} a system of $N$ decoupled linear equations of the form
 \begin{dmath} \label{eq:measurevaluedsubpopulationlinear}
     \dsdel{}{t} \int_0^1 v_j(x) \mu^j_t(dx) = -\int_0^1 \dsdel{v_j(x)}{x} x (1-x) \pi_j(x) \mu^j_t(dx) + \lambda \int_0^1 v_j(x) G_j(x)  \mu_t(dx).  
 \end{dmath}
 Given solutions $\mu^1_t(dx), \cdots, \mu^N_t(dx)$ to the linear dynamics of Equation \eqref{eq:measurevaluedsubpopulationlinear}, we can find a corresponding solution $\ol{\mu}^j_t(dx)$ to Equation \eqref{eq:measurevaluedsubpopulation} for the $j$th subpopulation through the normalization given by
 \begin{equation} \label{eq:mubartj}
    \int_0^1 v_j(x) \ol{\mu}^j_t(dx) = \frac{\int_0^1 v_j(x) \mu^j_t(dx)}{\sum_{j=1}^N \mu^j_t\left([0,1] \right)} 
 \end{equation}

\subsubsection{Motivating Example: Two-Strategy Evolutionary Games}
\label{sec:gamemotivation}

To formulate assumptions about the behavior of the functions $\pi(x)$ and $G(x)$ characterizing within-group and between-group competition,%
 we can consider the special case of the multilevel selection dynamics depend on payoffs from two-player, two-strategy social dilemmas \cite{cooney2019replicator,cooney2020analysis}. We consider games with symmetric payoff matrices of the form 
\begin{equation} \label{eq:payoffmatrix}
\begin{blockarray}{ccc}
& C & D \\
\begin{block}{c(cc)}
C & R & S \\
D & T & P \\
\end{block}
\end{blockarray},
\end{equation}
where the entries of the payoff matrix correspond to the reward for mutual cooperation ($R$), the sucker payoff from cooperating with a defector ($S$), the temptation to defect against a cooperator $(T)$, and the punishment for mutual defection ($P$). In this paper, we will consider the multilevel dynamics corresponding to generalizations of four two-strategy social dilemmas: the Prisoners' Dilemma (PD), the Hawk-Dove game (HD), the Stag-Hunt (SH), and the Prisoners' Delight (PDel). These four games are characterized by the following rankings of payoffs
\begin{subequations} \label{eq:payoffrankings}
\begin{align}
    \mathrm{PD} &: T > R > P > S \label{eq:PDpayoffs} \\
    \mathrm{HD} &: T > R > S > P \label{eq:HDpayoffs} \\
    \mathrm{SH} &: R > T > P > S \label{eq:SHpayoffs} \\
    \mathrm{PDel} &: R > T > S > P \label{eq:PDelpayoffs}.
\end{align}
\end{subequations}
For a group composed of fractions $x$ cooperators and $1-x$ defectors, the average payoff for a cooperator and defector are given by 
\begin{subequations} \label{eq:gamepayoffs}
\begin{align}
    \pi_C(x) &= R x + S (1-x) \\
    \pi_D(x) &= T x + P (1-x).
\end{align}
\end{subequations}
 In previous work on multilevel selection in evolutionary games, it was assumed that the group-level reproduction rate in an $x$-cooperator group depend on the average payoff of group members $G(x) = x \pi_C(x) + (1-x) \pi_D(x)$ \cite{cooney2019replicator,cooney2020analysis}. Using the cooperator and defector payoffs from Equation \eqref{eq:gamepayoffs}, we then have that the dynamics of Equation \eqref{multiselect} have the following dependence on the payoff matrix from Equation \eqref{eq:payoffmatrix}
\begin{subequations} \label{eq:gamepiG}
\begin{align} \label{eq:pixgame}
    \pi(x) &= P - S - \left(R - S - T + P \right) x \\
    G(x) &= P + \left( S + T - 2P \right)x + \left( R - S - T + P \right) x^2 \label{eq:Gxgame}.
\end{align}
\end{subequations}
From Equation \ref{eq:Gxgame}, we can see that the group-reproduction function satisfies 
\begin{equation} \label{eq:G01ranking}
G(1) = R > P = G(0)
\end{equation}
for each of the social dilemmas described in Equation \eqref{eq:payoffrankings}. Under the game-theoretic model, the characteristic curves from Equation \eqref{eq:replicatorcharacteristics} evolve according to
\begin{equation} \label{eq:characteristicsreplicatorgame}
   \dsddt{x(t)} = - x (1-x) \left[P - S - \left( R - S - T + P \right) x \right], 
\end{equation}
which has equilibria at $x = 0$, $x = 1$, and a possible interior equilibrium $x = x_{eq}$ given by
\begin{equation} \label{eq:interiorequilibrium}
    x_{eq} = \frac{P - S}{R - S - T + P}.
\end{equation}

For the PD game, we can use the payoff rankings and the fact that $\pi(x)$ is an affine function to see that
\[ \pi(x) \geq \min\left(\pi(0),\pi(1) \right) = \min\left(P - S,T - R \right) > 0, \]
and therefore the within-group dynamics feature global stability of the all-defector equilibrium $x = 0$. 
For the Prisoners' Delight game, we can use the payoff rankings from Equation \eqref{eq:PDelpayoffs} to see that 
\begin{equation} \label{eq:PDelpix}
    \pi(x) < \max(\pi(0),\pi(1)) = \max(P-S,T-R) < 0,
\end{equation}
and therefore full-cooperation is globally stable under individual-level selection.

Using these properties, we formulate a generalization of the multilevel PD and PDel dynamics by considering $G(x), \pi(x) \in C^1\left(\left[0,1\right]\right)$ satisfying $G(1) > G(0)$ and either $\pi(x) > 0$ (in the PD case) or $\pi(x) < 0$ (in the PDel case). This class of reproduction functions $G(x)$ and $\pi(x)$ include those used in the Luo-Mattingly model \cite{luo2014unifying}, in models used to describe the evolution of protocells \cite{cooney2021pde}, and within-group dynamics following the Fermi update rule for social learning \cite{traulsen2005stochastic}. We will take this generalization of the PD and PDel dynamics as a generic picture of multilevel selection in populations without internal within-group equilibria, and will assume that the dynamics of multiple subpopulations described by Equation \eqref{eq:subpopulationdensity} reflects a PD or PDel scenario within a given subpopulation. For games such as the HD and SH with internal within-group equilibria, we will apply the multiple subpopulation formulation using the fact that the dynamics above and below the within-group equilibria for such games reflect either a PD or PDel scenario.

\subsection{Summary of Main Results}
\label{sec:summary}

Now we present our main results for the dynamics of solutions to Equation \eqref{multiselect}. In Section \ref{sec:steadystatesummary}, we introduce the family of probability densities that are steady state solutions of Equation \eqref{multiselect} in the PD case, highlighting how the collective at steady state is limited by the group-level reproduction rate of the all-cooperator group. In Section \ref{sec:convergencesteadyresults}, we present Theorem \ref{wlimtheorem}, showing convergence of the population to a steady state density for the special of initial measures with a well-defined H{\"o}lder exponent near $x=1$. In Section \ref{sec:longtimeresults}, focus on the characterization of long-time behavior of the multilevel PD dynamics for all possible initial measures, providing bounds on the long-time collective outcome in Theorem \ref{thm:groupbounds} and classifying the conditions for long-time extinction or persistence of cooperation \ref{thm:extinctvspersist}. In Section \ref{sec:multiplepopulationsresults}, we turn to the dynamics of Equation \eqref{eq:subpopulationdensity} describing multilevel competition in the presence of multiple types of groups. We present Theorem \ref{thm:longtimemultiple}, providing a sufficient condition for the population to concentrate upon the group type with maximal principal growth rate. 

\subsubsection{Steady State Densities for Multilevel PD Dynamics}
\label{sec:steadystatesummary}

 We first look to understand steady-state solutions to Equation \eqref{multiselect} in the Prisoners' Dilemma case, and the conditions under which such solutions can be achieved under the multilevel dynamics. In Section \ref{sec:steadystates}, we show that, under generic conditions on $\lambda$, $G(\cdot)$, and $\pi(\cdot)$, the achievable steady states are delta-concentrations of full-defector groups $\delta(x)$ and full-cooperator groups $\delta(1-x)$, as well as a family of density steady states which we characterize below. 

In previous work on special cases of Equation \eqref{multiselect}, the long-time behavior and convergence to steady state of the multilevel dynamics was studied for initial measures $\mu_0(dx)$ with given \holder exponent near $x=1$ \cite{luo2017scaling,cooney2019replicator,cooney2020analysis}. This \holder exponent and its associated \holder constant quantify the extent to which the initial distribution concentrates or decays near the full-cooperator group, and is defined as follows.

\begin{definition} \label{def:holderexponent}
The measure $\mu_t(dx)$ has \holder exponent $\theta_t \geq 0$ near $x =1$ with associated H{\"o}lder constant $C_{\theta}  \in \RR_{\geq 0} \cup \{ \infty \}$ if it satisfies the following limiting behavior
\begin{equation} \label{eq:holderexponent}
    \lim_{x \to 0} \frac{\mu_t\left([1-x,1]\right)}{\Theta} = 
    \left\{
     \begin{array}{lr}
       0 & : \Theta < \theta \\
       C_{\theta} & : \Theta =  \theta \\
       \infty & : \Theta > \theta
     \end{array}
   \right. .
\end{equation}
\end{definition}
Measures of the form $\mu(dx) = \theta (1-x)^{\theta -1} dx$ for finite $\theta > 0$ have \holder exponent $\theta$ near $x=1$. Examples with \holder exponent of $0$ and $\infty$ are measures satisfying $\mu(\{1\}) > 0$ and $\mu\left([1-\epsilon,1]\right) = 0$ for some $\epsilon > 0$, respectively.

We show in Section \ref{sec:steadystates} that steady state density solutions to Equation \eqref{multiselect} can be parametrized by $\lambda$ and their \holder exponent $\theta > 0$ near $x=1$. We calculate that these steady states are given by the following densities
\begin{equation}\label{pthetactilde}
\begin{split}
p^{\lambda}_{\theta}(x)&=\frac{f^{\lambda}_{\theta}(x)}{\int_0^1 f^{\lambda}_{\theta}(x)dx},\\
f^{\lambda}_{\theta}(x)&= x^{\nu_{\theta}- 1}(1-x)^{\theta-1}\frac{\pi(1)}{\pi(x)}\exp\paren{ - \lambda \int_x^1\frac{\tilde{C}(s)ds}{\pi(s)}},
\end{split}
\end{equation}
where the parameter $\nu_\theta$ corresponds to 
\begin{equation} \label{eq:nudefinition}
    \nu_{\theta}:= \left(\frac{1}{\pi(0)}\right) \left( \lambda \left[G(1) - G(0) \right] - \theta \pi(1) \right)
\end{equation}
and the term $-\lambda \tilde{C}(x)$ is given by
\begin{dmath} \label{eq:lambdaCofx}
- \lambda \tilde{C}(x) = \lambda \left( \frac{G(x) - G(0)}{x} \right) + \nu_\theta \left( \frac{\pi(x) - \pi(0)}{x} \right)  
+ \lambda \left( \frac{G(x) - G(1) }{1 - x} \right) - \theta \left( \frac{\pi(x) - \pi(1)}{1 - x} \right).
\end{dmath}
The form of these steady state densities highlights the key contribution of the of the collective reproduction rates $G(1)$ and $G(0)$ and individual-level gain functions $\pi(1)$ and $\pi(0)$ for the all-cooperator and all-defector groups. 

Furthermore, our assumption of $C^1([0,1])$ replication rates $G(\cdot)$ and $\pi(\cdot)$ allows us to use Equation \eqref{eq:lambdaCofx} to deduce the boundedness of $\lambda \tilde{C}(x)$ on $[0,1]$.  Because we are considering $\theta > 0$, we see from Equations \eqref{pthetactilde} and \eqref{eq:nudefinition} that the density $f^{\lambda}_{\theta}(x)$ is integrable provided that $\nu_{\theta}> 0$, which occurs when between-group competition is sufficiently strong so that
\begin{equation} \label{eq:lambdastar}
    \lambda > \lambda^* := \frac{\theta \pi(1)}{G(1) - G(0)}.
\end{equation}
This threshold condition is increasing in $\pi(1)$, the relative within-group advantage of a defector over a cooperator in a full-cooperator group, and is decreasing in $G(1) - G(0)$, the relative between-group advantage of a full-cooperator group over a full-defector group. From this we see that the ability to promote cooperation via multilevel selection can be understood as the collective incentive to cooperate winning out over the individual incentive to defect against cooperators. In addition, $\lambda^*$ is an increasing function of $\theta$, which means that larger cohorts of near full-cooperator groups are able to maintain a steady state density featuring cooperation over a large range in the strength of between-group competition. 

We further show in Section \ref{sec:steadystates} that the average level of group reproduction $G(x)$ achieved in the steady state distribution $p^{\lambda}_{\theta}$ is given by

\begin{equation} \label{eq:averageofthetawithlambdastar}
    \langle G(\cdot) \rangle_p = \int_0^1 G(y) p^{\lambda}_{\theta}(y) dy = \left(\frac{\lambda^*}{\lambda}\right) G(0) + \left( 1 - \frac{\lambda^*}{\lambda} \right) G(1).
\end{equation}
The average group reproduction function (or collective payoff in the game-theoretic scenario) interpolates between the collective outcome of all-defector groups to that of all-cooperator groups as $\lambda$ ranges from $\lambda^*$ to $\infty$, and that the average at steady state $\langle G(\cdot) \rangle_f$ is limited by the collective outcome for the all-cooperator group. In particular, this means that if $G(x)$ is maximized by an interior level of cooperation $x^*$, then the average collective reproduction at steady state does not achieve its optimal possible level even in the limit of infinitely strong between-group competition when $\lambda \to \infty$. This generalizes the so-called ``shadow of lower-level selection" seen in previous work on multilevel selection in evolutionary games.

We may also visualize the impact of increasing the intensity of between-group competition by plotting the steady state densities from Equation \eqref{pthetactilde} for various values of $\lambda$. In Figure \ref{fig:densityplot}, we display these steady state densities for special cases of the game-theoretic examples from Section \ref{sec:gamemotivation} in which the collective payoff of the group is either maximized by full-cooperation ($x^* = 1$, Figure \ref{fig:a}) or by a composition of 75 percent cooperators ($x^* = \frac{3}{4}$, Figure \ref{fig:b}). In the former case, we see arbitrarily high levels of cooperation are achieved by the group for sufficiently large $\lambda$, whereas in the latter case the densities do not come close to achieving the optimal composition of cooperators. In fact, in the latter case, the density appears to concentrate around $\overline{x} = \frac{1}{2}$, the unique interior level of cooperation satisfying $G(x) = G(1)$. In the appendix, we formalize this intuition from Figure \ref{fig:densityplot} to show that the steady states of Equation \eqref{multiselect} in the PD case concentrate as $\lambda \to \infty$ upon measures supported only at points satisfying $G(x) = G(1)$.

\begin{figure}[tbhp]
\centering
\subfloat[$x^* = \ol{x} = 1$.]{\label{fig:a}\includegraphics[width = 0.5\textwidth]{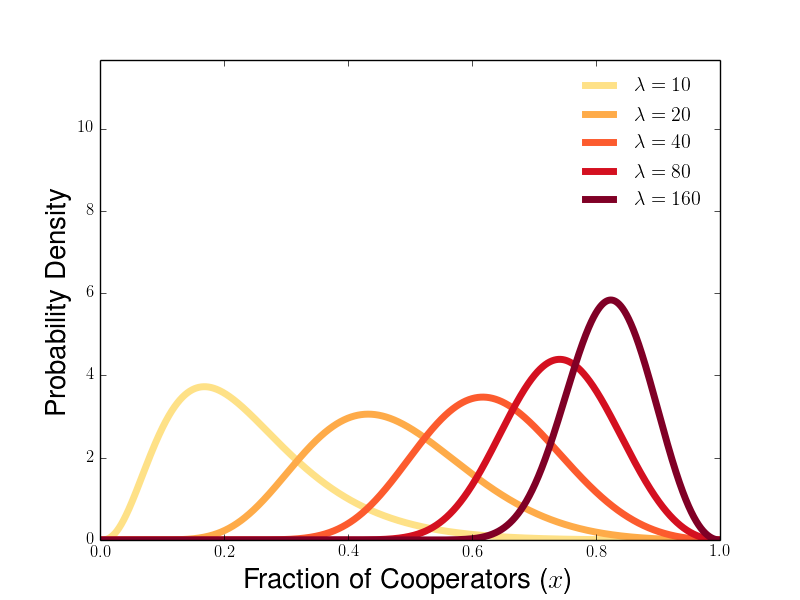}} %\hspace{-3mm}
\subfloat[$x^* = \frac{3}{4}$, $\overline{x} = \frac{1}{2}$.]{\label{fig:b}\includegraphics[width = 0.5\textwidth]{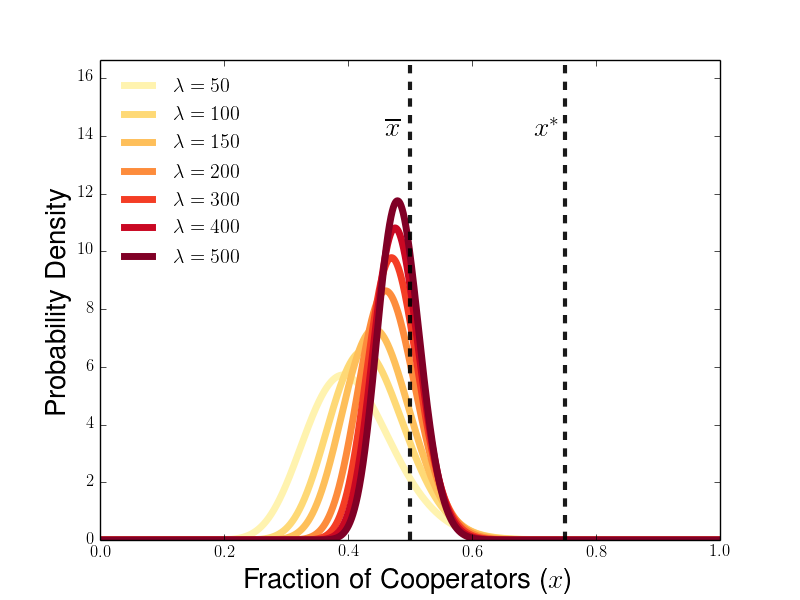}}
\caption{Steady state densities from Equation \eqref{pthetactilde} for various values of between-group selection intensity $\lambda$ for the game-theoretic case of Section \ref{sec:gamemotivation} and scenarios in which $G(x)$ is maximized by full-cooperation (left) or by 75 percent cooperators (right). The dashed lines in the bottom panel correspond to the group type $x^*$ maximizing $G(x)$ and the group type $\overline{x}$ for which $G(\overline{x}) = G(1)$. Densities have H{\"o}lder exponent $\theta = 3$ near $x=1$, and payoff parameters are chosen from the family of games with $S = 1$, $P = 2$, $T = R + 2$ with either $R = 2.5$ (left) or $R = 3$ (right).}
\label{fig:densityplot}
\end{figure}

\subsubsection{Convergence to Steady State Densities}
\label{sec:convergencesteadyresults}
 
Next, we explore the conditions under which steady states of the form $p^{\lambda}_{\theta}(x)$ are achieved as the long-time behavior under the dynamics of Equation \eqref{multiselect}. Considering initial populations that have a well-defined \holder exponent, we show in Theorem \ref{wlimtheorem} that solutions $\mu_t(dx)$ to Equation \eqref{multiselect} converge weakly to a steady state density when $\lambda > \lambda^*$% 
. This result confirms and generalizes \cite[Conjecture 1]{cooney2020analysis}, which addresses convergence to steady-state for multilevel selection in the case of replication rates arising from Prisoners' Dilemma games with the payoff matrix of Equation \eqref{eq:payoffmatrix}. 

\begin{theorem}\label{wlimtheorem}
Suppose that $G(x),\pi(x) \in C^1\left([0,1]\right)$, $G(1) > G(0)$, and $\pi(x) > 0$ for $x \in [0,1]$. Consider an initial measure $\mu_0(dx)$ having a \holder exponent $\theta > 0$ near $x=1$ with corresponding positive, finite \holder constant  $C_{\theta}$.
If $\lambda \left[G(1) - G(0) \right] > \theta \pi(1)$, then $\bar{\mu}_t$ converges weakly to the probability measure defined by the density function $p^{\lambda}_\theta(x)$ defined in Equation \eqref{pthetactilde}:
\begin{equation}\label{vmubarlim}
\lim_{t\to \infty} \int_0^1 v(x)\ov{\mu}_t(dx)=\int_0^1 v(x)p^{\lambda}_{\theta}(x)dx
\end{equation}
where $v(x)$ is an arbitrary continuous function on $[0,1]$.
\end{theorem}

The above, and indeed, many of our main results to follow depend on the careful study of the characteristic curves solving Equation \eqref{eq:replicatorcharacteristics} and the solutions along characteristics, which respectively describe the effects of within-group and between-group competition. Because $\pi(\cdot)$ is positive for the Prisoners' Dilemma case, $x=0$ and $x=1$ are the only steady states of this ODE and therefore the characteristic curves spend most of their time near $x=0$ or $x=1$. Thus, $G(0)$ and $G(1)$, the group replication rates near $x=0$ and $x=1$ respectively, and $\pi(1)$, the speed with which the ODE trajectory leaves $x=1$, can be expected to control growth of the unnormalized solution $\mu_t(dx)$ to Equation \eqref{eq:linearmeasuremultiselect}. This intuition is made precise by Lemma \ref{lem:integraltermbound}, where we use the continuous differentiability of $\pi(\cdot)$ and $G(\cdot)$ on $[0,1]$ to decompose $\mu_t(dx)$ into the product of a bounded, continuous function and an exponentially growing term whose growth rate is given by $\max\{\lambda G(0), \lambda G(1) - \theta \pi(1)\}$. This decomposition allow us to prove Theorem \ref{wlimtheorem} using our knowledge of $G(\cdot)$ and $\pi(\cdot)$ at the endpoints $0$ and $1$, without requiring explicit expressions for the characteristic curves and solutions alongs characteristics only available in special cases \cite{luo2017scaling,cooney2019replicator}. In addition to generalizing previous results, this approach allows us to glean further biological intuition into how the long-time behavior of Equation \eqref{multiselect} results from a tension between the collective incentive to achieve full-cooperation over full-defection $G(1) - G(0)$ and the individual-level incentive to defect $\pi(1)$ in a group with many cooperators.

\subsubsection{Long-Time Behavior for More General Initial Measures}
\label{sec:longtimeresults}

Not all initial measures have a well-defined \holder exponent as defined in Definition \ref{def:holderexponent}, as the limit characterizing the \holder exponent and constant in Equation \eqref{eq:holderexponent} does not necessarily exist. For such measures, the foregoing results do not apply, but we can still provide a characterization of the long-time behavior that holds for any initial measure and any relative strength of between-group selection. In Theorem \ref{thm:groupbounds}, we derive long-time upper and lower bounds for the time-averaged collective payoff of the population, showing that the long-time collective outcome is limited by the replication rate of the all-cooperator group. In Proposition \ref{thm:extinctvspersist}, we show that, for a given initial measure, there is a threshold strength of between-group competition required for cooperation to survive in the long-time population.

To supplement the H{\"o}lder exponent in characterizing initial measures by the behavior near $x=1$, we introduce the following quantities that can be defined for any initial measure. 

\begin{definition} \label{def:infimumHolder}
The infimum \holder exponent $\overline{\theta}$ near $x=1$ satisfies \begin{equation} \label{eq:infholderexponent} \overline{\theta} := \sup \left\{ \Theta \geq 0: \ds\liminf_{x \to 0}\frac{\mu_t([1-x,1])}{x^{\Theta}} = 0 \right\} \end{equation} 
Furthermore, the infimum \holder constant $C_{\overline{\theta}}$ is given by
\begin{equation}
    \ds\liminf_{x \to 0} \frac{\mu_t([1-x,1])}{x^{\overline{\theta}}} = C_{\overline{\theta}}.
\end{equation}
\end{definition}
\begin{definition} \label{def:supremumHolder}
The supremum \holder exponent $\underline{\theta}$ near $x=1$ satisfies \begin{equation} \label{eq:supholderexponent} \underline{\theta} := \sup \left\{ \Theta \geq 0: \ds\limsup_{x \to 0}\frac{\mu_t([1-x,1])}{x^{\Theta}} = 0 \right\} \end{equation} 
Furthermore, the supremum \holder constant $C_{\underline{\theta}}$ is given by
\begin{equation}
    \ds\limsup_{x \to 0} \frac{\mu_t([1-x,1])}{x^{\underline{\theta}}} = C_{\underline{\theta}}.
\end{equation}
\end{definition}
\begin{remark}
The infimum and supremum \holder exponents satisfy the inequality $\underline{\theta} \leq \overline{\theta}$. This is true because, for any $\Theta \geq 0$, \begin{equation} \label{eq:liminfleqlimsup} \liminf_{x \to 0} \frac{\mu_t\left(\left[1-x,1\right] \right)}{x^{\Theta}} \leq \limsup_{x \to 0} \frac{\mu_t\left(\left[1-x,1\right] \right)}{x^{\Theta}}.\end{equation} Therefore if the left-hand side is positive for a given $\Theta$, then right-hand side is positive as well. 

Furthermore, if there are $\theta$ and $C_{\theta}$ such that the infimum and supremum H{\"o}lder exponents satisfy $\overline{\theta} = \underline{\theta} = \theta$ and infimum and supremum H{\"o}lder constants satisfy $C_{\overline{\theta}} = C_{\underline{\theta}} = C_{\theta}$, then Definitions \ref{def:infimumHolder} and \ref{def:supremumHolder} imply that the limiting behavior of Equation \eqref{eq:holderexponent} is satisfied by the measure $\mu_t(dx)$. In other words, if the infimum and supremum H{\"o}lder data agree for a measure $\mu_t(dx)$ near $x=1$, then $\mu_t(dx)$ has a well-defined H{\"o}lder exponent $\theta$ and H{\"o}lder constant $C_{\theta}$ near $x=1$ in the sense of Definition \ref{def:holderexponent}. 
\end{remark}

\begin{remark}
It can be shown that, if our initial measure $\mu_0(dx)$ has infimum and supremum H{\"o}lder exponents $\overline{\theta}$ and $\underline{\theta}$ near $x=1$ with constants $C_{\overline{\theta}}$ and $C_{\underline{\theta}}$, then the solution $\mu_t(dx)$ to Equation \eqref{eq:linearmeasuremultiselect} has the same infimum and supremum \holder exponents $\overline{\theta}$ and $\underline{\theta}$ near $x=1$ with constants $C_{\overline{\theta}} e^{\left[\lambda G(1) - \theta \pi(1) \right]t}$ and $C_{\underline{\theta}} e^{\left[\lambda G(1) - \theta \pi(1) \right]t}$. This allows us to see that the set of measures with well-defined \holder exponent and \holder constant near $x=1$ is closed under our multilevel dynamics.
\end{remark}

We can now use this characterization of the initial measure $\mu_0(dx)$ in terms of its infimum and supremum H{\"o}lder exponents $\overline{\theta}$ and $\underline{\theta}$ near $x=1$ to study the long-time behavior of measure-valued solutions $\ol{\mu}_t(dx)$ to Equation \eqref{multiselect}. To understand the collective outcome achieved by the population, we consider the average group-level replication rate across all the groups in the population, which we denote by
\begin{equation} \label{eq:mutaverage}
    \langle G(\cdot) \rangle_{ \ol{\mu}_t} := \int_0^1 G(x) \ol{\mu}_t(dx).
\end{equation}
In Theorem \ref{thm:groupbounds}, we show that the time-average of this collective reproduction rate $\langle G(\cdot) \rangle_{\ol{\mu}_s}$ eventually satisfies bounds in terms of the supremum and infimum \holder exponents of the initial measure $\mu_0(dx)$.

\begin{theorem} \label{thm:groupbounds}
Suppose that $G(x), \pi(x)$ satisfy the assumptions of Theorem \ref{wlimtheorem} and that the initial distribution $\overline{\mu}_0(dx)$ has positive supremum and infimum \holder exponents $\underline{\theta}$ and $\overline{\theta}$ near $x=1$. Under these assumptions, the average group-level reproduction rate $\langle G(\cdot) \rangle_{\ol{\mu}_t}$ of solutions $\ol{\mu}_t(dx)$ to Equation \eqref{multiselect} satisfy
\begin{subequations} \label{eq:groupboundslims}
\begin{align}
    \ds\limsup_{t \to \infty} \frac{1}{t} \int_0^t \langle G(\cdot) \rangle_{\ol{\mu}_s} ds &= \max\left\{  G(1) - \frac{\underline{\theta} \pi(1)}{\lambda}, G(0) \right\} \label{eq:groupupperbound} \\ 
     \ds\liminf_{t \to \infty} \frac{1}{t} \int_0^t \langle G(\cdot) \rangle_{\ol{\mu}_s} ds &= \max\left\{G(1) - \frac{\overline{\theta} \pi(1)}{\lambda}, G(0) \right\}  \label{eq:grouplowerbound}.
\end{align}
\end{subequations}
\end{theorem}

The bounds from Theorem \ref{thm:groupbounds} tell us that, in a time-averaged sense, the long-time collective outcome is limited by the group-reproduction rate $G(1)$ of the full-cooperator group. This extends our idea of the shadow of lower-level selection seen in Equation \eqref{eq:averageofthetawithlambdastar}, showing that this limitation on the collective reproduction rate holds for any initial measure and making an explicit connection between this limitation and the dynamics of Equation \eqref{multiselect}. 

The proof of Theorem \ref{thm:groupbounds} relies on the following implicit representation for the mass of the unnormalized solution $\mu_t(dx)$ of Equation \eqref{eq:linearmeasuremultiselect}, which is derived in Section \ref{sec:exponentialnormalization}:
\begin{equation} \label{eq:expnormintro}
    \mu_t\left([0,1] \right) = \int_0^1 \mu_t(dx) = \exp\left( \int_0^t \langle G(\cdot) \rangle_{\ol{\mu}_s} ds \right).
\end{equation}
This expression can be combined with our estimates for the principal growth rates for $\mu_t([0,1])$ derived in Lemmas \ref{lem:probholderbounds}, \ref{lem:integraltermbound}, and \ref{lem:massinfimumbound} to deduce the bounds of Equation \eqref{eq:groupboundslims}. This relationship between the principal growth rates and long-time average group-level replication rate for the population illustrates the importance of collective group-level in the structure of the multilevel replicator dynamics described by Equation \eqref{multiselect}. Furthermore, the form of the bounds we obtain on the collective outcome highlights the key roles played by all-cooperator group in maintaining cooperation in the population, and how the tension between individual and group incentives hinges upon the interaction between reproduction rates $G(0)$, $G(1)$, and $\pi(1)$ and the infimum and supremum H{\"o}lder exponents $\overline{\theta}$ and $\underline{\theta}$ of the initial measure near $x=1$.

\sloppy{Examining the bounds of Equation \eqref{eq:groupboundslims}, we see that $\limsup_{t \to \infty} \frac{1}{t} \int_0^t \langle G(\cdot) \rangle_{\ol{\mu}_s} ds = G(1) - \lambda^{-1} \underline{\theta} \pi(1) > G(0)$ when $\lambda \left[ G(1) - G(0) \right] > \underline{\theta} \pi(1)$. This allows us to identify the following threshold value $\lambda^*(\underline{\theta})$ of the between-group selection strength, analogous to that of Equation \eqref{eq:lambdastar}, which is given by}
\begin{equation} \label{lambdastarsup}
    \lambda^*(\underline{\theta}) := \frac{\underline{\theta} \pi(1)}{G(1) - G(0)},
\end{equation}
such that the time-averaged value of $\langle G(\cdot) \rangle_{\ol{\mu}_t}$ exceeds the all-defector reproduction rate $G(0)$ infinitely often if $\lambda > \lambda^*(\underline{\theta})$. If $\lambda > \lambda^*(\underline{\theta})$ and $\underline{\theta} \ne \overline{\theta}$, we see from Equation \eqref{eq:groupboundslims} that $\liminf_{t \to \infty} \frac{1}{t} \int_0^t \langle G(\cdot) \rangle_{\ol{\mu}_s} ds  \neq \limsup_{t \to \infty} \frac{1}{t} \int_0^t \langle G(\cdot) \rangle_{\ol{\mu}_s} ds$, and therefore the time-average $\langle G(\cdot) \rangle_{\ol{\mu}_t}$ does not converge as $t \to \infty$. Because we assume that $G(x) \in C^1([0,1])$, $G(x)$ is a valid test-function $v(x)$ for our measure-valued formulation of the multilevel dynamics as described by Equation \eqref{eq:measuremultiselect}, and thus $\ol{\mu}_t(dx)$ cannot converge to a density steady state when there is disagreement between the supremum and infimum H{\"o}lder exponents $\overline{\theta}$ and $\underline{\theta}$ near $x=1$ for the initial measure $\mu_0(dx)$.

Even though solutions $\ol{\mu}_t(dx)$ do not necessarily converge to any steady state in the long-time limit, we can use the threshold $\lambda^*(\underline{\theta})$ to characterize whether cooperation will survive in the long-run population given any initial measure $\mu_0(dx)$. In Theorem \ref{thm:extinctvspersist}, we show that cooperation vanishes from the population when $\lambda < \lambda^*(\underline{\theta})$, while cooperation survives when $\lambda > \lambda^*(\underline{\theta})$. Mathematically, this consists of showing that the population converges to a delta-measure $\delta(x)$ concentrated upon the all-defector group when $\lambda < \lambda^*(\underline{\theta})$, while the fraction of cooperators in the positive exceeds a positive threshold infinitely often when $\lambda > \lambda^*(\underline{\theta})$. This sense in which cooperation survives is called weak persistence \cite{freedman1990persistence}, and has often been used to characterize the survival and coexistence of strategies in evolutionary games and related ecological models under individual-level dynamics \cite{hofbauer1981general,hofbauer1998evolutionary,hofbauer2004persist,bratus2017solutions,freedman1990persistence}. For the edge case in which $\lambda = \lambda^*(\underline{\theta})$, we know from Theorem \ref{thm:groupbounds} that the time-averaged collective outcome $\langle G(\cdot) \rangle_{\ol{\mu}_t}$ converges to that of the all-defector group $G(0)$, and show in Section \ref{sec:convergencedeltazero} that the population will converge to $\delta(x)$ for a more restricted class of group-reproduction functions $G(x)$ and initial measures.

\begin{theorem} \label{thm:extinctvspersist}
Suppose that $G(x), \pi(x)$ satisfy the assumptions of Theorem \ref{thm:groupbounds}, and that the initial distribution $\overline{\mu}_0(dx)$ has positive supremum and infimum \holder exponents $\underline{\theta}$ and $\overline{\theta}$ near $x=1$. If $\lambda \left[G(1) - G(0) \right] < \underline{\theta} \pi(1)$, then $\ol{\mu}_t(dx) \rightharpoonup \delta(x)$ as $t \to \infty$. If $\lambda \left[G(1) - G(0) \right] > \underline{\theta} \pi(1)$, then the average fraction of cooperators satisfies
\begin{equation}
    \ds\limsup_{t \to \infty} \int_0^1 x \overline{\mu}_t(dx) > 0.
\end{equation}
\end{theorem} 

Our result from Theorem \ref{thm:extinctvspersist} on extinction and weak persistence of cooperation holds for any initial measure $\mu_0(dx)$. Consequently, we can view weak persistence as serving as a more general criterion for identifying the survival of cooperation than convergence to a steady state density in the sense of Theorem \ref{wlimtheorem}. In this light, we can see Theorem \ref{thm:extinctvspersist} as providing a natural classification for the long-term behavior of solutions to Equation \eqref{multiselect} for replication rates satisfying the assumptions of the multilevel PD dynamics.

\myindent To fully characterize the dynamics of our multilevel replicator equations for within-group replication rates that feature no internal equilibria, we now consider the long-time behavior for the case of the generalized Prisoners' Delight game (in which $G(1) > G(0)$ and $\pi(x) < 0$.) Because both within-groups and between-group competition push to increase cooperation under the PDel dynamics, we show in Proposition \ref{prop:PDeldelta} that the population concentrates upon the full-cooperator group when there is any between-group competition and any cooperators in the initial population.

\begin{proposition} \label{prop:PDeldelta}
Suppose that $G(x), \pi(x) \in C^1\left[0,1 \right]$, $G(1) > G(0)$, and $\pi(x) < 0$ for $x \in [0,1]$. If $\mu_0\left(\left(0,1\right]\right) > 0$ and $\lambda > 0$, then $\ol{\mu}_t(dx) \rightharpoonup \delta(x-1)$ as $t \to \infty$.
\end{proposition}

\subsubsection{Results for Multiple Population Dynamics}
\label{sec:multiplepopulationsresults}

\myindent  As we saw for the dynamics of a single interval, the growth of the non-normalized solution can be associated with a principal exponential growth rate that can be associated with the average payoff at steady state. For a given interval $j$ with well-defined \holder exponent $\theta^j$ near $x=1$, 
 
  \begin{equation} \label{eq:rjwelldefined}
   r_j  = \left\{
     \begin{array}{cl}
       \lambda G_j(1) & : \pi_j(x) < 0\\
       \lambda G_j(1) - \theta^j \pi_j(1) & : \pi_j(x) > 0\: , \: \lambda \left[ G_j(1) - G_j(0) \right] - \theta^j \pi_j(1) > 0 \\ 
       \lambda G_j(0) & : \pi_j(x) > 0 \: , \: \lambda \left[ G_j(1) - G_j(0) \right] - \theta^j \pi_j(1) < 0 \\ 
     \end{array}
   \right. .
\end{equation}
However, in the case in which $\mu^j_0(dx)$ does not have a well-defined \holder exponent near $x=1$, we can possibly only bound the principal growth rate in terms of the infimum and supremum \holder exponents $\overline{\theta}^j$ and $\underline{\theta}^j$ near $x=1$. Recalling that $\underline{\theta}^j \leq \overline{\theta}_j$, we can see that the principal growth rate $r_j \in [r_j^m,r_j^M]$, where we can our lower bound in terms of $\underline{\theta}^j$ as 

\begin{subequations}
\begin{equation}
r_j^m  = \left\{
     \begin{array}{cl}
       \lambda G_j(1) & : \pi_j(x) < 0\\
       \lambda G_j(1) - \overline{\theta}^j \pi_j(1) & : \pi_j(x) > 0\: , \: \lambda \left[ G_j(1) - G_j(0) \right] - \overline{\theta}^j \pi_j(1) > 0 \\ 
       \lambda G_j(0) & : \pi_j(x) > 0 \: , \: \lambda \left[ G_j(1) - G_j(0) \right] - \overline{\theta}^j \pi_j(1) < 0 \\ 
     \end{array}    \right. 
\end{equation}
and our upper bound in terms of $\underline{\theta}^j$ as
\begin{equation}
 r_j^M  = \left\{
     \begin{array}{cl}
       \lambda G_j(1) & : \pi_j(x) < 0\\
       \lambda G_j(1) - \underline{\theta}^j \pi_j(1) & : \pi_j(x) > 0\: , \: \lambda \left[ G_j(1) - G_j(0) \right] - \underline{\theta}^j \pi_j(1) > 0 \\ 
       \lambda G_j(0) & : \pi_j(x) > 0 \: , \: \lambda \left[ G_j(1) - G_j(0) \right] - \underline{\theta}^j \pi_j(1) < 0 \\ 
     \end{array}    \right. .
\end{equation}
\end{subequations}
In particular, this means that the possible growth rates for two intervals can potentially overlap. In the remainder of the paper, we will focus on the case in which there is an interval $j$ such that $r_j^m > r_k^M$ for all other intervals $k$. 
When such a condition holds, %
we can show that the whole population will eventually concentrate upon the subpopulation with the dominant principal growth rate. 
\begin{theorem} \label{thm:longtimemultiple}
Suppose that each supopulation $j \in \mathcal{N}$ has reproduction functions $G_j(x), \pi_j(x) \in C^1\left([0,1]\right)$ satisfying $G_j(1) > G_j(0)$ and either $\pi_j(0) > 0$ for $x \in [0,1]$ or $\pi_j(0) < 0$ for $x\ in [0,1]$ and that its initial measure $\mu_0^j(dx)$ have infimum and supremum \holder exponents $\overline{\theta}^j > 0$ and $\underline{\theta}^j > 0$ near $x = 1$.  Suppose there is a subpopulation $k$ such that $r^m_k > r^M_j$ for $j \in \mathcal{N} - \{k\}$. Then, for all such $j$, we see that $\overline{\mu}_t^j\left([0,1]\right) \to 0$ as $t \to \infty$. Furthermore, if $\mu_0^k(dx)$ has well-defined \holder exponent $\theta^k$ and \holder constant $C_{\theta^k}$ near $x=1$ that are positive and finite, then $\ol{\mu}_t^k \rightharpoonup p^{\lambda,k}_{\theta} dx$ as $t \to \infty$ if $\lambda \left[G_k(1) - G_k(0)\right] > \theta^k \pi_k(1)$, where $p^{\lambda}_{\theta,k}(x)$ is the steady state density given by Equation \eqref{pthetactilde} with $G(x) := G_k(x)$ and $\pi(x) := \pi_k(x)$. 
\end{theorem}

\subsection{Outline of Paper}
\label{sec:outline}
 The remainder of the main paper is structured as follows. In Section \ref{sec:dynamicalproperties}, we calculate the steady state densities for Equation \eqref{multiselect} in the PD case and characterize the average collective payoff at steady state. In Section \ref{sec:dynamicalproperties}, we demonstrate useful properties of time-dependent solutions to Equation \eqref{multiselect} and present the main estimates for the principal growth rates for solutions along characteristics. In Section \ref{sec:convergencesteady}, we use the estimates from Section \ref{sec:dynamicalproperties} to prove Theorem \ref{wlimtheorem}, showing convergence of the population to steady states supporting cooperation when the initial measure has a well-defined H{\"o}lder exponent near full-cooperation and between-group competition is sufficiently strong. In Section \ref{sec:extinctionsurvival}, we derive long-time bounds on the time-average of the population's collective reproduction rate (Theorem \ref{thm:groupbounds}), and  we characterize how cooperation can either collapse or persist in the population (Theorem \ref{thm:extinctvspersist}), depending on the supremum H{\"o}lder exponent of the initial population and the relative intensity of within-group and between-group competition. In Section \ref{sec:multiplepopulations}, we characterize the long-time behavior of the multipopulation dynamics described by Equation \eqref{eq:measurevaluedsubpopulation}, showing that multilevel selection can promote concentration upon the group type favoring the highest long-time collective payoff (Theorem \ref{thm:groupbounds}). In Section \ref{sec:discussion}, we discuss our results and directions for future research. 

 We also present additional results in the appendix. In Section \ref{sec:wellposedness}, we demonstrate well-posedness for solutions to Equation \eqref{multiselect} in the measure-valued sense required for our results on long-time behavior. In Section \ref{sec:steadystates}, we discuss additional properties of the density steady state solutions for the PD multilevel dynamics, characterizing the regularity of the densities and showing the concentration of the steady states upon group compositions achieving the same payoff as an all-cooperator group in the limit of infinite between-group competition. Finally, in Section \ref{sec:HDSH}, we formulate a generalized version of the multilevel HD and SH dynamics in terms of the multipopulation framework discussed in Section \ref{sec:multiplepopulations}
 and then show how we can apply Theorem \ref{thm:longtimemultiple} to understand the long-time behavior of solutions to Equation \eqref{multiselect} for the case of HD or SH games.

\section{Steady State Solutions of Multilevel Dynamics} \label{sec:steadystates}

In this section, we derive the steady state densities presented in Equation \eqref{pthetactilde} and \eqref{eq:lambdaCofx} for the PD case of the multilevel dynamics. We show that this family of steady states can be parameterized by their H{\"o}lder exponent $\theta$ near the full-cooperator equilibrium at $x=1$, and we use this parametrization to calculate the average group-level reproduction rate of the steady state population. We also characterize additional properties for these steady state densities in Section \ref{sec:densitysteady}.

From the results of Lemmas \ref{lem:weaksteadystate} and \ref{prop:steadydelta}, we know that the only possible long-time steady states of Equation \eqref{multiselect} in the PD case are delta-measures $\delta(x)$ and $\delta(x-1)$, concentrated at the all-defector and all-cooperator equilibria, and a family densities $f(x)$ which satisfy the ordinary differential equation
\begin{equation} \label{eq:timeindependentODE}
    0 = \dsdel{}{x} \left[ x (1-x) \pi(x) f(x) \right] + \lambda f(x) \left[ G(x) - \int_0^1 G(y) f(y) dy \right],
\end{equation}
where $f(x)$ is continuously differentiable and $f(x) > 0$ for $x \in (0,1)$. For such strong solutions to the steady state problem, we can use separation of variables to see that $f(x)$ must satisfy

\begin{equation} \label{eq:steadystateODE}
\frac{f'(x)}{f(x)} = - \frac{\left[x(1-x) \pi(x) \right]'}{x(1-x) \pi(x)} + \lambda \frac{\langle G(\cdot) \rangle_f - G(x) }{x (1-x) \pi(x)}.   
\end{equation}
We can rewrite the last term of Equation \ref{eq:steadystateODE} as 
\begin{dmath} \label{eq:steadyrighthand}
\frac{\langle G(\cdot) \rangle_f - G(x) }{x (1-x) \pi(x)}  = \frac{\lambda}{\pi(0)} \left( \frac{\langle G(\cdot) \rangle_f - G(0)}{x} \right) +  \frac{\lambda}{\pi(1)} \left( \frac{\langle G(\cdot) \rangle_f - G(1)}{1-x} \right) + \frac{\lambda C(x)}{\pi(x)}
\end{dmath}
where, using the shorthand notation $\tilde{G}(x) := G(x) - \langle G(\cdot) \rangle_f$, we can write $C(x)$ as
\begin{dmath} \label{eq:Cofxformula}
C(x)
= \left( \frac{1}{x(1-x)} \right) \left[ -\tilde{G}(x) +  \frac{\tilde{G}(0)}{\pi(0)} (1-x) \pi(x) +  \frac{\tilde{G}(1)}{\pi(1)} x \pi(x)   \right] 
= \frac{  \tilde{G}(0) - \tilde{G}(x)}{x} + \left(\frac{\tilde{G}(0)}{\pi(0)} \right)\left(\frac{\pi(x) - \pi(0) }{x}\right) + \frac{\tilde{G}(1) - \tilde{G}(x)}{1-x} +   \left(\frac{\tilde{G}(1)}{\pi(1)} \right)\left(\frac{\pi(x) - \pi(1) }{1-x}\right).  
\end{dmath}
Using Equation \eqref{eq:steadyrighthand},
we can see that steady state densities must satisfy the implicit expression
\begin{equation} \label{eq:steadyimplicit}
 f(x) = \frac{1}{Z_f} x^{\frac{\lambda}{\pi(0)}\left[\langle G(\cdot) \rangle_f - G(0) \right] - 1} \left( 1-x\right)^{\frac{\lambda}{\pi(1)}\left[\langle G(\cdot) \rangle_f - G(1) \right] - 1} \pi(x)^{-1} \exp\left( -\lambda \int_x^1 \frac{C(s)}{\pi(s)} ds \right).
\end{equation}
Because $G(x)$ and $\pi(x)$ are $C^1$ functions, we see from Equation \eqref{eq:Cofxformula} that $C(x)$ remains bounded near 0 and 1. This means that the density $f(x)$ from Equation \eqref{eq:steadyimplicit} will be integrable on $(0,1)$ if the average payoff $\langle G(\cdot) \rangle_f$ satisfies the following bounds
\begin{equation} \label{eq:averagepayoffbound}
 G(0) < \langle G(\cdot) \rangle_f < G(1).   
\end{equation}
In particular, this tells us that valid steady state densities cannot have a higher group-reproduction rate than the rate of a full-cooperator group $G(1)$, providing a signature of the shadow of lower-level selection. Furthermore, the implicit form of the density provided in Equation \eqref{eq:steadyimplicit} highlights the principal contributions of the group-level replication rates of all-cooperator and all-defector groups $G(1)$ and $G(0)$ in determining whether a given distribution of group compositions can be maintained at steady state.

From the implicit relation of Equation \eqref{eq:steadyimplicit}, we see that there are infinitely many possible steady state densities $f(x)$ for a given relative selection strength $\lambda$, one for each % 
value of $\langle G(\cdot) \rangle_f$ satisfying the bounds from Equation \eqref{eq:averagepayoffbound}.  Because the \holder exponent near $x=1$ is preserved under the dynamics of Equation \eqref{multiselect}, we will parametrize the measures corresponding to the densities from Equation \eqref{eq:steadyimplicit} by their \holder exponents $\theta$ near $x=1$ to obtain an explicit representation for our family of density steady states. 

Noting that $C(x)$ is bounded on $[0,1]$, we can compute that
\begin{dmath*}
 \ds\lim_{x \to 0} \frac{\int_{1-x}^1 f(y) dy}{x^{\Theta}} = \ds\lim_{x \to 0}  \frac{f(1-x)}{\Theta x^{\Theta - 1}} =  \frac{1}{\Theta Z_f} \ds\lim_{x \to 0} \left[  x^{\frac{\lambda}{\pi(1)} \left[\langle G(\cdot) \rangle_f - G(1) \right] - \Theta } \left( 1 - x\right)^{\frac{\lambda}{\pi(0)} \left[\langle G(\cdot) \rangle_f - G(0) \right] - 1} % 
 \pi(x)^{-1} e^{- \lambda \int_{1-x}^1 \frac{C(s)}{\pi(s)} ds} \right] 
 = \left\{
     \begin{array}{lr}
       0 & : \Theta < \left( \frac{\lambda}{\pi(1)} \right) \left[\langle G(\cdot) \rangle_f - G(1) \right] \\
       (\Theta \pi(1) Z_f)^{-1} & : \Theta =  \left( \frac{\lambda}{\pi(1)} \right) \left[\langle G(\cdot) \rangle_f - G(1)\right] \\
       \infty & : \Theta > \left( \frac{\lambda}{\pi(1)} \right) \left[\langle G(\cdot) \rangle_f - G(1)\right]
     \end{array}
   \right. .
\end{dmath*}
Therefore we can deduce from Definition \ref{def:holderexponent} that the \holder exponent that the Hölder exponent near $x=1$ for our steady state densities $f(x)$ is given by
\begin{equation} \label{eq:thetaofaverage}
\theta = \frac{\lambda}{\pi(1)} \left[\langle G(\cdot) \rangle_f - G(1) \right].
\end{equation}
We can then use this expression to obtain the explicit family of steady states $f^{\lambda}_{\theta}(x)$ of Equation \eqref{pthetactilde}. Furthermore, the average of the group-reproduction function $G(x)$ on such steady states is given by
\begin{equation} \label{eq:averageoftheta}
    \langle G(\cdot) \rangle_f = G(1) - \frac{\pi(1) \theta}{\lambda}.
\end{equation}
 Using the expression from Equation \eqref{eq:lambdastar} for the threshold $\lambda^*$ required for integrability of the density $f^{\lambda}_{\theta}(x)$, we can deduce that
\begin{equation} \label{eq:averageofthetawithlambdastaragain}
    \langle G(\cdot) \rangle_f = \left(\frac{\lambda^*}{\lambda}\right) G(0) + \left( 1 - \frac{\lambda^*}{\lambda} \right) G(1).
\end{equation}
This provides an improvement upon the bounds from Equation \eqref{eq:averagepayoffbound}, showing that $ \langle G(\cdot) \rangle_f$ interpolates between $G(0)$ when $\lambda = \lambda^*$ and $G(1)$ as $\lambda \to \infty$.

\section{Useful Properties of Multilevel Dynamics}
\label{sec:dynamicalproperties}

In this section, we provide some useful properties for the measure-valued solutions to Equation \ref{multiselect} and the behavior of the dynamical properties of our model of multilevel selection. In section \ref{sec:methodofcharacteristics}, we use the method of characteristics to obtain a representation formula for the time-dependent solutions of the multilevel dynamics. In Section \ref{sec:growthrates}, we use this representation formula and assumptions about the supremum and infimum H{\"o}lder exponents of the initial distribution near the full-cooperator equilibrium to derive upper and lower bounds for the principal growth rates of solutions to the linear form of the multilevel dynamics given by Equation \eqref{eq:linearmeasuremultiselect}.

\subsection{Representing Time-Dependent Solutions of Multilevel Dynamics}
\label{sec:methodofcharacteristics}

First, we characterize the impacts of within-group and between-group competition on solutions $\mu_t(dx)$ through properties of the characteristic curves  and the solutions along characteristics.
To do this, we first consider solutions to the the linear problem of Equation \eqref{eq:linearmeasuremultiselect}. We consider the following ordinary differential equation
\begin{equation}\label{cODE}
\begin{split}
\D{x}{t}&=-x(1-x)\pi(x), \; \D{q}{t}=\lambda G(x) q, \\
x(0)&=y, \; q(0)=1.
\end{split}
\end{equation}
whose solution we denote by
\begin{equation}\label{phipsi}
x(t)=\phi_t (y), \; q(t)=\psi_t(y).
\end{equation}
We can then represent the solution $\mu_t(dx)$ to the linear multilevel dynamics of Equation \eqref{eq:linearmeasuremultiselect} by pushing forward the initial measure along characteristic curves, allowing us to obtain
\begin{equation} \label{eq:pushforward}
\int_0^1 v(x)\mu_t(dx)=\int_0^1 v(\phi_t(y))\psi_t(y) \mu_0(dx), \; w_t(x)=\psi_t(\phi_t^{-1} (x))
\end{equation}
for all continuous test functions $v(x)$. To further understand how the measure $\mu_t(dx)$ evolves in time, we now study in Lemmas \ref{lem:epi1tformula}, \ref{lem:wtformula}, and \ref{lem:phitinvformula} the behavior of the backwards characteristic curves $\phi_t^{-1}$ (describing within-group competition) and the solutions along characteristics $w_t(x)$ (describing between-group competition).

First we obtain an expression for the backward characteristic curves $\phi_t^{-1}(x)$.
\begin{lemma} \label{lem:epi1tformula}
Suppose that $G(x), \pi(x)$ satisfy the assumptions of Theorem \ref{wlimtheorem} and let $\phi_t$ be as in Equation \eqref{phipsi}. For $0<x\leq 1$, we have that
\begin{equation}\label{zest}
\begin{split}
\exp(\pi(1)t)(1-\phi_t^{-1}(x))&=(1-x)\exp\paren{\int_x^{\phi_t^{-1}(x)} \frac{Q(s)ds}{s\pi(s)}},\\
Q(s)&=\frac{\pi(1)-\pi(s)}{1-s}+\pi(s).
\end{split}
\end{equation}
In particular, 
\begin{equation}\label{zestlim}
\lim_{t\to \infty}\exp(\pi(1)t)(1-\phi_t^{-1}(x))=(1-x)\exp\paren{\int_x^1 \frac{Q(s)ds}{s\pi(s)}}.
\end{equation}
\end{lemma}
\begin{proof}
Let $z(t)=1-\phi_t^{-1}(x)$. Given $q(t)=\phi_t(x_0)$ satisfies Equation \eqref{cODE},  
$z(t)$ satisfies the following differential equation:
\begin{equation}
\D{z}{t}=-z(1-z)\pi(1-z), \; z(0)=z_0=1-x.
\end{equation}
We thus have:
\begin{equation}
\int_{z_0}^z \frac{ds}{s(1-s)\pi(1-s)}= -t.
\end{equation}
We may rewrite the left hand side as:
\begin{equation}
\int_{z_0}^z \frac{ds}{s(1-s)\pi(1-s)}=\int_{z_0}^z \frac{ds}{\pi(1)s}
+\int_{z_0}^z \frac{((\pi(1)-\pi(1-s))/s +\pi(1-s))ds}{\pi(1)(1-s)\pi(1-s)}.
\end{equation}
We thus have:
\begin{equation}
\ln\paren{\frac{z}{z_0}}+\int_{1-z_0}^{1-z} \frac{Q(s)ds}{s\pi(s)}=-\pi(1)t
\end{equation}
Exponentiating both sides, we obtain Equation \eqref{zest}. We obtain Equation \eqref{zestlim} by noting that $\phi_t^{-1}(x)\to 1$ as $t\to \infty$ 
and that $Q(s)$ is a bounded function for $0<x\leq s\leq 1$ since $\pi(s)$ is a $C^1$ function by assumption. 
\end{proof}
We next study $w_t(x)$ to describe the effect of between-group competition. 
\begin{lemma} \label{lem:wtformula}
Suppose that $G(x), \pi(x)$ satisfy the assumptions of Theorem \ref{wlimtheorem} and consider $w_t(x)$ given in Equation \eqref{eq:pushforward}. For $0<x\leq 1$, we have that
\begin{equation}\label{wtest}
\begin{split}
\exp(-\lambda G(1)t)w_t(x)&=\exp\paren{-\lambda \int_x^{\phi_t^{-1}(x)}\frac{R(s)}{s\pi(s)}ds}\\
R(s)&=\frac{G(1)-G(s)}{1-s}.
\end{split}
\end{equation}
In particular, we have:
\begin{equation}\label{wtestlim}
\lim_{t\to \infty} \exp(-\lambda G(1) t)w_t(x)=\exp\paren{-\lambda \int_x^1\frac{R(s)}{s\pi(s)}ds}.
\end{equation}
\end{lemma}
\begin{proof}
Let $q(t)=w_t(x)$. It is readily seen from \eqref{cODE} and \eqref{phipsi} that $q(t)$ satisfies the differential equation:
\begin{equation}\label{pqeqn}
\begin{split}
\D{p}{t}&=p(1-p)\pi(p), \; \D{q}{t}=\lambda G(p) q,\\
p(0)&=p_0=x, \; q(0)=1.
\end{split}
\end{equation}
Note here that $p(t)=\phi_t^{-1}(x)$. Let $r(t)=q(t)\exp(-\lambda G(1)t)$. The function $r(t)$ satisfies the equation:
\begin{equation}
\D{r}{t}=\lambda(G(p)-G(1))r.
\end{equation}
Thus, using \eqref{pqeqn}, we have:
\begin{equation}
\D{r}{p}=\frac{\lambda (G(p)-G(1))}{p(1-p)\pi(p)}r.
\end{equation}
Using $r(0)=q(0)=1$, we have:
\begin{equation}
r(t)=\exp\paren{-\lambda \int_{p_0}^{p(t)}\frac{R(s)}{s\pi(s)}ds}.
\end{equation}
We obtain Equation \eqref{zestlim} by noting that $\phi_t^{-1}(x)\to 1$ as $t\to \infty$ 
and that $G(s)$ is $C^1$ and thus the integral in Equation \eqref{zestlim} is bounded as the upper bound of the integral tends to $1$.
\end{proof}

Next, we obtain an expression for $\partial\phi_t^{-1}(x)/\partial x$.

\begin{lemma} \label{lem:phitinvformula}
Suppose that $G(x), \pi(x)$ satisfy the assumptions of Theorem \ref{wlimtheorem} and consider $\phi_t^{-1}(x)$ as defined in Equation \eqref{phipsi}. We have that
\begin{equation}\label{Jacobian}
\PD{}{x}\phi_t^{-1}(x)=\frac{(\phi_t^{-1}(x)(1-\phi_t^{-1}(x))\pi(\phi_t^{-1}(x))}{x(1-x)\pi(x)}
\end{equation}
\end{lemma}
\begin{proof}
Let $J(t)=\partial \phi_t^{-1}(x)/\partial x$. Then, using \eqref{cODE}, we see that $J$ must satisfy the following
differential equation:
\begin{equation}
\begin{split}
\D{p}{t}&=p(1-p)\pi(p), \; \D{J}{t}=\paren{(1-p)\pi(p)-p\pi(p)+p(1-p)\pi'(p)}J,\\
p(0)&=p_0=x, \; J(0)=1,
\end{split}
\end{equation}
where $\pi'$ is the derivative of $\pi$, which exists given our assumption that $\pi$ is $C^1$.
From the above, we see that:
\begin{equation}
\D{J}{p}=\paren{\frac{1}{p}-\frac{1}{1-p}+\frac{\pi'(p)}{\pi(p)}}J.
\end{equation}
Integrating the above differential equation, we obtain the desired formula:
\begin{equation}
J(t)=\frac{p(t)(1-p(t))\pi(p(t))}{p_0(1-p_0)\pi(p_0)}. \qedhere
\end{equation}
\end{proof}

Having characterized solutions $\mu_t(dx)$ to the associated linear problem of Equation \eqref{eq:linearmeasuremultiselect}, in Section \ref{sec:exponentialnormalization} we show another way to interpret measure-valued solutions $\ol{\mu}_t(dx)$ to Equation \eqref{multiselect} in terms of the linear dynamics.

\subsubsection{Expressing Solutions to Multilevel Dynamics via Exponential Normalization Relation}
\label{sec:exponentialnormalization}

In this section, we will show solutions $\ol{\mu}_t(dx)$ to \eqref{multiselect} can be expressed in terms of the mass of the solutions $\mu_t([0,1])$ to Equation \eqref{eq:linearmeasuremultiselect} and on the average group-level reproduction rate across the groups in the population. This average collective outcome across a measure $\mu(dx)$ is defined as
\begin{equation} \label{eq:Gmutavg}
    \langle G(\cdot) \rangle_{\mu} := \frac{ \int_0^1 G(x) \mu(dx)}{\int_0^1 \mu(dx)}.
\end{equation}
Using the normalization relation of Equation \eqref{eq:olmutnormalized} and the fact that $\ol{\mu}_t(dx)$, we can compute that
\begin{equation} \label{eq:Gavgequiv}
    \langle G(\cdot) \rangle_{\mu_t} = \frac{ \int_0^1 G(x) \mu(dx)}{\int_0^1 \mu(dx)} = \int_0^1 G(x) \left( \frac{\mu_t(dx)}{\int_0^1 \mu_t(dy)} \right) = \int_0^1 G(x) \ol{\mu}_t(dx) = \langle G(\cdot)\rangle_{\ol{\mu}_t},
\end{equation}
so the unnormalized and normalized solutions $\mu_t(dx)$ and $\ol{\mu}_t(dx)$ feature the same average group-level reproduction rates. 

By considering the test-function $v(x) = 1$ and using the assumption that $\mu_0(dx)$ is a probability measure, we can calculate the mass of $\mu_t(dx)$ solving Equation \eqref{eq:linearmeasuremultiselect} satisfies the following ordinary differential equation 
\begin{equation}
   \dsddt{} \mu_t\left([0,1]\right) =  \dsddt{} \int_0^1 \mu_t(dx) = \lambda \int_0^1 G(x) \mu_t(dx) %\: \: , \: \: \int_0^1 \mu_0(dx) = 1.
\end{equation}
Dividing both sides by $\mu_t([0,1])$, we can see from Equation \eqref{eq:Gmutavg} that
\begin{equation}
    \frac{1}{\mu_t([0,1])} \dsddt{} \mu_t([0,1]) = \frac{1}{\int_0^1 \mu_t(dx)} \dsddt{} \int_0^1 \mu_t(dx) = \frac{\int_0^1 G(x) \mu_t(dx)}{\int_0^1 \mu_t(dx)} = \langle G(\cdot) \rangle_{\mu_t}.
\end{equation}
Integrating this differential equation, we see that the mass of the solution $\mu_t(dx)$ to the linear multilevel dynamics is given by
\begin{equation} \label{eq:solutionmass} 
    \mu_t\left([0,1]\right) = \int_0^1 \mu_t(dx) = \exp\left( \lambda \int_0^t \langle G(\cdot) \rangle_{\mu_s} ds \right).
\end{equation}
Then, applying this to the normalization relation from Equation \eqref{eq:olmutnormalized}, we can express solutions $\ol{\mu}_t$ to the full multilevel dynamics in terms of $\mu_t(dx)$ by
\begin{equation} \label{eq:olmutexpweight}
   \int_0^1 v(x) \ol{\mu}_t(dx) = \frac{\int_0^1 v(x) \mu_t(dx)}{\mu_t\left([0,1]\right)} = \exp\left(-\lambda \int_0^t \langle G(\cdot) \rangle_{\mu_s} ds \right) \int_0^1 v(x)  \mu_t(dx).
\end{equation}
This representation of the solution $\ol{\mu}_t(dx)$ for the full nonlinear multilevel dynamics of Equation \eqref{multiselect} in terms of the solutions $\mu_t(dx)$ of the linear dynamics is particularly useful for understanding the long-time dynamics of the collective outcome $\int_0^t \langle G(\cdot) \rangle_{\mu_s} ds$. In particular, we will use this and the fact that solutions $\ol{\mu}_t(dx)$ are normalized to obtain the long-time bounds on time-averaged collective fitness presented in Theorem \ref{thm:groupbounds}. Noting that $\ol{\mu}_t([0,1]) = 1$ and that $\langle G(\cdot) \rangle_{\mu_s} = \langle G(\cdot) \rangle_{\ol{\mu}_s}$ from Equation \eqref{eq:Gavgequiv}, we can apply the test-function $v(x) = 1$ in Equation \eqref{eq:olmutexpweight} to see that
\begin{equation} \label{eq:expnormsecond}
    1 = \ol{\mu}_t([0,1]) = \exp\left(-\lambda \int_0^t \langle G(\cdot) \rangle_{\ol{\mu}_s} ds \right) \int_0^1 \mu_t(dx).
\end{equation}
Rearranging this equation allows us to deduce Equation \eqref{eq:expnormintro}, which highlights the connection between principal growth rates for $\mu_t(dx)$ and the average group-level reproduction rate $\langle G(\cdot) \rangle_{\ol{\mu}_t}.$

\subsection{Upper and Lower Bounds for Principal Growth Rates of Solutions}
\label{sec:growthrates}

Now we can use our push-forward representation to estimate the growth rate for the tails of solutions $\mu_t(dx)$ to the linear multilevel dynamics of Equation \eqref{eq:linearmeasuremultiselect}. In Lemma \ref{lem:wtbounds}, we use the formula from $w_t(x)$ from Lemma \ref{lem:wtformula} to find upper and lower bounds for the growth rates of solutions along characteristics described by $\psi_t(x) = w_t(\phi_t(x))$. Next, in Lemma \ref{lem:probholderbounds}, we apply Lemma \ref{lem:wtbounds} to obtain bounds on the growth rate of the mass $\mu_t\left([a,1]\right)$ for any $a > 0$. These bounds are expressed in terms of the group-reproduction rate of the full-cooperator group $G(1)$, the individual-level advantage of defectors in an otherwise full-cooperator group $\pi(1)$, and the infimum and supremum \holder exponents $\overline{\theta}$ and $\underline{\theta}$ of the initial measure $\mu_0(dx)$ near $x=1$. The form of these rates highlights the conflict between the individual-level incentive to defect and the collective incentive to achieve full-cooperation.

\begin{lemma} \label{lem:wtbounds}
Suppose that $G(x) \in C^1([0,1])$ and that $G(0) < G(1)$. Then, for $x \in (0,1)$, there exist positive constants $\underline{M}, \overline{M} < \infty$ such that
\begin{equation} \label{eq:wtbounds}
    \underline{M} e^{\lambda G(1) t} \phi_t(x)^{\lambda \pi(0)^{-1} \left[G(1) - G(0) \right]} \leq \psi_t(x) \leq \overline{M} e^{\lambda G(1) t}.
\end{equation}
\end{lemma}
\begin{proof}
 For our upper bound, we can use Equation \eqref{wtest} to estimate that
 \begin{equation} \label{eq:wtupperbound} \psi_t(x) = w_t(\phi_t(x)) = e^{\lambda G(1) t} \exp\left(\lambda \int_{\phi_t(x)}^{x} \left[\frac{G(q) - G(1)}{q(1-q) \pi(q)} \right] dq \right) \leq e^{\lambda G(1) t} \underbrace{\exp\left(\lambda \int_{0}^{1} \frac{\left[G(q) - G(1)\right]_+}{q(1-q) \pi(q)} dq \right)}_{:= \overline{M}}, \end{equation}
 where we have ignored non-positive contributions to the integral using the notation
 \begin{equation} \label{eq:Gdiffpiecewise}
 \left[ G(q) - G_a(1) \right]_{+} = \left\{
     \begin{array}{cr}
       G(q) - G(1) & : G(q) > G(1)\\
       0 & : G(q) \leq G(1)
     \end{array}
   \right.   .
\end{equation}
Noting that the integrand in the last term of Equation \eqref{eq:wtupperbound} is bounded as $x \to 1$ because $G(x)$ is $C^1$ and bounded as $x \to 0$ because $G(0) < G(1)$, we can deduce that $\overline{M} < \infty$.

To find a corresponding lower bound on $\psi_t(x)$, we can first our the expression $R(q)$ from Equation \eqref{wtest} as
\begin{equation} \label{eq:Rqexpand} 
R(q) =  \frac{G(0) - G(1)}{\pi(0) q} + \frac{1}{(1-q) \pi_j(q)} \left( \frac{G(q) - G(0) + \left[G(0) - G(1)\right] \left[1 - (1-q) \frac{\pi(q)}{\pi(0)} \right]}{q} \right),
\end{equation}
where we can check that the second term on the righthand side is bounded above on $[0,1]$ because $G(x) \in C^1\left([0,1]\right)$. This means that there exists $N > -\infty$ such that 
\begin{equation}
\int_{\phi_t(x)}^x R(q) dq \geq N +   \left[\frac{G(1) - G(0)}{\pi(0)} \right] \log\left(\phi_t(x) \right), 
\end{equation}
and then, letting $\underline{M} := e^{\lambda N}$, we can use Equation \eqref{wtest} to estimate that
\begin{equation} \label{eq:wtlowerbound}
    \psi_t(x) = e^{\lambda G(1) t} \exp\left(\lambda \int_{\phi_t(x)}^x R(q) dq  \right) \geq \underline{M} e^{\lambda G(1) t} \phi_t(x)^{\lambda \pi(0)^{-1} \left[G(1) - G(0) \right]}.  \qedhere
\end{equation}
\end{proof}

\begin{lemma} \label{lem:probholderbounds}
Suppose that $G(x) \in C^1([0,1])$ and $G(0) < G(1)$. Consider $a \in (0,1)$ and $\mu_0(dx)$ with supremum and infimum \holder exponents near $x=1$ satisfying $0 < \underline{\theta} \leq \overline{\theta}$. For any $\Theta < \underline{\theta}$, there exists $\supholderupperbound{\Theta} < \infty$ such that 
\begin{subequations} \label{eq:probhoulderbounds}
\begin{equation} \label{eq:probaholderupperboundalways}
    \mu_t\left(\left[a,1\right]\right) \leq \supholderupperbound{\Theta}  e^{\left[\lambda G(1) - \Theta  \pi(1) \right] t}, 
\end{equation}
while for any $\Theta > \underline{\theta}$, there exists $\supholderlowerbound{\Theta} > 0$ and a sequence of times $\{t_n\}_{n \in \NN}$ such that 
\begin{equation} \label{eq:probaholderupperboundsubsequence}
    \mu_{t_n}\left(\left[a,1\right]\right) \geq \supholderlowerbound{\Theta}  e^{\left[\lambda G(1) - \Theta \pi(1) \right] t_n}. 
\end{equation}
Similarly, for $\Theta > \overline{\theta}$, there exists $\infholderlowerbound{\Theta} > 0$ such that  
\begin{equation} \label{eq:probaholderlowerboundalways}
    \mu_t\left(\left[a,1\right]\right) \geq \infholderlowerbound{\Theta} e^{\left[\lambda G(1) - \Theta \pi(1) \right] t}, 
\end{equation}
and for $\Theta< \overline{\theta}$, there exists $\infholderupperbound{\Theta} < \infty$ and a sequence of times $\{s_n\}_{n \in \NN}$ such that
\begin{equation} \label{eq:probaholderlowerboundsubsequence}
    \mu_t\left(\left[a,1\right]\right) \leq \infholderupperbound{\Theta} e^{\left[\lambda G(1) - \Theta \pi(1) \right] s_n}.
\end{equation}
\end{subequations}
Furthermore, when the \holder constants $C_{\underline{\theta}}$ or $C_{\overline{\theta}}$ of $\mu_0(dx)$ near $x=1$ are positive and finite, we can obtain versions of each of the bounds from Equation \eqref{eq:probhoulderbounds} with $\Theta = \underline{\theta}$ or $\Theta = \overline{\Theta}$, respectively. 
\end{lemma}

\begin{proof}
Using the upper bound on $\psi_t(x)$ from Lemma \ref{lem:wtbounds}, we know that there exists an $\overline{M} < \infty$ such that
\begin{equation} \label{eq:mutafirstupperbound}
    \int_{a}^1 \mu_t(dx) = \int_{\phi_t^{-1}(a)}^1 \psi_t(x) \mu_0(dx) \leq \overline{M} e^{\lambda G(1) t} \int_{\phi_t^{-1}(a)}^1  \mu_0(dx) = \overline{M} e^{\lambda G(1) t} \mu_0\left(\left[\phi_t^{-1}(a) ,1 \right] \right).
\end{equation}
Using the lower bound from Lemma \ref{lem:wtbounds}, we know that there is an $\underline{M} > 0$ such that 
\begin{equation}
        \int_{a}^1 \mu_t(dx) = \int_{\phi_t^{-1}(a)}^1 \psi_t(x) \mu_0(dx) \geq \overline{M} e^{\lambda G(1) t} \int_{\phi_t^{-1}(a)}^1  \phi_t(x)^{\lambda \pi(0)^{-1} \left[G(1) - G(0) \right]}\mu_0(dx).
\end{equation}
Using our assumptions that $G(0) < G(1)$ and $\pi(0) > 0$, as well as the fact that $\phi_t(x) \in [a,1]$ for $x \in [\phi_t^{-1}(a),1]$, we can further estimate that 
\begin{equation}  \label{eq:mutafirstlowerbound}
    \int_{a}^1 \mu_t(dx) \geq \underline{M} a^{\lambda \pi(0)^{-1} \left[G(1) - G(0)\right]} \mu_0\left( \left[\phi_t^{-1}(a),1\right] \right).
\end{equation}
 Our next step is to obtain bounds on $\mu_0\left( \left[\phi_t^{-1}(a),1 \right]\right)$. We can use the formula from Lemma \ref{lem:epi1tformula} for $1 - \phi_t^{-1}(x)$ to see that
\begin{dmath} \label{eq:mu0identity}
\mu_0\left(\left[ \phi_t^{-1}(a), 1 \right] \right) = \mu_0\left(\left[ 1 - \left(1 - \phi_t^{-1}(a)\right), 1 \right] \right) = \mu_0\left(\left[ 1 - e^{- \pi(1) t} \left( 1 - a \right) \exp\left(\int_{a}^{\phi_t^{-1}(a)} \frac{Q(q)}{q \pi(q)} dq \right), 1 \right] \right)
\end{dmath}
From our assumption that $\mu_0(dx)$ has supremum \holder exponent $\underline{\theta} > 0$, we know that $\limsup_{x \to 0} x^{-\Theta} \mu_0([1-x,1]) = 0$ for $\Theta < \underline{\theta}$. For such $\Theta$, there is therefore a constant $C$ such that
\begin{dmath}
    \mu_0\left(\left[\phi_t^{-1}(a),1\right] \right) \leq C e^{-\Theta \pi(1) t} \left( 1 - a\right)^{\Theta} \exp\left(\Theta \int_{a}^{\phi_t^{-1}(a)} \frac{Q(q)}{q \pi(q)} dq\right). 
\end{dmath}
Combining this with our estimate from Equation \eqref{eq:mutafirstupperbound} and the expression for $Q(q)$ from Equation \eqref{zest}, we can see that
\begin{equation} \label{eq:mutaupperbound}
    \mu_t\left(\left[a,1\right]\right) \leq \underbrace{C \overline{M} \left(1-a\right)^{\Theta} \exp\left( \Theta \int_{a}^{1} \frac{\left[\pi(1) - q \pi(q)\right]_+}{q \pi(q)} dq\right)}_{:= \supholderupperbound{\Theta}} e^{\left[\lambda G(1) - \Theta \pi(1) \right]t}.
\end{equation}
Similarly, we note that $\limsup_{x \to 0} x^{-\Theta} \mu_0\left(\left[x,1\right]\right) > 0$ for $\Theta > \underline{\theta}$. Combining this with the expression from Equations \eqref{eq:mutafirstlowerbound} and \eqref{eq:mu0identity}, we see that there is $C > 0$ and a sequence of times $\{t_n\}$ tending to infinity such that
\begin{equation} \label{eq:mutasequentiallowerbound}
    \mu_{t_n}\left(\left[a,1\right]\right) \geq \underbrace{C \underline{M} \left(1-a\right)^{\Theta} \exp\left(\Theta \int_{a}^{1} \frac{\left[\pi(1) - q \pi(q)\right]_-}{q \pi(q)} dq\right)}_{:= \supholderlowerbound{\Theta}} e^{\left[\lambda G(1) - \Theta \pi(1) \right]t}.
\end{equation}
Further, we know that $\supholderlowerbound{\Theta} > 0$ and $\supholderupperbound{\Theta} < \infty$ because $\pi(x)$ is $C^1$. We can obtain the desired bounds depending on the infimum \holder exponent $\overline{\theta}$ of $\mu_0(dx)$ near $x=1$ with an analogous approach.
\end{proof}

We can also derive approximate lower bounds for the growth rate of $\mu_t(dx)$ in terms of the growth rate near the equilibrium of the within-group dynamics.

\begin{lemma} \label{lem:muGwbound}
Suppose that $G(x), \pi(x)$ satisfy the assumptions of Theorem \ref{wlimtheorem} and that $\mu_0(dx)$ has positive \holder exponent $\underline{\theta}$ near $x=1$ and consider the quantity $\check{G}_w := \min_{x \in [0,a]} G(x)$. For $w$ sufficiently close to $0$, there exists a constant $A_w > 0$ such that 
\begin{equation} \label{eq:muGwbound}
 \mu_t\left(\left[0,1 \right]\right) \geq A_w e^{\lambda \check{G}_w t}. 
\end{equation}
\end{lemma}
\begin{proof}
 Due to our assumption that $\mu_0(dx)$ has supremum \holder exponent $\underline{\theta} > 0$, there is $w' < 1$ such that $\mu_0\left([0,w'] \right) > 0$. Because $\phi_t(x)$ decreases in time and satisfies $\phi_t(x) \to 0$ for $x \in (0,1)$ as $t \to \infty$, we have that, for any $w < w'$, there is a $T_w$ such that $\phi_{T_w}(w') = w$ and $\phi_t(w') \leq w$ for $t > T_w$. Now we can estimate the integral of the group-reproduction function along characteristic curves by
\begin{equation}
    \int_0^t G(\phi_s(x)) ds = \int_{T_w}^t G(\phi_s(x)) ds + \int_{0}^{T_w} G(\phi_s(x)) ds \geq %
    \check{G}_w t + \left( \check{G}_{w'} - \check{G}_w \right) T_w.
\end{equation}
We can then apply this estimate and the fact that $\mu_0(dx)$ is a probability measure to deduce that% 
\begin{dmath}
{ \int_0^1 \mu_t(dx) =  \ds\int_0^1 \exp\left(\lambda \int_0^t G(\phi_s(x)) ds\right) \mu_0(dx)} \geq \ds\int_0^{w'} \exp\left(\lambda \int_0^t G(\phi_s(x)) ds\right) \mu_0(dx) \geq e^{\lambda \left( \check{G}_{w'} - \check{G}_w \right) T_w} e^{\lambda \check{G}_w t} \ds\int_0^{w'} \mu_0(dx) %
\geq  e^{\lambda \left( \check{G}_{w'} - \check{G}_w \right) T_w} e^{\lambda \check{G}_w t}. \mbox{\qedhere}
\end{dmath}
\end{proof}

\section{Convergence to Steady State Population for Initial Measures with Well-Defined H{\"o}lder Exponents} \label{sec:convergencesteady}

In this section, we consider the long-time behavior of solutions $\ol{\mu}_t(dx)$ to Equation \eqref{multiselect} for initial conditions with well-defined \holder exponents and \holder constants. In Section \ref{sec:weakconvergencesteady}, we prove Theorem \ref{wlimtheorem}, demonstration weak convergence of the population to steady state densities for sufficiently strong between-group competition and initial measures with well-defined nonzero, finite \holder exponents and constants. In Section \ref{sec:deltaone}, we state and prove Proposition \ref{prop:deltaone}, which tells us that a population with an initial partial delta-peak at full-cooperation will fix full-cooperation in the group-structure population in the long-time limit.

\subsection{Weak Convergence for Measure-Valued Initial Population} \label{sec:weakconvergencesteady}

Before presenting the proof of Theorem \ref{wlimtheorem}, we will first rewrite our expression for the steady state densities from Equation \eqref{pthetactilde} into a form that is most compatible with the expressions related to the method of characteristics derived in Lemmas \ref{lem:epi1tformula}, \ref{lem:wtformula}, and \ref{lem:phitinvformula}. We start by considering the following expression for steady state density solutions to Equation \eqref{multiselect}
\begin{equation}\label{ptheta}
f^{\lambda}_{\theta}(x)=(1-x)^{\theta-1}\frac{\pi(1)}{x\pi(x)}\exp\paren{\int_x^1\frac{M_{\theta}(s)ds}{s\pi(s)}}, \; M_{\theta}(s)=-\lambda R(s)+\theta Q(s).
\end{equation}
Combining the steady state expressions from Equation \eqref{pthetactilde} and \eqref{ptheta}, we can write $M_{\theta}(s)$ using the following decomposition
\begin{equation} \label{eq:Mdecomp}
    M_{\theta}(s) = \frac{\nu}{s} + \frac{C(s)}{\pi(s)}
\end{equation}
where $\nu$ and $C(s)$ are given by Equations \eqref{eq:nudefinition} and \eqref{eq:lambdaCofx}, respectively.

Now we will study the convergence of $\mu_t(dx)$ to these density steady states. We will start by the integrating the righthand side of Equation \eqref{eq:pushforward} by parts to see that
\begin{equation}
\begin{split}
\int_0^1 v(\phi_t(y))\psi_t(y)\mu_0(dx)=&v(1)\psi_t(1)F(1^-)-v(0)\psi_t(0)F(0^+)\\
&-\int_0^1\PD{}{y}\paren{v(\phi_t(y))\psi_t(y)}F(y) dy, \: \: F(y)=\mu_0([0,y]).
\end{split}
\end{equation}
In the above, we used the fact that $\phi_t(0)=0$ and $\phi_t(1)=1$.
For initial measures $\mu_0(dx)$ with positive supremum \holder exponent $\underline{\theta}$, $F(1^-)=1$ and therefore 
\begin{equation}\label{vpsimu0}
\int_0^1 v(\phi_t(y))\psi_t(y)\mu_0(dx)=v(0)\psi_t(0)(1-F(0^+))
+\int_0^1\PD{}{y}\paren{v(\phi_t(y))\psi_t(y)}(1-F(y))dy.
\end{equation}

In Lemma \ref{lem:integraltermbound}, we estimate the growth rate of both terms in Equation \eqref{vpsimu0}, and show that the time-dependent family of integrands in the second term is bounded by an integrable function that is independent of time. Using Lemma \ref{lem:integraltermbound}, we can then apply the Dominated Convergence Theorem to help prove convergence to steady state densities in Theorem \ref{wlimtheorem}. 
\begin{lemma} \label{lem:integraltermbound}
Suppose that $G(x), \pi(x)$ satisfy the assumptions of Theorem \ref{wlimtheorem} and that $\mu_0(dx)$ has supremum \holder exponent $\underline{\theta}$ near $x=1$ that is nonzero and finite. If $\nu_{\theta}= \lambda \left[G(1) - G(0)\right] - \underline{\theta} \pi(1) > 0$, then, for any $\Theta < \underline{\theta}$, there exists a constant $C^v(\Theta) < \infty$ such that
\begin{equation} \label{eq:intergraltermbound} \begin{aligned}
    \bigg| \PD{}{y}\paren{v(\phi_t(y))\psi_t(y)}(1-F(y)) \bigg| &\leq C^v(\Theta) e^{\left[\lambda G(1) - \Theta \pi(1)\right] t} d(x)  \\ d(x) &:=  \begin{cases}
(1-x)^{\Theta-1}x^{\nu-1} &\text{ if } \nu<1,\\
(1-x)^{\Theta-1} &\text{ if } \nu\geq 1.
\end{cases} 
\end{aligned}
\end{equation}
Furthermore, we can use Equation \eqref{vpsimu0} to see that there are positive constants $C_1^v$ and $C_2^v(\Theta)$ such that 
\begin{equation} \label{eq:bothtermrates}
 \int_0^1 v(\phi_t(y))\psi_t(y)\mu_0(dx) \leq C_1^v e^{\lambda G(0) t} + C_2^v(\Theta) e^{\left[\lambda G(1) - \Theta \pi(1) \right] t}.  
\end{equation}
If, in addition, the supremum \holder constant of $\mu_0(dx)$ is finite and nonzero, then we can obtains bounds analogous to those of Equations \eqref{eq:intergraltermbound} and \eqref{eq:bothtermrates} for $\Theta = \underline{\theta}$. 
\end{lemma}
\begin{proof}
 We start with the integral in Equation \eqref{vpsimu0}.
Changing variables to $x=\phi_t(y)$, we have:
\begin{equation}
\int_0^1\PD{}{y}\paren{v(\phi_t(y))\psi_t(y)}(1-F(y))dy
=\int_0^1 \PD{}{x}\paren{(v(x)w_t(x)}(1-F(\phi_t^{-1}(x)))dx
\end{equation}
Our goal is to estimate the following quantity
\begin{dmath}\label{I1I2}
 {\int_0^1 I_1I_2dx := e^{-(\lambda G(1)-\theta\pi(1))t} \int_0^1\PD{}{x}\paren{(v(x)w_t(x)}(1-F(\phi_t^{-1}(x)))dx,}\\
I_1=e^{-\lambda G(1) t} \PD{}{x}\paren{(v(x)w_t(x)},\quad I_2=e^{\theta \pi(1)t}(1-F(\phi_t^{-1}(x))).
\end{dmath}
Let us first consider $I_1$.
Using Equation \eqref{wtest}, we have that
\begin{equation}\label{I1exp}
\begin{split}
I_1&=\paren{\PD{v}{x}-\lambda K(x) v(x)}\exp\paren{-\lambda \int_x^{\phi_t^{-1}(x)}\frac{R(s)}{s\pi(s)}ds},\\
K(x)&=\frac{R(\phi_t^{-1}(x))}{\phi_t^{-1}(x) \pi(\phi_t^{-1}(x))}\PD{}{x}\phi_t^{-1}(x)-\frac{R(x)}{x\pi(x)}.
\end{split}
\end{equation}
Using Equation \eqref{Jacobian} and the definition of $R$ in Equation \eqref{wtest}, we have:
\begin{equation}\label{Kexp}
K(x)=\frac{G(x)-G(\phi_t^{-1}(x))}{x(1-x)\pi(x)}.
\end{equation}
Using the fact that $G$ is $C^1$ and that $\phi_t^{-1}(x)\geq x$, we have that
\begin{equation}
\abs{G(x)-G(\phi_t^{-1}(x))}=\abs{\int_x^{\phi_t^{-1}(x)} G'(z)dz}\leq \norm{G'}_{\infty}(\phi_t^{-1}(x)-x)
\end{equation}
Denoting $\pi_{\rm min}:=\min_{0\leq s\leq 1} \pi(s)>0$, we may thus estimate $K(x)$ by
\begin{equation}
\abs{K(x)}\leq \frac{\norm{G'}_{\infty}(\phi_t^{-1}(x)-x)}{x(1-x)\pi(x)}\leq \frac{\norm{G'}_{\infty}\phi_t^{-1}(x)}{\pi_{\rm min}x(1-x)},%,\\
\end{equation}
and therefore deduce that
\begin{equation}
\abs{\PD{v}{x}-\lambda K(x) v(x)}\leq \norm{\PD{v}{x}}_{\infty}+\frac{\lambda\norm{G'}_{\infty}\norm{v}_{\infty}\phi_t^{-1}(x)}{\pi_{\rm min}x(1-x)} \leq \underbrace{\left(\norm{\PD{v}{x}}_{\infty}+\frac{\lambda\norm{G'}_{\infty}\norm{v}_{\infty}}{\pi_{\rm min}}\right)}_{:= C_1} \frac{\phi_t^{-1}(x)}{x(1-x)}.
\end{equation}
Returning to Equation \eqref{I1exp}, we see that $I_1$ can be estimated as
\begin{equation}\label{I1est}
\abs{I_1}\leq C_1\frac{\phi_t^{-1}(x)}{x(1-x)}\exp\paren{-\lambda \int_x^{\phi_t^{-1}(x)}\frac{R(s)}{s\pi(s)}ds}.
\end{equation}
We  now estimate $I_2$ in Equation \eqref{I1I2}. Because $\mu_0(dx)$ has supremum \holder exponent $\underline{\theta}$ near $x=1$ and $1-F(\phi_t^{-1}(x)) = \mu_0\left(\left(\phi_t^{-1}(x),1\right]\right)$, we can use Equation \eqref{eq:mu0identity} to see that, for any $\Theta < \underline{\theta}$, there is a constant $C_2 < \infty$ such that
\begin{equation} \label{I2est}
\abs{I_2}\leq \exp(\Theta \pi(1)t) (1-\phi_t^{-1}(x))^\Theta \leq 
C_2(1-x)^\Theta \exp\paren{\Theta \int_x^{\phi_t^{-1}(x)}\frac{Q(s)}{s\pi(s)}ds}.%
\end{equation}
Combining Equations \eqref{I1est} and \eqref{I2est}, we have that
\begin{equation}\label{I1I2est}
\begin{split}
\abs{I_1I_2}&\leq C_1C_2(\Theta)(1-x)^{\theta-1}\frac{\phi_t^{-1}(x)}{x}I_3,\; I_3=\exp\paren{\int_x^{\phi_t^{-1}(x)}\frac{M_{\Theta}(s)}{s\pi(s)}ds},
\end{split}
\end{equation}
where $M_{\Theta}(s)$ is defined as in Equation \eqref{ptheta}.
To estimate $I_3$, we see from Equation \eqref{eq:Mdecomp} that
\begin{equation}
\int_x^{\phi_t^{-1}(x)}\frac{M_{\Theta}(s)}{s\pi(s)}ds =-\nu_{\Theta} \ln\paren{\frac{\phi_t^{-1}(x)}{x}}+\int_x^{\phi_t^{-1}(x)}\frac{C(s)}{\pi(s)}ds,
\end{equation}
and we can use the fact that $C(s)$ is bounded on $[0,1]$ to see that 
\begin{equation}
\abs{I_3}\leq C_3 \paren{\frac{x}{\phi_t^{-1}(x)}}^\nu, \; C_3=\exp\paren{\int_0^1\frac{C(s)}{\pi(s)}ds} < \infty.
\end{equation}

Combining this with Equation \eqref{I1I2est}, we see that
\begin{equation}
\abs{I_1I_2}\leq C^{v}(\Theta)(1-x)^{\theta-1}\paren{\frac{x}{\phi_t^{-1}(x)}}^{\nu_{\Theta}-1}, \; C^{v}(\Theta)=C_1C_2(\Theta)C_3.
\end{equation}
Since $\phi_t^{-1}(x)\geq x$, we can use the definitions of $d(x)$, $I_1$, and $I_2$ from Equations \eqref{eq:intergraltermbound} and \eqref{I1I2} to see that $\abs{I_1I_2} \leq d(x) e^{\left[\lambda G(1) - \Theta \pi(1) \right] t}$, where we note that $d(x)$ is integrable because $\Theta > 0$ and $\nu_{\Theta} > \nu_{\theta} > 0$ for any $\Theta < \theta$. We can further use this estimate of  $|\abs{I_1I_2}|$, the fact that $\psi_t(0) = e^{\lambda G(0) t}$, and the choice of constants $C_1^v = |v(0)| (1 - F(0^+))$ and $C_2^v(\Theta) = C^v(\Theta) \int_0^1 d(x) dx$ to obtain the estimate of Equation \eqref{eq:bothtermrates}.
\end{proof}

\begin{proof}[Proof of Theorem \ref{wlimtheorem}]
We first prove Equation \eqref{vmubarlim} when $v(x)$ is a $C^1$ function. We will evaluate the following quantity as $t\to \infty$
\begin{equation}
e^{-(\lambda G(1)-\theta \pi(1))t} \int_0^1 v(x)\mu_t(dx) 
=e^{-(\lambda G(1)-\theta \pi(1))t} \int_0^1 v(\phi_t(y))\psi_t(y)\mu_0(dx),
\end{equation}
which we can express through Equation \eqref{vpsimu0}. From the boundary term of Equation \eqref{vpsimu0}, we can use the fact that $\psi_t(0)=e^{\lambda G(0) t}$ to see that
\begin{equation}\label{bdryest}
\lim_{t\to \infty} e^{-(\lambda G(1)-\theta \pi(1))t} v(0)\psi_t(0)(1-F(0^+)) 
=\lim_{t\to \infty} e^{-(\lambda(G(1)-G(0))-\theta \pi(1))t}v(0)(1-F(0^+))=0,
\end{equation}
where the last equality follows from the fact that $\mu_0(dx)$ has \holder exponent $\theta$ near $x=1$. Recalling that $\theta>0$ and $\nu>0$ by assumption, we can apply Lemma \ref{lem:integraltermbound} to see that the integrand on the righthand side of Equation \eqref{vpsimu0} is bounded by an integrable function independent of $t$. We may thus use the Dominated Convergence Theorem to pass to the limit as $t\to \infty$ in the integral in Equation \eqref{I1I2}.

Using Equations \eqref{wtestlim} and \eqref{Kexp}, we find that
\[ \lim_{t \to \infty} K(x) = \lim_{t \to \infty} \frac{G(x) - G(\phi_t^{-1}(x))}{x(1-x) \pi(x)} = - \frac{1}{x \pi(x)} \left( \frac{G(1) - G(x)}{1-x} \right) = - \frac{R(x)}{x \pi(x)}. \] 
Combining this with Equation \eqref{I1est}, we can compute that 
\begin{equation}
\lim_{t\to \infty} I_1 = \paren{\PD{v}{x}+\frac{\lambda R(x)v(x)}{x\pi(x)}}\exp\paren{-\lambda \int_x^1\frac{R(s)}{s\pi(s)}ds} 
=\PD{}{x}\paren{v(x)\exp\paren{-\lambda \int_x^1\frac{R(s)}{s\pi(s)}ds}}.
\end{equation}
Furthermore, can use Equation \eqref{zestlim} and the fact that $\mu_0(dx)$ has \holder exponent $\theta$ with constant $C_{\theta}$ near $x=1$ to see that
\begin{equation}
\lim_{t\to \infty} I_2=\lim_{t\to \infty}\frac{1-F(\phi_t^{-1}(x))}{(1-\phi_t^{-1}(x))^\theta}\paren{\exp(\pi(1)t)(1-\phi_t^{-1}(x))}^\theta
=C_\theta (1-x)^\theta\exp\paren{\theta \int_x^1 \frac{Q(s)ds}{s\pi(s)}}.
\end{equation}
Therefore we deduce that 
\begin{dmath}
\lim_{t\to \infty} \int_0^1 I_1I_2dx
=\int_0^1 \PD{}{x}\paren{v(x)\exp\paren{-\lambda \int_x^1\frac{R(s)}{s\pi(s)}ds}}
C_\theta (1-x)^\theta\exp\paren{\theta \int_x^1 \frac{Q(s)ds}{s\pi(s)}}dx
= \left.\paren{v(x)C_\theta(1-x)^\theta \exp\paren{\int_x^1 \frac{M_{\theta}(s)ds}{s\pi(s)}}}\right|_0^1
-\int_0^1 \paren{v(x)\exp\paren{-\lambda \int_x^1\frac{R(s)}{s\pi(s)}ds}}
\PD{}{x}\paren{C_\theta (1-x)^\theta\exp\paren{\theta \int_x^1 \frac{Q(s)ds}{s\pi(s)}}}dx,
\end{dmath}
where we integrated by parts in the second equality. Using Equation \eqref{eq:Mdecomp} and the fact that $\nu_{\theta} > 0$, we see that the boundary term vanishes.
After some simplifications, we see that:
\begin{equation}
\lim_{t\to \infty} \int_0^1 I_1I_2dx=C_\theta \theta \int_0^1 v(x)f^{\lambda}_{\theta}(x)dx.
\end{equation}
Combining this with Equation \eqref{bdryest}, we see that:
\begin{equation}\label{vmutlim}
\lim_{t\to \infty} e^{-\left[\lambda G(1)-\theta \pi(1)\right]t}\int_0^1 v(x)\mu_t(dx)=C_\theta \theta \int_0^1 v(x)f^{\lambda}_{\theta}(x)dx.
\end{equation}
Using the test-function $v(x) = 1$, we further have that
\begin{equation}\label{1mutlim}
\lim_{t\to\infty}  e^{-\left[\lambda G(1)-\theta \pi(1)\right]t} \int_0^1 \mu_t(dx) =C_\theta \theta \int_0^1 f^{\lambda}_{\theta}(x) dx.
\end{equation}
Using the normalization relation from Equation \eqref{eq:olmutnormalized}, we may now compute the limit in Equation \eqref{vmubarlim}.
\begin{equation}\lim_{t\to \infty} \int_0^1 v(x)\ov{\mu}_t(dx) = \lim_{t\to \infty} \frac{ e^{-\left[\lambda G(1)-\theta \pi(1)\right]t}  \int_0^1 v(x)\mu_t(dx)}{ e^{-\left[\lambda G(1)-\theta \pi(1)\right]t}  \int_0^1 \mu_t(dx)} =  \frac{\int_0^1 v(x)f^{\lambda}_{\theta}(x)dx}{\int_0^1 f^{\lambda}_{\theta}(x)dx} = \int_0^1 v(x) p^{\lambda}_{\theta}(x) dx,
\end{equation}
where we used Equation \eqref{vmutlim} and \eqref{1mutlim} in the second equality.
We have thus established \eqref{vmubarlim} when $v(x)$ is $C^1$. To see that Equation \eqref{vmubarlim} is valid for merely continuous $v(x)$, we may use a standard approximation argument. 
Note that we may approximate $v(x)$ arbitrarily closely 
by a $C^1$ function $w(x)$ in the sup norm:
\begin{equation}\label{wprop}
\text{for any } \epsilon>0, \text{ there exists } w(x)\in C^1 \text{ such that } \norm{v-w}_{\infty}\leq \epsilon.
\end{equation}
Using this and the fact that $\overline{\mu}_t(dx)$ and $p^{\lambda}_{\theta}(x) dx$ are probability measures, we see that
\begin{dmath}
\abs{\int_0^1 v\bar{\mu}_t(dx)-\int_0^1 vp^{\lambda}_{\theta} dx}
\leq \int_0^1 \abs{v-w}\bar{\mu}_t(dx)+
\abs{\int_0^1 w\bar{\mu}_t(dx)-\int_0^1 wp^{\lambda}_{\theta} dx}
+\int_0^1 \abs{v-w}p^{\lambda}_{\theta} dx
\leq 2\epsilon+\abs{\int_0^1 w\bar{\mu}_t(dx)-\int_0^1 wp^{\lambda}_{\theta} dx},
\end{dmath}
As $t\to \infty$, the last expression tends to $0$ since $w(x)$ is $C^1$. Since $\epsilon$ is arbitrary, we obtain the desired conclusion.
\end{proof}

\subsection{Convergence to Delta-Function at Full-Cooperation}
\label{sec:deltaone}

Next, we consider the case in which a positive fraction of groups in the initial population are concentrated at the all-cooperator composition. We show in Proposition \ref{prop:deltaone} that if the initial population contains a positive probability of full-cooperator groups, then the whole population will concentrate upon full-cooperation in the long-time limit. 

\begin{proposition} \label{prop:deltaone}
Suppose that $G(x), \pi(x)$ satisfy the assumptions of Theorem \ref{wlimtheorem} and that the initial population can be written as 
\begin{equation} \label{eq:interiordelta1}
\mu_0(dx) = a_1 \delta(1-x) + \left( 1 - a_1 \right) \rho_0(dx)
\end{equation}
for $a_1 > 0$ and $\rho_0(dx)$ a probability measure with a supremum \holder exponent of $\underline{\theta} > 0$ near $x = 1$. Then the solution to Equation \eqref{multiselect} $\ol{\mu}_t(dx)$ converges weakly to the delta-function $\delta(1-x)$ concentrated at full-cooperation as $t \to \infty$. 
\end{proposition}

\begin{proof} %[Proof of Proposition \ref{prop:deltaone}]
 Using the push-forward representation of $\mu_t(dx)$, we find that
 \begin{dmath} \label{eq:mubartdelta1}
    \int_0^1 v(x) \ol{\mu}_t(dx) = \frac{e^{-\lambda G(1) t} \int_0^1 v(x) \mu_t(dx) }{e^{-\lambda G(1) t} \int_0^1 \mu_t(dx)} = \frac{a_1 v(1) + \left(1 - a_1\right) e^{-\lambda G(1) t} \int_0^1 v(\psi_t(x)) \psi_t(x) \rho_0(dx) }{a_1  + \left(1 - a_1\right) e^{-\lambda G(1) t} \int_0^1  \psi_t(x) \rho_0(dx)}.
 \end{dmath}
 Using Lemma \ref{lem:integraltermbound}, we know that there exist positive constants $C_1^v,C_2^v < \infty$ such that for any $\tilde{\theta} \in (0,\underline{\theta})$
 \[ \left(1 - a_1\right) e^{-\lambda G(1) t} \bigg|\int_0^1 v(\psi_t(x)) \psi_t(x) \rho_0(dx)\bigg| \leq \left(1-a_1\right) \left[ C_1^v e^{\lambda \left[G(0) - G(1) \right] t} + C_2^v e^{-\tilde{\theta} \pi(1) t} \right]
 \]
 Because $G(0) < G(1)$ and $\tilde{\theta} > 0$, we can further see that 
 \[ \left(1 - a_1\right) e^{-\lambda G(1) t} \bigg|\int_0^1 v(\psi_t(x)) \psi_t(x) \rho_0(dx)\bigg| \to 0 \: \: \mathrm{as} \: \: t \to \infty. \]
 Applying this limit to the numerator and denominator in Equation \ref{eq:mubartdelta1}, we can then conclude that
 \begin{equation}
     \ds\lim_{t \to \infty} \int_0^1 v(x) \ol{\mu}_t(dx) = v(1),
 \end{equation}
 and we have shown that $\ol{\mu}_t(dx) \rightharpoonup \delta(1-x)$. 
\end{proof}

\begin{remark}
Proposition \ref{prop:deltaone} generalizes results for the Luo-Mattingly model and models from evolutionary games in which full-cooperation maximizes collective payoff \cite{luo2017scaling,cooney2020pde}. The proofs in those cases had relied on the fact that $G(1) \geq \langle G(x)\rangle_  {\mu_t(dx)}$ at all times $t$, but did not require the additional assumption that the portion of the initial measure not concentrated at full-cooperation have a positive supremum \holder exponent near $x=1$. In the present case, the population concentrates upon full-cooperation even when group reproduction function is maximized by an interior level of cooperation. 
\end{remark}

 \section{Long-Time Behavior of Multilevel PD Dynamics for General Initial Measures} \label{sec:extinctionsurvival}

In this section, we consider the long-time behavior of solutions of Equation \eqref{multiselect} for initial measures $\mu_0(dx)$ with a given supremum \holder exponent $\underline{\theta}$ near $x=1$. 
For the Prisoner's Dilemma case, we present the proof of Theorem \ref{thm:groupbounds} characterizing long-time bounds on the time-averaged collective group reproduction-rate in Section \ref{sec:groupbounds} and the proof of Theorem \ref{thm:extinctvspersist} for extinction or weak persistence of cooperation in Section \ref{sec:convergencedeltazero}. In Section \ref{sec:PDel}, we prove Proposition \ref{prop:PDeldelta}, showing that $\ol{\mu}_t(dx)$ concentrates upon full-cooperation in the Prisoners' Delight case.

\subsection{Bounds on Average Group Reproduction Function}
\label{sec:groupbounds}

Before presenting the proof of Theorem \ref{thm:groupbounds}, we first provide a result analogous to Lemma \ref{lem:integraltermbound} showing the existence of a sequence tending on which solutions to Equation \eqref{eq:linearmeasuremultiselect} can be bounded in terms of the infimum \holder exponent near $x=1$. 
\begin{lemma} \label{lem:massinfimumbound}
Suppose that $G(x), \pi(x)$ satisfy the assumptions of Theorem \ref{wlimtheorem} and that $\mu_0(dx)$ has infimum \holder exponent $\overline{\theta}$ near $x=1$ that is nonzero and finite. Then, for $\Theta < \overline{\theta}$, there are positive constants $E_1^v$ and $E_2^v(\Theta)$ and a sequence $\{t_n\}_{n \in \NN}$ satisfying $t_n \to \infty$ such that 
\begin{equation} \label{eq:infbothtermrates}
 \int_0^1 v(\phi_{t_n}(y))\psi_{t_n}(y)\mu_0(dx) \leq E_1^v e^{\lambda G(0) t_n} + E_2^v(\Theta) e^{\left[\lambda G(1) - \Theta \pi(1) \right] t_n}.  
\end{equation}
\end{lemma}

\begin{proof}[Proof of Theorem \ref{thm:groupbounds}]
 We will first consider the case in which $\overline{\theta} \geq  \underline{\theta} > 0$, and then mention how to generalize our argument to the case in which $\underline{\theta} = 0$. Using the exponential normalization of Equation \eqref{eq:expnormsecond} and the assumption that $\mu_0(dx)$ has supremum \holder exponent $\underline{\theta}$ near $x=1$, %
 we can apply Lemma \ref{lem:integraltermbound} for the test-function $v(x) = 1$ to see that, for $\Theta < \underline{\theta}$, there are constants $C_1^1$, $C_2^1(\Theta)$, and $\tilde{C}$ such that
 \begin{dmath} \label{eq:mutupperalways}
     \ol{\mu}_t\left([0,1]\right) \leq \left( C_1^1 e^{\lambda G(0) t} +  C_2^1(\Theta) e^{\left[\lambda G(1) - \Theta \pi(1)\right]t} \right) \exp\left( - \lambda \int_0^t \langle G(\cdot) \rangle_{\ol{\mu}_s} ds \right) \leq \tilde{C} \exp\left(\left[\max\{\lambda G(1) - \Theta \pi(1), \lambda G(0) \}\right]t   - \lambda \int_0^t \langle G(\cdot) \rangle_{\ol{\mu}_s} ds \right). 
 \end{dmath}
Noting that $\ol{\mu}_t\left([0,1]\right) = 1$, we must eventually have that $\frac{\lambda}{t}  \int_0^t \langle G(\cdot) \rangle_{\ol{\mu}_s} ds \leq \max\{\lambda G(1) - \Theta \pi(1), \lambda G(0)\}$, as otherwise the righthand side of Equation \eqref{eq:mutupperalways} will exceed $1$ for sufficiently large $t$. Because this is true for all $\Theta < \underline{\theta}$, we can deduce that $\limsup_{t \to \infty} \frac{\lambda}{t} \int_0^t \langle G(\cdot) \rangle_{\ol{\mu}_s} ds \leq  \max\{ \lambda G(1) - \underline{\theta} \pi(1), \lambda G(0)\}$.

Next, we look to confirm the reverse inequality. First, we may use Lemma \ref{lem:probholderbounds} and Equation \eqref{eq:expnormsecond} to see that, for any $\Theta > \underline{\theta}$, there is a sequence of times $\{t_n\}_{n \in \NN}$ satisfying $t_n \to \infty$ and a constant $\supholderlowerbound{\Theta}$ such that
\begin{equation} \label{eq:olmutsuplowerbound}
    \ol{\mu}_{t_n}\left([0,1]\right) \geq \supholderlowerbound{\Theta} \exp\left(\left[\lambda G(1) - \Theta \pi(1) \right] t_n - \lambda \int_0^{t_n} \langle G(\cdot) \rangle_{\ol{\mu}_s} ds  \right).
\end{equation}
Using this estimate, we may deduce that $\limsup_{t \to \infty} \frac{\lambda}{t} \int_0^t \langle G(\cdot) \rangle_{\ol{\mu}_s} ds \geq \lambda G(1) - \Theta \pi(1)$ for $\Theta > \underline{\theta}$. \sloppy{Second, we denote $\check{G}_w := \min_{x \in [0,w]} G(x)$ and apply Lemma \ref{lem:muGwbound} to see that, for $w$ sufficiently close to $0$, there exists $A_w > 0$ such that
\begin{equation} \label{eq:olmutG0upperbound}
  \ol{\mu}_t([0,1]) \geq A_w \exp\left(\lambda \left[ \check{G}_w t - \int_0^t \langle G(\cdot) \rangle_{\ol{\mu}_s} ds  \right]  \right).
 \end{equation}
 This allows us to deduce that $\limsup_{t \to \infty} \frac{\lambda}{t} \int_0^t \langle G(\cdot) \rangle_{\ol{\mu}_s} ds \geq \lambda \check{G}_w$ for $w$ sufficiently close to $0$. Combining our two lower bounds, we see that $\limsup_{t \to \infty} \frac{\lambda}{t} \int_0^t \langle G(\cdot) \rangle_{\ol{\mu}_s} ds \geq \max\{ \lambda G(1) - \Theta \pi(1), \lambda \check{G}_w \}$ for any $\Theta > \underline{\theta}$ and $w$ close enough to $0$. Because this bound holds for all such $\Theta$ and $w$, we can use the fact that $\check{G}_w \to G(0)$ to deduce that $\limsup_{t \to \infty} \frac{\lambda}{t} \int_0^t \langle G(\cdot) \rangle_{\ol{\mu}_s} ds = \max\{\lambda G(1) - \underline{\theta} \pi(1), \lambda G(0) \}$. %
Combining this with the reverse inequality, we can conclude that $\limsup_{t \to \infty} \frac{\lambda}{t} \int_0^t \langle G(\cdot) \rangle_{\ol{\mu}_s} ds =  \max\{\lambda  G(1) - \underline{\theta} \pi(1) ,\lambda G(0)\}$.}

To study the corresponding limit infimum, we apply Lemma \ref{lem:massinfimumbound} for the test-function $v(x) = 1$ and the exponential normalization of Equation \eqref{eq:expnormsecond} to see that there are constants $E_1^1$, $E_2^1$, and $\tilde{E}$ and a sequence $\{\tau_n\}_{n \in \NN}$ satisfying $\tau_n \to \infty$ such that, for any $\Theta < \overline{\theta}$,  
\begin{dmath} \label{eq:Esubsequencebound}
\ol{\mu}_{\tau_n}\left([0,1]\right) \leq \left( E_1^1 e^{\lambda G(0) \tau_n} +  E_2^1(\Theta) e^{\left[\lambda G(1) - \Theta \pi(1)\right] \tau_n} \right) \exp\left( - \lambda \int_0^{\tau_n} \langle G(\cdot) \rangle_{\ol{\mu}_s} ds \right) \leq \tilde{E} \exp\left(\max\left\{ \lambda G(1) - \Theta \pi(1), \lambda G(0) \right\}\tau_n   - \lambda \int_0^{\tau_n} \langle G(\cdot) \rangle_{\ol{\mu}_s} ds \right).
\end{dmath}
Because $\ol{\mu}_{\tau_n}\left([0,1]\right)$ remains normalized for all time, Equation \eqref{eq:Esubsequencebound} tells us  that $\liminf_{t \to \infty} \frac{\lambda}{t} \int_0^t \langle G(\cdot) \rangle_{\ol{\mu}_s} ds \leq \max\left\{ \lambda G(1) - \Theta \pi(1), \lambda G(0) \right\}$ for any $\Theta < \overline{\theta}$, and therefore we can deduce that $\liminf_{t \to \infty} \frac{\lambda}{t} \int_0^t \langle G(\cdot) \rangle_{\ol{\mu}_s} ds \leq \max\left\{ \lambda G(1) - \overline{\theta} \pi(1), \lambda G(0) \right\}$. 

To study the analogous lower bound on the time-average collective group-reproduction rate, we first we apply Lemma \ref{lem:probholderbounds} and Equation \eqref{eq:expnormsecond} to see that, for $\Theta \geq \overline{\theta}$, there exists $\infholderlowerbound{\Theta} > 0$ such that 
\begin{equation} \label{eq:olmutinflowerbound}
  \ol{\mu}_t\left([0,1]\right) \geq \infholderlowerbound{\Theta} \exp\left( \left[\lambda G(1) - \Theta \pi(1) \right] t - \lambda \int_0^t \langle G(\cdot) \rangle_{\ol{\mu}_s} ds \right).
\end{equation}
\sloppy{Combining the bounds of Equations \eqref{eq:olmutG0upperbound} and \eqref{eq:olmutinflowerbound} allows us to see that $\liminf_{t \to \infty} \frac{\lambda}{t} \int_0^t \langle G(\cdot) \rangle_{\ol{\mu}_s} ds \geq  \max\left\{ \lambda G(1) - \Theta \pi(1), \lambda \check{G}_w \right\}$ for any $\Theta > \overline{\theta}$ and $w$ sufficiently close to $0$. 
Because this holds for all such $\Theta$ and $w$, we can conclude that %
$\liminf_{t \to \infty} \frac{\lambda}{t} \int_0^t \langle G(\cdot) \rangle_{\ol{\mu}_s} ds \geq \max\left\{ \lambda G(1) - \overline{\theta} \pi(1), \lambda G(0) \right\}$. Because we have also confirmed the reverse inequality, we can then conclude that $\liminf_{t \to \infty} \frac{\lambda}{t} \int_0^t \langle G(\cdot) \rangle_{\ol{\mu}_s} ds = \max\left\{ \lambda G(1) - \overline{\theta} \pi(1), \lambda G(0) \right\}$. }

 After dividing both side by $\lambda$ in our expressions for $\liminf_{t \to \infty} \frac{\lambda}{t} \int_0^t \langle G(\cdot) \rangle_{\ol{\mu}_s} ds$ and $\limsup_{t \to \infty} \frac{\lambda}{t} \int_0^t \langle G(\cdot) \rangle_{\ol{\mu}_s} ds$ then provides us with the long-time bounds presented in Equation \eqref{eq:groupboundslims} on the time-averaged collective outcome $\frac{1}{t} \int_0^t \langle G(\cdot) \rangle_{\ol{\mu}_s} ds$ for the population.

For the cases in which $\underline{\theta} = 0$ or $\overline{\theta} = 0$ , we can still obtain the upper bounds from Equations \eqref{eq:olmutsuplowerbound} and \eqref{eq:olmutinflowerbound} using our existing proof. To obtain the corresponding upper bounds on $\ol{\mu}_t\left([0,1]\right)$, we can use the bound of solutions along characteristics from Lemma \ref{lem:wtbounds} and the fact that $\mu_0(dx)$ is a probability measure. Then applying these bounds allows us to deduce the limiting time-averaged behavior of Equation \eqref{eq:groupboundslims} in these cases as well.
\end{proof}

\subsection{Long-Time Extinction or Persistence of Cooperation} \label{sec:convergencedeltazero}

\begin{proof} [Proof of Theorem \ref{thm:extinctvspersist}]

\sloppy{For the case in which $\lambda \left[G(1) - G(0) \right] > \underline{\theta} \pi(1)$, the long-time persistence of cooperation follows from the bounds from Equation \eqref{eq:groupboundslims} for the time-averaged collective outcome. Because $\limsup_{t \to \infty} \frac{1}{t} \int_0^t \langle G(\dot) \rangle_{\ol{\mu}_s} ds = G(1) - \lambda^{-1} \underline{\theta} \pi(1) > G(0)$, the time-average of $\langle G(\cdot) \rangle_{\ol{\mu}_t} = \int_0^1 G(x) \ol{\mu}_t(dx)$ will achieve a value bounded away from $G(0)$ at an infinite sequence of times, and therefore we can show that the fraction of cooperators $\int_0^1 x \ol{\mu}_{t}(dx)$ must be bounded away from $0$ on an infinite sequence as well.}

Next we turn to proving the extinction of cooperation when $\lambda \left[G(1) - G(0) \right] < \underline{\theta} \pi(1)$. To show weak convergence of $\ol{\mu}_t(dx)$ to $\delta(x)$, we consider any continuous test function $v(x)$ and look to show that $\int_0^1 v(x) \ol{\mu}_t(dx) \to v(0)$ as $t \to \infty$. We can use the continuity of $v(x)$ to show that for any $\epsilon > 0$, there  is a $\delta > 0$ such that
\begin{equation} \label{eq:deltaallGinitialestimate} \bigg| \int_0^1 v(x) \ol{\mu}_t(dx) - v(0)   \bigg| \leq \epsilon + 2 ||v||_{\infty} \int_{\delta}^1 \ol{\mu}_t(dx)  %
\end{equation}
We note that $\mu_0(dx)$ must have a supremum \holder exponent
$\underline{\theta} > 0$ near $x=1$ in order to satisfy our assumption on $\lambda$. Therefore we can apply Lemma \ref{lem:muGwbound} and the normalization relation from Equation \eqref{eq:olmutnormalized} to say that, for $w$ sufficiently close to $0$, there exists $A_w > 0$ such that 
\begin{equation} \label{eq:PDDeltafirstequation}
\int_{\delta}^1 \ol{\mu}_t(dx) = \frac{\int_{\delta}^1 \mu_t(dx)}{\int_0^1 \mu_t(dx) }\leq A_w^{-1} e^{-\lambda \check{G}_w t} \int_{\delta}^1 \mu_t(dx),
\end{equation} 
where $\check{G}_w = \min_{x \in [0,w]} G(x)$. We can apply Lemma \ref{lem:probholderbounds} to the integral on the righthand side of Equation \eqref{eq:PDDeltafirstequation} to say that, for $\Theta < \underline{\Theta}$, there exists $\supholderupperbound{\Theta} < \infty$ such that 
\begin{equation}
    \int_{\delta}^1 \ol{\mu}_t(dx) \leq \supholderupperbound{\Theta}  A_w^{-1} \exp\left(\left\{ \lambda \left[G(1) - \check{G}_w \right] - \Theta \pi(1) \right\} t \right)
\end{equation}
Because $G(\cdot) \in C^1([0,1])$, we know that we can make $\check{G}_w$ arbitrarily close to $G(0)$ by choosing $w$ sufficiently close to $0$. Because our assumption on $\lambda$ is the strict inequality $\lambda \left[ G(1) - G(0) \right] < \underline{\theta}  \pi(1)$, we know, for any given $\lambda$, that we can choose $\Theta$ sufficiently close to $\underline{\theta}$ and $w$ sufficiently close to $0$ such that \begin{equation} \label{eq:conditionapproximate} \lambda \left[ G(1) - \check{G}_w\right] <  \Theta \pi(1) \end{equation} as well. This condition then guarantees that 
\begin{dmath*} 2 ||v||_{\infty}  \int_{\delta}^1 \ol{\mu}_t(dx) \leq 2 B(\Theta) A_w^{-1} ||v||_{\infty}  \exp\left( \left\{ \lambda \left[G(1) - \check{G}_w \right] - \Theta \pi(1) \right\} t  \right)   \to 0 \: \: \mathrm{as} \: \: t \to \infty, \end{dmath*}
which, in combination with Equation \eqref{eq:deltaallGinitialestimate}, allows us to conclude that $\mu_t(dx) \rightharpoonup \delta(x)$ as $t \to \infty$ when $\lambda \left[G(1) - G(0)\right] < \underline{\theta} \pi(1)$.
\end{proof}

Because we show in Theorem \ref{thm:extinctvspersist} that cooperation dies out for $\lambda < \lambda^*$ and weakly persists when $\lambda > \lambda^*$, it is natural to ask what happens in the edge case when $\lambda = \lambda^*$. For the special case in which the group reproduction achieves a unique minimum at $x = 0$, we can show that the population still concentrates at $\delta(x)$ when $\lambda \left[ G(1) - G(0) \right] = \underline{\theta} \pi(1)$. We rely on the following lemma, which was previously used to study special cases of the model under consideration \cite{luo2017scaling,cooney2020pde}.

\begin{lemma} \label{lem:averagepayoffconvergence}
Consider $\pi(x), G(x) \in C^1([0,1])$ and suppose that $G(x) > G(0)$ for $x \in (0,1]$. If $\int_0^{\infty}\left[ \langle G(\cdot) \rangle_{\mu_s} - G(0) \right] ds < \infty$, then $\langle G(\cdot) \rangle_{\mu_t} \to G(0)$ and $\mu_t(dx) \rightharpoonup \delta(x)$ as $t \to \infty$. 
 
\end{lemma}

\begin{proposition} \label{prop:longtimedeltaequality}
Suppose the initial distribution $\mu_0(dx)$ has supremum \holder exponent $\underline{\theta}$ near $1$ with corresponding \holder constant $C_{\underline{\theta}} < \infty$. If $G(x) > G(0)$ for $x \in (0,1]$ and $\lambda \left[G(1) - G(0) \right] = \underline{\theta} \pi(1)$, the $\mu_t(dx) \rightharpoonup \delta(x)$ as $ t \to \infty$. 
\end{proposition}

\begin{proof} [Proof of Proposition \ref{prop:longtimedeltaequality}]
From the continuity of the test function $v(x)$ to see that, for any $\epsilon > 0$, there is a $\delta > 0$  such that %
\begin{dmath} \label{eq:mainestimateequalitycase} \bigg| \int_0^1 v(x) \ol{\mu}_t(dx) - v(0)   \bigg|   \leq \epsilon +  2 ||v||_{\infty}   \int_{\delta}^1 \ol{\mu}_t(dx) %
\end{dmath}
From our assumptions that $\lambda \left[G(1) - G(0)\right] = \underline{\theta} \pi(1)$ and $\mu_0(dx)$ has \holder exponent $\underline{\theta} > 0$ with positive, finite \holder constant, can use Lemma \ref{lem:probholderbounds} and the exponential normalization from Equation \eqref{eq:olmutexpweight} to see that there is a constant $\supholderupperbound{\underline{\theta}}< \infty$ such that  
\begin{dmath} \label{eq:probexpestimate} \int_{\delta}^1 \ol{\mu}_t(dx) \leq \supholderupperbound{\underline{\theta}}  e^{\left[\lambda G(1) - \underline{\theta} \pi(1)\right] t} \exp\left(- \lambda \int_0^t \langle G(\cdot) \rangle_{\mu_s} ds \right) = \supholderupperbound{\underline{\theta}} \exp\left(- \lambda \int_0^t \left[ \langle G(\cdot) \rangle_{\mu_s} - G(0) \right] ds \right). \end{dmath}
Because $G(x) \geq G(0)$ for $x \in [0,1]$, either $\int_0^{\infty} \left[ \langle G(\cdot) \rangle_{\mu_s} - G(0) \right]ds < \infty$ or $ \int_0^{\infty} \left[ \langle G(\cdot) \rangle_{\mu_s} - G(0) \right]ds  \to \infty$ as $t \to \infty$. In the former case, Lemma \ref{lem:averagepayoffconvergence} tells us that $\ol{\mu}_t(dx) \rightharpoonup \delta(x)$ as $t \to \infty$. In the alternate case,  $ \int_0^{\infty}  \left[ \langle G(\cdot) \rangle_{\mu_s} - G(0) \right]ds  \to \infty$ corresponds to $e^{- \int_0^{t} \left[ \langle G(\cdot) \rangle_{\mu_s} - G(0) \right]ds} \to 0$, and we can use Equation \eqref{eq:mainestimateequalitycase} allows us to conclude that
\begin{dmath*} \bigg| \int_0^1 v(x) \ol{\mu}_t(dx) - v(0)   \bigg|   \leq \epsilon \condition{as $t \to \infty$}, \end{dmath*}
and therefore we see that $\ol{\mu}_t(dx) \rightharpoonup \delta(x)$ when $\lambda \left[G(1) - G(0) \right] = \underline{\theta} \pi(1)$. 
\end{proof}

\begin{remark}
A key aspect of the proof of convergence to $\delta(x)$ in Theorem \ref{thm:extinctvspersist} was that the strict inequality condition that $\lambda < \lambda^*$ provided us the freedom to consider a slightly smaller exponent $\tilde{\theta} < \underline{\theta}$ and still obtain exponential decay to $\delta(x)$. In the equality case with $\lambda = \lambda^*$ considered in Proposition \ref{prop:longtimedeltaequality}, we had to assume a finite supremum \holder constant $C_{\underline{\theta}}$ to get a bound involving the exact supremum \holder exponent $\underline{\theta}$.
\end{remark}

 \subsection{Convergence to Full-Cooperation in the Prisoners' Delight}
\label{sec:PDel}

Before proceeding to the proof of Proposition \ref{prop:PDeldelta}, we introduce several lemmas that allow us to estimate the measure $\mu_t(dx)$ solving Equation \eqref{multiselect} under the Prisoners' Delight scenario. These lemmas serve as analogues to Lemma \ref{lem:probholderbounds} and \ref{lem:muGwbound}, and can be proved with a similar approach.

\begin{lemma} \label{lem:PDelintervalbound}
Suppose that $G(x), \pi(x)$ satisfy the assumptions of Proposition \ref{prop:PDeldelta}. For any $a \in (0,1)$, there exists $M_a < \infty$ such that 
\begin{equation} \label{eq:Pdelintervalbound}
   \rho_t\left([0,1]\right) := \mu_t\left((0,a]\right) \leq M_a e^{\lambda G(0) t}. 
\end{equation}
\end{lemma}

\begin{lemma} \label{lem:muGzbound}
Suppose that $G(x), \pi(x)$ satisfy the assumptions of Proposition \ref{prop:PDeldelta}. Considering the quantity $\hat{G}_{a} := \min_{x \in [a,1]} G(x)$, we see that for $z$ sufficiently close to $1$, there exists $A_z > 0$ such that
\begin{equation} \label{eq:muGzbound}
 \mu_t\left([0,1]\right) \geq A_z e^{\lambda \hat{G}_z t}.    
\end{equation}
\end{lemma}

\begin{proof}[Proof of Proposition \ref{prop:PDeldelta}]
To show that $\ol{\mu}_t(dx)$ converges to a delta-function at $x=1$, we consider a continuous test function $v(x)$ and use the fact that $\ol{\mu}_t(dx)$ is a probability distribution to see that, for any $\epsilon > 0$, there is a $\delta > 0$ such that
\begin{dmath} \label{eq:DeltaPDelfirstestimate}
\bigg| \int_0^1 v(x) \ol{\mu}_t(dx) - v(1) \bigg| \leq  \int_0^{1 - \delta} |\left( v(x) - v(1) \right)| \ol{\mu}_t(dx) + \int_{1-\delta}^1 |\left( v(x) - v(1) \right)| \ol{\mu}_t(dx) \leq  \epsilon +  2||v||_{\infty} \int_0^{1 - \delta}  \ol{\mu}_t(dx)   
\end{dmath}
We can use the assumption that $a_0 := \mu_0\left(\{0\}\right) < 1$ to write our initial measure by a decomposition of the form
\begin{equation} \label{eq:deltazerodecomp}
    \mu_0\left(dx\right) = a_0 \delta(x) + \left(1-a_0\right) \rho_0(dx) 
\end{equation}
where $\rho_0(dx)$ is a probability measure satisfying $\rho_0\left(\left(0,1\right] \right) = 1$. 
Because we have assumed that $\mu_0(dx) \neq \delta(x)$, we can apply Lemma \ref{lem:muGzbound} to see that, for $z$ sufficiently close $1$, there is $A_z > 0$ such that $\mu_t\left([0,1]\right) \geq A_z e^{\lambda \hat{G}_z t}$. Combining this with the normalization relation of Equation \eqref{eq:olmutnormalized} and the decomposition of our initial measure from Equation \eqref{eq:deltazerodecomp}, we can estimate that
\begin{equation} \label{eq:PDelsecondestimate}
\int_{0}^{1-\delta} \ol{\mu}_t(dx) = \frac{\int_{0}^{1-\delta} \mu_t(dx)}{\int_{0}^{1} \mu_t(dx)} \leq A_z^{-1} e^{-\lambda \hat{G}_z t} \left[ a_0 e^{\lambda G(0) t} + (1-a_0) \rho_t\left([0,1-\delta]\right)  \right]
\end{equation}
From Lemma \ref{lem:PDelintervalbound}, we know that there is $M_a < \infty$ such that $\rho_t\left([0,1-\delta]\right) \leq M_a e^{\lambda G(0) t}$. Combining this with our estimates from Equation \eqref{eq:DeltaPDelfirstestimate} and \eqref{eq:PDelsecondestimate} allows us to see that there exists $\tilde{M} < \infty$ such that
\begin{equation}
\bigg| \int_0^1 v(x) \ol{\mu}_t(dx) - v(1) \bigg| \leq \epsilon + 2 \tilde{M} ||v||_{\infty} \exp\left(\lambda \left[G(0) - \hat{G}_z \right] t \right).
\end{equation}
Because $G(1) > G(0)$ for the PDel game, we know that we can pick $z$ sufficiently close to $1$ such that $\hat{G}_z > G(0)$. For such choices of $z$, we can deduce that
\[2 \tilde{M} ||v||_{\infty} \exp\left(\lambda \left[G(0) - \hat{G}_z \right] t \right) \to 0 \: \: \mathrm{as} \: \: t \to \infty \]
as long as $\lambda > 0$, and we can conclude that $\ol{\mu}_t(dx) \rightharpoonup \delta(1-x)$ as $t \to \infty$ for $\lambda > 0$ and $\mu_0(dx) \ne \delta(x)$.
\end{proof}

\section{Generalization to Multilevel Competition with \texorpdfstring{$N$}{N} Populations}
\label{sec:multiplepopulations}

In this section, we discuss results for the $N$-population multilevel selection model whose dynamics are described by Equation \eqref{eq:measurevaluedsubpopulation}. We provide the proof of Theorem \ref{thm:longtimemultiple}, demonstrating a sufficient condition for the long-time behavior of the population to feature concentration upon the subpopulation with the maximal principal growth rate. In Section \ref{sec:HDSH}, we apply these results to study the multilevel dynamics for the generalizations of the Hawk-Dove and Stag-Hunt games in a single population by reformulating these results as a two-population problem and characterize the long-time behavior for Equation \eqref{multiselect} for these games.

\begin{proof}[Proof of Theorem \ref{thm:longtimemultiple}]
 Using the normalization from Equation \eqref{eq:mubartj}, we can estimate the $\ol{\mu}_t^j\left([0,1]\right)$ of having groups in subpopulation $j \in \mathcal{N} - \{k\}$ by
 \begin{equation} \label{eq:normalizemultipleratio}
     \ol{\mu}_t^j\left([0,1]\right) =  \frac{\mu_t^j\left([0,1]\right) } {\sum_{i=1}^N \mu_t^i\left([0,1]\right)} \leq \frac{\mu_t^j\left([0,1]\right)}{\mu_t^j\left([0,1]\right) + \mu_t^k\left([0,1]\right)}.
 \end{equation} 
If faced with the relevant case for the principal growth $r_k^m$, we can use the quantities $\check{G}_w = \min_{x\in[0,w]}G(x)$ for $w$ sufficiently close to $0$,  $\hat{G}_z = \min_{x \in [z,1]} G(x)$  for $z$ sufficiently close to $1$, or $\overline{\Theta}^k$ sufficiently close to but greater than $\overline{\theta}^k$ to introduce a modified principal growth rate $\tilde{r}_k^m$ that satisfies 
 \begin{equation}
\tilde{r}_k^m  = \left\{
     \begin{array}{cl}
       \lambda \hat{G}_k^z & : \pi_k(x) < 0\\
       \lambda G_k(1) - \overline{\Theta}^k \pi_k(1) & : \pi_k(x) > 0\: , \: \lambda \left[ G_k(1) - G_k(0) \right] - \overline{\theta}^k \pi_k(1) > 0 \\ 
       \lambda \tilde{G}_k^w & : \pi_k(x) > 0 \: , \: \lambda \left[ G_k(1) - G_k(0) \right] - \overline{\theta}^k \pi_k(1) < 0 \\ 
     \end{array}    \right\} > r_j^M.
\end{equation}
Across the three cases, we can apply Lemmas \ref{lem:probholderbounds}, \ref{lem:muGwbound}, and \ref{lem:muGzbound} to see that there exists $L > 0$ such that $\mu^k_t\left( \left[0,1\right]\right) \geq L e^{\tilde{r}_k^m t}$. We can then combine this with Equation \eqref{eq:normalizemultipleratio} to estimate that
\begin{equation} \label{eq:multipledenominatorunfixed}
\ol{\mu}_t^j\left([0,1]\right) = \frac{e^{-\tilde{r}_k^m}\mu_t^j\left([0,1]\right)}{e^{-\tilde{r}_k^m}\mu_t^j\left([0,1]\right) + e^{-\tilde{r}_k^m}\mu_t^k\left([0,1]\right)} \leq \frac{e^{-\tilde{r}_k^m t}\mu_t^j\left([0,1]\right)}{e^{-\tilde{r}_k^m t}\mu_t^j\left([0,1]\right) + L} \leq L^{-1} e^{-\tilde{r}_k^m t}\mu_t^j\left([0,1]\right). \end{equation}
Our next step is to estimate $\mu_t^j\left([0,1]\right)$. If $\pi_j(x) < 0$ for $x \in [0,1]$, we may estimate that
\begin{equation}
    \int_0^t \left(G(\phi_s(x)) - G(1)\right) ds = \int_x^{\phi_t(x)} \frac{G(q) - G(1)}{q(1-q) |\pi(q)|} dq \leq \int_0^1 \frac{\left[G(q) - G(1)\right]_+}{q(1-q) |\pi(q)|} dq < \infty,
\end{equation}
and therefore there exists an $A_1 < \infty$ such that
\begin{equation}
    \mu_t^j([0,1]) = e^{\lambda G(1) t} \int_0^1 e^{\lambda \int_0^t \left[G(\phi_s(x)) - G(1)\right] ds} \mu_0(dx) \leq A_1 e^{\lambda G(1) t} = A_1 e^{r_j^M t}.
\end{equation}
If instead $\pi_j(x) > 0$ for $x \in [0,1]$, we can use the fact that $\mu_0^j(dx)$ has supremum \holder exponent $\underline{\theta} > 0$ to apply Lemma \ref{lem:integraltermbound}. Therefore we can introduce a modified supremum \holder exponent $\underline{\Theta} < \underline{\theta}$ sufficiently close to $\underline{\Theta}$ and a modified principal growth rate $\tilde{r}_j^M \in (r_j^M,\tilde{r}_k^m)$, which allows us to see that there exist constants $C_1^1,C_2^1,A_2 < \infty$ such that
\begin{equation} \label{eq:subdominantpopestimate}
\mu^j_t\left( \left[0,1\right] \right) \leq C_1^1 e^{\lambda G_j(0) t} + C_2^2 e^{\left[\lambda G_j(1) - \underline{\theta}^j \pi_j(1) \right] t} \leq A_2 e^{\max\left\{\lambda G_j(0),\lambda G_j(1) - \underline{\theta}^j \pi_j(1)\right\} t} \leq A_2 e^{\tilde{r}_j^M t}.
\end{equation}
Between the two cases for the sign of $\pi_j(x)$, we can see that there exists $A < \infty$ and $\tilde{r}_j^M < \tilde{r}^k_m$ such that $\mu^j_t\left( \left[0,1\right] \right) \leq A e^{\tilde{r}_j^M t}$. 
Combining this estimate with Equation \eqref{eq:multipledenominatorunfixed}, we can use the fact that $\tilde{r}_k^m  > \tilde{r}_j^M$ to deduce that
\begin{equation}
    \ol{\mu}_t^j\left([0,1]\right) \leq A L^{-1} \exp\left( \left[\tilde{r}_j^M - \tilde{r}_k^m \right] t \right) \to 0 \: \: \mathrm{as} \: \: t \to \infty,
\end{equation}
and we can conclude that the population will concentrate upon subpopulation $k$ in the long-time limit.

Now we turn to describing the long-time behavior of $\ol{\mu}_t^k(dx)$ in the case for which $\mu_0^k(dx)$ has a well-defined \holder exponent and constant near $x = 1$. We can write describe the distribution of group compositions in subpopulation $k$ as
\begin{equation} \label{eq:mukfraction}
   \int_0^1 v_k(x) \ol{\mu}_t^k(dx) = \frac{\int_0^1 v_j(x) \mu_t^j(dx)}{ \ds\sum_{j=1}^N \mu_t^j \left([0,1] \right)} =  \frac{e^{\left[\theta^k \pi_k(1) - \lambda G_k(1)\right]t}\int_0^1 v_j(x) \mu_t^j(dx)}{e^{\left[\theta^k \pi_k(1) - \lambda G_k(1)\right]t} \mu_t^k\left([0,1]\right) + \ds\sum_{\substack{j = 1 \\j \ne k}}^N \left\{e^{\left[\theta^k \pi_k(1) - \lambda G_k(1)\right]t} \mu_t^j\left([0,1]\right) \right\}}.
\end{equation}
Using the same approach as in the proof of Theorem \ref{wlimtheorem}, we can see that 
\begin{equation}
\left\{
     \begin{array}{cl}
     e^{\left[\theta^k \pi_k(1) - \lambda G_k(1)\right]t}\int_0^1 v_k(x) \mu_t^k(dx) & \to \int_0^1 v_j(x) f^{\lambda}_{\theta^k}(x) dx \\
     e^{\left[\theta^k \pi_k(1) - \lambda G_k(1)\right]t} \mu_t^k\left([0,1]\right) &\to \int_0^1 f^{\lambda}_{\theta^k}(x) dx
    \end{array}  \right\} \: \: \mathrm{as} \: \: t \to \infty.
    \end{equation}
Because $r_k^m = \lambda G_k(1) - \theta^k \pi_k(1)$ in this case, we know from above that $e^{\left[\theta^k \pi_k(1) - \lambda G_k(1)\right]t} \mu_t^j\left([0,1]\right) \to 0$ for $j \ne k$. Applying this to Equation \eqref{eq:mukfraction} allows us to conclude that 
\begin{equation} \label{eq:muklimit}
   \int_0^1 v_k(x) \ol{\mu}_t^k(dx) \to \frac{\int_0^1 v_k(x) f^{\lambda}_{\theta^k}(dx)}{\int_0^1  f^{\lambda}_{\theta^k}(dx)} = \int_0^1 v_k(x) p^{\lambda}_{\theta_k}(dx) \: \: \mathrm{as} \: \: t \to \infty.  \qedhere
   \end{equation}
\end{proof}

\section{Discussion} \label{sec:discussion}

In this paper, we have analyzed the long-time behavior in a PDE model of multilevel selection, in which a tension exists between the individual-level incentive to defect and group-level competition favoring groups that cooperate. We show that defectors take over the group-structured population when within-group competition is stronger than between-group competition, and that cooperation can weakly persist in the population for all time when the relative strength of between-group competition exceeds a threshold value. We also provide sufficient conditions for the population to converge to a long-time steady state density featuring coexistence of cooperators and defectors, and further characterize the average level of cooperation and group-reproduction rate at steady state in the limit of strong between-group competition. These results generalize and extend previous work on PDE models of multilevel selection with within-group and between-group dynamics arising from frequency-independent competition \cite{luo2014unifying,van2014simple,luo2017scaling} or from the payoffs of evolutionary games \cite{cooney2019replicator,cooney2020analysis,cooney2019assortment}.

\myindent By considering arbitrary $C^1$ functions $\pi(x)$ and $G(x)$ to describe the within-group and between-group competition, we have a general analysis to study how long-term cooperation depends on the tug-of-war of between the individual-level and group-level incentives. In this more general setting, we are able to understand the key role the full-cooperator group plays in determining the level of cooperation and collective average payoff supported by the long-time behavior of multilevel dynamics.  In particular, we see that increasing the relative advantage of the full-cooperator group or increasing the initial cohort of many-cooperator groups (corresponding to lower \holder exponent $\theta$) helps to promote the evolution of cooperative behavior via multilevel selection, consistent with analytical and simulation results from finite population models of multilevel selection \cite{markvoort2014computer,traulsen2006evolution,traulsen2008analytical}.

\myindent Considering a broader class of initial conditions for which the infimum and supremum H{\"o}lder exponents or constants disagree also reveals important properties of our multilevel dynamics. In particular, convergence to a steady solution of Equation \eqref{multiselect} in not guaranteed for generic initial probability measures of group compositions, and instead a more natural notion for quantifying the ability for cooperation to survive via multilevel selection is the weak persistence of cooperation for sufficiently strong between-group competition. This distinction between weak persistence and convergence to steady state may also be relevant for exploring multilevel selection for more complex strategy spaces, as it may be more difficult to identify or quantify the possible steady state behaviors beyond a one-dimensional state space for group compositions. As a question for future work, one could look for analogous weak persistence thresholds for PDE models of multilevel selection including additional evolutionary forces like genetic drift or migration \cite{ogura1987stationary,ogura1987stationary2,fontanari2013solvable,fontanari2014effect,fontanari2014nonlinear}, or models in which the assumption of fixed group size is relaxed and group-level events such as fission and fusion can help to drive the evolution of cooperation \cite{simon2010dynamical,simon2012numerical,simon2013towards,simon2016group}.

\myindent  As in the special cases previously studied for evolutionary games \cite{cooney2019replicator,cooney2020analysis}, we establish that the collective payoff of the steady state population is limited by the payoff of a full-cooperator group. This means that the so-called ``shadow of lower-level selection'' is present for all group reproduction functions which are maximized by an intermediate level of cooperation: no level of between-group competition produces a steady-state population that achieves the maximum possible group-reproduction rate.
Even in the limit of infinitely strong between-group competition, the population still concentrates as a delta-function at a level of cooperation that produces the same collective-reproduction rate as that of a (sub-optimal) full-cooperator group. In addition, we have now established a dynamical analogue to the shadow of lower-level selection in terms of the bounds on the time-average of the group-reproduction rate in the population, highlighting that this limitation of the collective outcome to that of a full-cooperator groups can be seen through the dynamics of our model of multilevel selection. Given the bounds on the time-averaged collective outcome, a natural question for future research is whether there is a sense in which the potentially oscillatory long-time solutions of the multilevel dynamics concentrate upon group compositions with the group-reproduction rate $G(1)$ for sufficiently strong between-group competition, extending the concentration behavior seen for steady-state densities. 

\myindent Our analysis of the threshold selection strength and average payoff at steady state also provide a window to understanding how mechanisms that alter within-group and between-group competition may facilitate cooperation. In particular, we see that decreasing the incentive to defect in a many-cooperator group, $\pi(1)$, can help decrease the between-group competition strength needed to allow long-time survival of cooperation, but it cannot increase the maximum possible achievable group-average payoff in the limit of strong between-group competition. Because altering within-group interactions -- through the mechanisms of assortment, other-regarding preference, indirect reciprocity, and network reciprocity -- can increase $\pi(0)$ to a positive value and provide cooperators an individual-level advantage in a many-defector group \cite{grafen1979hawk,smith1982evolution,taylor2007transforming,cooney2019assortment}, we can also see that such within-group mechanisms can change the dynamics of multilevel selection from following the generalized Prisoners' Dilemma assumptions to following the generalized Hawk-Dove or Stag-Hunt assumptions. Through our results on the long-time behavior of the multilevel dynamics under each of each of these generalized games, we can further explore the synergistic effects of within-group population structure and between-group competition on the evolution of cooperation \cite{cooney2019assortment}.

\myindent This paper constitutes an in-depth analysis of a class of hyperbolic PDEs that generalize the equation studied by Luo and Mattingly \cite{luo2014unifying} to include variety of replication functions corresponding to different within-group and between-group competition scenarios. In turn, Equation \eqref{multiselect} corresponds to the two-level replicator equation that arises from a generalization of the two-level Moran process introduced by Luo \cite{luo2014unifying} under one possible scaling in the limit of infinite group size and of an infinite number of groups. Because our PDE is derived from an individual-based model with finite populations, an important consideration is the extent to which the behavior of solutions to the PDE correspond to properties of the underlying finite population dynamics. In particular, an interesting question for future research is the extent to which behaviors like persistence of cooperation or convergence to steady state densities in the PDE corresponds to any analogous support for cooperation in the original stochastic two-level selection model with finite populations or in alternate deterministic scaling limits that retain the diffusive effects of finite group size \cite{luo2014unifying,luo2017scaling,cooney2019replicator}. 

\myindent Our results in Section \ref{sec:multiplepopulations} on multiple population dynamics also provide motivation for future work on the coevolution of cooperative strategies and the games played by members of groups. Because the principle growth rate of a subpopulation often corresponds to the average collective outcome at steady state, we can see the long-time behavior of the multipopulation competition as picking out the pair of gain and group-reproduction functions $\pi_j(x)$ and $G_j(x)$ most capable of producing a beneficial level of cooperation under the two-level dynamics for a given set of initial distributions. Future modeling work can examine how competition between different group types can impact the evolution of group properties such as within-group population structures or social norms \cite{cooney2019assortment,santos2007multi}. In addition, the approach of reformulating multilevel selection models multipopulation models can be extended to evolutionary games with more complicated within-group dynamics, allowing for the analysis of long-time behavior for nonlinear public goods games \cite{archetti2011coexistence,pacheco2009evolutionary} or for models of efficient extraction of common-pool resources \cite{tavoni2012survival,tilman2017maintaining}. Going forward, the study of multilevel selection models with multiple group types can serve as a potential next step for mathematically understanding the role which multilevel selection can play on the cusp of major evolutionary transitions. 

\renewcommand{\abstractname}{Acknowledgments}
\begin{abstract} 
DBC received support from the National Science Foundation through grant DMS-1514606 and the Army Research Office through grant W911NF-18-1-032x5. YM received support from DMS-1907583. Both DBC and YM were supported by the Math+X grant from the Simons Foundation. The authors thank Joshua Plotkin for many fruitful conversations about the model and helpful comments on the manuscript. DBC also thanks Simon Levin and Denis Patterson for helpful discussions. 

\end{abstract}

\bibliographystyle{unsrt}
\bibliography{references}

\appendix

\section{Well-Posedness of Multilevel Dynamics}
\label{sec:wellposedness}

In this section, we address the well-posedness of measure-valued solutions to Equation \eqref{multiselect}.  Our approach will be to first establish well-posedness of solutions $\mu_t(dx)$ to the linear multilevel dynamics of Equation \eqref{eq:linearmeasuremultiselect}, and then to use the exponential normalization from Equation \eqref{eq:olmutexpweight} to demonstrate that Equation \eqref{multiselect} is well-posed as well. 

For the auxiliary linear problem of Equation \eqref{eq:linearmeasuremultiselect}, we can use Equations \eqref{cODE} and \eqref{eq:pushforward} to obtain the following representation formula for its solution $\mu_t(dx)$ starting from an initial measure $\mu_0(dx)$:
\begin{equation} \label{eq:mutpushforwardGt}
    \int_0^1 v(x) \mu_t(dx) = \int_0^1 v(\phi_t(x)) \exp\left( \int_0^t G(\phi_s(x)) ds \right) \mu_0(dx). 
\end{equation}
The righthand Equation \eqref{eq:mutpushforwardGt} is a positive linear functional, so we can deduce that $\mu_t(dx)$ is a Borel measure due to the Riesz-Markov-Kakutani theorem. We can also check from this representation formula to see that $\mu_t(dx)$ solves Equation \eqref{eq:linearmeasuremultiselect} for any test-function $v(x) \in C^1([0,1])$. Existence for any $C([0,1])$ test-function can then be established by density of $C^1([0,1])$ functions in $C([0,1])$. 

To further explore the well-posedness of solutions of the full multilevel dynamics of Equation \eqref{multiselect}, it is helpful to establish a bijective correspondence between the measures $\mu_t(dx)$ and $\ol{\mu}_t(dx)$ solving Equations \eqref{eq:linearmeasuremultiselect} and \eqref{multiselect}. Taking inspiration from the exponential normalization relation from Equation \eqref{eq:olmutexpweight}, we use the following expressions to show that, for any continuous test-function $v(x)$, $\mu_t(dx)$ and $\ol{\mu}_t(dx)$ can be expressed a as functions of the other: expressions
\begin{subequations} \label{eq:expnorms}
\begin{align} 
    \int_0^1 v(x) \ol{\mu}_t(dx) &=  \exp\left(- \lambda \int_0^t \langle G(\cdot) \rangle_{\mu_s} ds  \right) \int_0^1 v(x)  \mu_t(dx). \label{eq:olmutexpweight2}\\
      \int_0^1 v(x)  \mu_t(dx) &= \exp\left( \lambda \int_0^t \langle G(\cdot) \rangle_{\ol{\mu}_s} ds  \right) \int_0^1 v(x)  \ol{\mu}_t(dx). \label{eq:mutexpweight}
\end{align}
\end{subequations}
This correspondence can be established for $C^1([0,1])$ test-functions by applying the push-forward representation of Equation \eqref{eq:pushforward} to the righthand side, and then differentiating both sides with respect to time. Equation \eqref{eq:expnorms} can then be confirmed for $C([0,1])$ test-functions by the density of continuously differentiable functions within the space of continuous functions.

From the normalization relation provided by Equation  \eqref{eq:olmutnormalized}, we can use the existence of the solution $\mu_t(dx)$ to Equation \eqref{eq:linearmeasuremultiselect} to deduce the existence of a Borel measure $\ol{\mu}_t(dx)$ solving Equation \eqref{multiselect} in the weak sense. For uniqueness, we will prove  make use of the mapping from Equation \eqref{eq:mutexpweight} to help deduce uniqueness for Equation \eqref{multiselect} from uniqueness of the linear problem.

It is also important to consider for which space our solutions $\ol{\mu}_t(dx)$ live. From the weak form of the multilevel dynamics presented in Equation \eqref{eq:measuremultiselect}, it natural to consider $C^1([0,1])$ test-functions due to the advection term and to consider a family measures $\mu_t(dx)$ with a differentiable dependence in time due to the time-derivative in the PDE. From the push-forward representation formula of Equation \eqref{eq:pushforward} and the fact that our results on long-time behavior only require test-functions $v(x) \in C^1([0,1])$, it also makes sense to consider uniformly continuous test functions and measures $\mu_t(dx)$ that vary continuously in time. As a result, in Proposition \ref{prop:wellposedness}, we establish well-posedness of solutions $\ol{\mu_t}(dx)$ of Equation \eqref{multiselect} both taking values in the space of Borel measures with uniformly continuous time-dependence and taking values in the dual space of $C^1([0,1])$ functions with continuously differentiable time-dependence.

\begin{proposition} \label{prop:wellposedness}
Denote by $\mc{M}([0,1])$ and $C^1([0,1])^*$ the set of Borel measures on $[0,1]$ and the dual space of $C^1([0,1])$, respectively. Given the initial measure $\mu_0(dx) \in \mc{M}([0,1])$ and any $T \geq 0$, there exists a unique $\ol{\mu}_t(dx) \in C\left(\left[0,T\right]; \mc{M}([0,1]) \right) \cap C^1\left(\left[0,T\right] ;C^1([0,1])^*  \right)$ such that $\ol{\mu}_t(dx)$ is a weak solution to Equation \eqref{multiselect}. 
\end{proposition}
\begin{proof}
 We start by establishing the well-posedness of solutions $\mu_t(dx)$ to the linear model from Equation \eqref{eq:linearmeasuremultiselect}. Knowing from the representation formula of Equation \eqref{eq:pushforward} that $\mu_t(dx)$ exists in the appropriate space, we now suppose that Equation \eqref{eq:linearmeasuremultiselect} has two solutions $\mu_t^1(dx)$ and $\mu_t^2(dx)$ for given initial measure $\mu_0(dx)$. Considering the difference $\tilde{\mu}_t(dx) := \mu_t^1(dx) - \mu_2^1(dx)$ with $\tilde{\mu}_0(dx) = 0dx$. Considering a test-function $v(t,x) \in C^{1,1}\left([0,T] \times [0,1] \right)$ we can use the linearity of Equation \eqref{eq:linearmeasuremultiselect} to see that
 \begin{dmath} \label{eq:timetestfunction}
 \dsdel{}{t} \int_0^1 v(t,x) \tilde{\mu}_t(dx) = \int_0^1 \left[\dsdel{v(t,x)}{t} +  x (1-x) \pi(x) \dsdel{v(t,x)}{x} + \lambda G(x) v(t,x)  \right] \tilde{\mu}_t(dx).
 \end{dmath}
 We can simplify Equation \eqref{eq:timetestfunction} by choosing $v(t,x)$ solve the dual problem to the linear dynamics of Equation \eqref{eq:linearmeasuremultiselect}, given by
 \begin{equation} \label{eq:dualpde}
     \dsdel{v(t,x)}{t} + x (1-x) \pi(x) \dsdel{v(t,x)}{x} = - \lambda G(x) v(t,x) \: \:, \: \: v(t,x) = \omega(x) \in C^1([0,1]).
 \end{equation}
Noting that such a solution $v(t,x)$ exists by the method of characteristics, we can see with this choice that the righthand side of Equation \eqref{eq:timetestfunction} will vanish. This allows us to integrate Equation \eqref{eq:timetestfunction} in time to see that
\begin{equation}
    \int_0^1 v(T,x) \tilde{\mu}_T(dx) = \int_0^1 v(0,x) \tilde{\mu}_0(dx) \Longrightarrow \int_0^1 \omega(x) \tilde{\mu}_T(dx) = 0.
\end{equation}
Because this identity holds for each $\omega(x) \in C^1\left(\left[0,1\right]\right)$, we can deduce that the solution $\mu_t(dx)$ to Equation \eqref{eq:linearmeasuremultiselect} is unique in $C^1([0,1])^*$, and further that it is unique in $\mc{M}([0,1])$ by the density of continuous test-functions in $C^1([0,1])$.  

Existence of a weak solution $\ol{\mu}_t(dx)$ to Equation \eqref{multiselect} in $ C\left(\left[0,T\right]; \mc{M}([0,1]) \right) \cap C^1\left(\left[0,T\right] ;C^1([0,1])^*  \right)$ follows from the existence of a solution $\mu_t(dx)$ to Equation \eqref{eq:linearmeasuremultiselect} and the mapping of Equation \eqref{eq:olmutexpweight}. To establish uniqueness, we consider solutions $\ol{\mu}_t^1(dx)$ and $\ol{\mu}_t^2(dx)$ to Equation \eqref{multiselect} for initial measure $\mu_0(dx)$, and apply the mapping of Equation \eqref{eq:mutexpweight} and the uniqueness of the solution $\mu_t(dx)$ to Equation \eqref{eq:linearmeasuremultiselect} to see that, for any $C^1$ test-function $v(x)$,
\begin{equation}
  \int_0^1 v(x)  \mu_t(dx) =  \int_0^1 v(x) \exp\left( \lambda \int_0^t \langle G(\cdot) \rangle_{\ol{\mu}_s^1} ds  \right) \ol{\mu}_t^1(dx) =  \int_0^1 v(x) \exp\left( \lambda \int_0^t \langle G(\cdot) \rangle_{\ol{\mu}_s^2} ds  \right) \ol{\mu}_t^2(dx).
\end{equation}
We can then use this to deduce that there is a function of time $c(t)$ such that
\begin{equation}
    \int_0^1 v(x) \ol{\mu}_t^1(dx) = c(t) \int_0^1 v(x) \ol{\mu}_t^2(dx).
\end{equation}
Because $\ol{\mu}_t^1(dx)$ and $\ol{\mu}_t^2(dx)$ are probability measures for all $t \geq 0$, we can further deduce that $c(t) \equiv 1$ and conclude that $\ol{\mu}_t^1(dx) = \ol{\mu}_t^2(dx)$ in $C^1\left(\left[0,1\right]\right)^*$ and in $\mathcal{M}\left([0,1]\right)$ by density of continuous test-functions in $C^1([0,1])$.
\end{proof}

\section{Properties of Density Steady State Solutions}
\label{sec:densitysteady}

In this section, we discuss additional properties of steady state solutions to Equation \eqref{multiselect} for the PD case. In Section \ref{sec:possiblesteady}, we characterize the properties of weak steady-state solutions and show distributions can only consist of combinations of delta-concentrations at equilibria of the within-group dynamics and densities supported on non-equilibrium levels of cooperation that are strong solutions to the steady-state ODE. In Section \ref{sec:steadycompetition}, we strengthen this observation by showing that the only steady states that can be achieved as the long-time behavior for solutions of the generalized PD case of Equation \eqref{multiselect} are the delta functions $\delta(x)$ and $\delta(x-1)$ concentrated at the all-defector and all-cooperator equilibria, as well as the family of steady state densities $p^{\lambda}_{\theta}(x)$ described in Section \ref{sec:steadystates}.

In Section \ref{sec:steadypropconcentration}, we prove Propositions \ref{prop:interiordelta} an \ref{prop:deltaedge}, showing that, in the limit of strong between-group competition, the steady state densities $p^{\lambda}_{\theta}(x)$ concentrate upon points at which the group-level replication rate is equal to $G(1)$, the collective replication rate of the all-cooperator group. For the special case of within-group and between-group replication rates arising from the payoff matrix of a PD game in which group payoff is maximized by an intermediate level of cooperation, this means that the population will always feature less cooperation than in collectively optimal. This concentration upon a delta-measure at a suboptimal level of cooperation strengthens observations on the shadow of lower-level selection derived in prior work, in which it was shown that the modal level of cooperation at steady state featured a suboptimal level of cooperation in the limit of infinite between-group selection strength \cite{cooney2019replicator,cooney2020analysis,cooney2021pde}.

\subsection{Possible Steady State Solutions of Multilevel Dynamics}
\label{sec:possiblesteady}

One aspect of the long-term behavior we would like to characterize is the conditions under which solutions to the measure-valued dynamics $\ov{\mu}_t$ converge to weak solutions $\mu(dx)$ of the steady-state ODE given by
\begin{equation} \label{eq:steadymeasure}
     \int_0^1 \dsddx{v(x)}{x} x (1-x) \pi(x) \mu(dx) = \int_0^1 v(x) \left[ G(x) - \int_0^1 G(y) \mu(dy) \right] \mu(dx),  
\end{equation}
and, furthermore, when such time-independent solutions to Equation \ref{eq:steadymeasure} have corresponding densities $f(x)$ that are strong solutions  to the steady-state ODE given by
\begin{equation} \label{eq:steadystrong}
 -\dsdel{}{x} \left[ x(1-x) \pi(x) f(x) \right]  = \lambda f(x) \left[ G(x) - \int_0^1 G(y) f(y) \right]. 
\end{equation}
In Lemma \ref{lem:weaksteadystate}, we show that if solutions $\ol{\mu}_t(dx)$ converges to a limit $\mu_{\infty}(dx)$ as $t \to \infty$, then the limit must be a strong solution to the steady state ODE at all non-equilibrium points of the within-group dynamics. Under the PD dynamics, this means that probability mass can accumulate at the within-group equilibrium at the endpoints $0$ or $1$, and that either the steady state $\mu_{\infty}(dx)$ vanishes identically on $(0,1)$ or is everywhere nonzero on $(0,1)$.

\begin{lemma} \label{lem:weaksteadystate}
If a weak solution $\ol{\mu}_t(dx)$ to Equation \eqref{multiselect} converges weakly to a limit function $\mu_{\infty}(dx)$ as $t \to \infty$, then the limit $\mu_{\infty}(dx)$ is a weak solution to the steady state ODE in Equation \eqref{eq:steadymeasure}. Furthermore, denoting the set of points $E = \{ x \in (0,1) : \pi(x) \neq 0\}$, we see that there $\mu_{\infty}(dx)$ has corresponding density $f_{\infty}(x)$ for $x \in E$, which solves $f_{\infty}(x)$ the strong form of the steady state ODE given in Equation \eqref{eq:steadystrong}. For any interval $\mc{I} \subset E$,  either $f_{\infty}(x) > 0$ for all $x \in \mc{I}$ or $f_{\infty} \equiv 0$ for all $x \in \mc{I}$.
\end{lemma}

\begin{proof}[Proof of Lemma \ref{lem:weaksteadystate}]
Using the measure-valued formulation of the multilevel dynamics, we integrate Equation \eqref{eq:measuremultiselect} in time from $T$ to $T+1$ and obtain 
\begin{dmath} \label{eq:multiweakT}
    \int_0^1 v(x) \ol{\mu}_{T+1}(dx) - \int_0^1 v(x) \ol{\mu}_T(dx) = -\int_{T}^{T+1} \int_0^1 \dsdel{v(x)}{x} x (1-x) \pi(x) \ol{\mu}_t(dx) + \lambda \int_T^{T+1} \int_0^1 v(x) \left[ G(x) - \int_0^1 G(y) \ol{\mu}_t(dy) \right] \ol{\mu}_t(dx).
\end{dmath}

 Because $\ol{\mu}_t(dx)$ converges weakly to $\mu_{\infty}(dx)$ by assumption, we know that, for any $C^1$ test function $V(x)$, $\int_0^1 V(x) \ol{\mu}_t dx \to \int_0^1 V(x) \mu_{\infty}(dx)$ as $ t \to \infty$. In particular, this tells us that the lefthand side of Equation \eqref{eq:multiweakT} vanishes in the limit as $T \to \infty$. Furthermore, because $G(\cdot)$, $\pi(\cdot)$, and $v(\cdot)$ are $C^1$ functions, we see that the following limits hold as $t \to \infty$
 \begin{equation*}
 \begin{aligned}
 \int_0^1 \dsdel{v(x)}{x} x (1-x) \pi(x) \ol{\mu}_t(dx) & \to \int_0^1 \dsdel{v(x)}{x} x (1-x) \pi(x) \mu_{\infty}(dx) \\ 
 \int_0^1 v(x) \left[ G(x) - \int_0^1 G(y) \ol{\mu}_t(dx) \right] \ol{\mu}_t(dx) &\to \int_0^1 v(x) \left[ G(x) - \int_0^1 G(y) \mu_{\infty}(dy) \right] \mu_{\infty}(dx).
 \end{aligned}
 \end{equation*}
 Therefore we can take the limit as $T \to \infty$ on both sides of Equation \eqref{eq:multiweakT}, which allows us to see that the limiting measure $\mu_{\infty}(dx)$ is a weak solution to the steady state ODE from Equation \eqref{eq:steadymeasure}.

 Using the shorthand $Q(dx) = x (1-x) \pi(x) \mu_{\infty}(dx)$, we can rewrite Equation \eqref{eq:steadymeasure} to see that
 \begin{equation} \label{eq:weakq}
    \int_0^1 \dsdel{v(x)}{x} Q(dx)  = \int_0^1 v(x) \lambda  \left[ G(x) - \int_0^1 G(y) \mu_{\infty} (dy) \right] \mu_{\infty}(dx).
 \end{equation}
Because Equation \eqref{eq:weakq} tells us that the distributional derivative of $Q(dx)$ is a finite (signed) measure, we can deduce that there is a function $q(x)$ such that $Q(dx) = q(x) dx$ that has bounded variation, and furthermore that $q(x) \in L^1\left([0,1]\right)$ \cite{evans2015measure}. Then, for any closed interval $\mc{J} \subset E$ such that $\min_{x \in \mc{J}} |\pi(x)| > 0$, we can deduce that there is a density $f_{\infty}(dx)$ such that 
\[f_{\infty}(x)  = \frac{q(x)}{x (1-x) \pi(x)} \in L^1\left(\mc{J}\right). \] Then, restricting ourselves to test-functions with support contained in $\mc{J}$, we can rewrite Equation \eqref{eq:weakq} as
 \begin{equation} \label{eq:weakqdensity}
    \int_{\mc{J}}  \dsdel{v(x)}{x} q(x) dx = \int_{\mc{J}} v(x) \left[G(x) - \int_0^1 G(y) \mu_{\infty}(dy) \right] \mu_{\infty}(dx).
 \end{equation}
Because $f_{\infty}(x)$ is integrable on $\mc{J}$, we see from the righthand side of Equation \eqref{eq:weakqdensity} that $q(x)$ has an integrable weak derivative, and therefore $q(x)$ and $f_{\infty}(x)$ are absolutely continuous. Applying this again to Equation \eqref{eq:weakqdensity} tells us that $q(x)$ has a continuous weak derivative, and therefore $q(x), f_{\infty}(x) \in C^1\left(\mc{J} \right)$. Furthermore, this means that $f_{\infty}(x)$ is actually a strong solution to the steady state ODE of Equation \eqref{eq:steadystrong} for $x \in \mc{J}$. We can then extend our definition of the interval $\mc{J}$ as needed to show that this also holds for any $x \in E$.  
 
To discuss positivity, we can rewrite the the steady-state relation Equation \eqref{eq:steadystrong}  for $f_{\infty}(x)$ as 
\begin{equation} \label{eq:finfODE}
    \dsdel{f_{\infty}(x)}{x} = \left(\frac{1}{x (1-x) \pi(x)}\right) \left( \dsdel{}{x}\left[x(1-x) \pi(x) \right] -\lambda \left[ G(x) - \int_0^1 G(y) f_{\infty}(y) dy \right]  \right) f_{\infty}(x).
    \end{equation}
This ODE has a unique solution on any interval $\mc{J}$ on which $\min_{x \in \mc{J}} |\pi(x)| > 0$, as the righthand side is Lipschitz in $f_{\infty}(x)$ and is continuous in $x$ away from the equilibria of the within-group dynamics. Consequently, if $f_{\infty}(x) = 0$ for an $x \in \mc{J}$, then $\del{f_{\infty}(x)}{x} = 0$ and therefore $f_{\infty}(x)$ is identically $0$ on $\mc{J}$. This allows us to conclude that $f_{\infty}(x)$ is either strictly positive or identically $0$ on any connected component of the set $E$ of points in $(0,1)$ that are not equilibria of the within-group dynamics. 
\end{proof}

 \subsection{Achievable Long-Time Steady States for PD Dynamics}
\label{sec:steadycompetition}

 Lemma \ref{lem:weaksteadystate} tells us that the only possible steady states of Equation \eqref{multiselect} are convex combinations delta-functions supported at equilibria of the within-group equilibria and densities that are strong solutions to the steady state ODE on intervals between within-group equilbria. 
 As a first step to exploring which steady states can actually be achieved through the long-time behavior of Equation \eqref{multiselect} for the PD case, we can consider initial measures of the form
 \begin{equation} \label{eq:initialconvexsteady}
 \mu_0(dx) = a_0 \delta(x) + a_1 \delta(1-x) + \left( 1 - a_0 - a_1 \right) p^{\lambda}_{\theta}(x) dx
 \end{equation}
 for $a_0, a_1 \geq 0$ satisfying $a_0 + a_1 \leq 1$ and $p^{\lambda}_{\theta}(x)$ given by Equation \eqref{pthetactilde}. 
In Proposition \ref{prop:steadydelta}, we characterize the long-time behavior of solutions $\ol{\mu}_t(dx)$ to Equation \eqref{multiselect} for such initial measures, and we see that the only possible long-time steady states are $\delta(x)$, $\delta(1-x)$, and $p^{\lambda}_{\theta}(x)$ if there is any between-group competition (i.e. when $\lambda > 0$).

 \begin{proposition} \label{prop:steadydelta}
 Suppose that $\lambda \left[G(1) - G(0) \right] > \pi(1) \theta$ and that the population has initial measure $\mu_0(dx)$ given by Equation \eqref{eq:initialconvexsteady}. If $a_1 > 0$, then the solution $\mu_t(dx)$ to Equation \eqref{multiselect} satisfies $\ol{\mu}_t(dx) \rightharpoonup \delta(1-x)$ as $t \to \infty$. If $a_1 = 0$ and $a_0 < 1$, then we have instead that $\ol{\mu}_t(dx) \rightharpoonup p^{\lambda}_{\theta}(x)$.
 \end{proposition}

 \begin{proof} 
Noting from Proposition \ref{prop:wellposedness} that solutions $\mu_t(dx)$ to Equation \eqref{eq:linearmeasuremultiselect} are unique, we can check the solution $\mu_t(dx)$ corresponding to the initial measure of Equation \eqref{eq:initialconvexsteady} is given by
  \begin{equation} \label{eq:steadyweighted}
      \mu_t(dx) = a_0 e^{\lambda G(0) t} \delta(x) + a_1 e^{\lambda G(1) t} \delta(1-x) + \left( 1 - a_0 - a_1\right) e^{\left[\lambda G(1) - \theta \pi(1) \right] t} p^{\lambda}_{\theta}(x) dx.
  \end{equation}
We can use Equation \eqref{eq:steadyweighted} and the normalization relation from Equation \eqref{eq:olmutnormalized} to further see that solutions $\ol{\mu}_t(dx)$ to the full multilevel dynamics satisfy
 \begin{equation}
     \ol{\mu}_t(dx) = \frac{e^{-\lambda G(1) t} \mu_t(dx)}{e^{-\lambda G(1) t} \int_0^1 \mu_t(dy)} = \frac{a_0 e^{\lambda \left[G(0) - G(1) \right] t} \delta(x) + a_1 \delta(1-x) + \left( 1 - a_0 - a_1\right) e^{ - \theta \pi(1) t} p^{\lambda}_{\theta}(x) dx}{a_0 e^{\lambda\left[ G(0) - G(1) \right] t } + a_1 + \left( 1 - a_0 - a_1\right) e^{ - \theta \pi(1) t}}.
 \end{equation}
 Then, using the fact that that $G(1) > G(0)$ and $\theta \pi(1) > 0$, we can further see that $\ol{\mu}_t(dx) \rightharpoonup \delta(1-x)$ as long as $a_1 > 0$. If, instead, $a_1 = 0$, we can see that 
 \begin{equation} \label{eq:convexzerointerior} 
  \ol{\mu}_t(dx) = \frac{e^{\left[ \theta \pi(1) - \lambda G(1) \right] t } \mu_t(dx) }{e^{\left[ \theta \pi(1) - \lambda G(1) \right] t } \int_0^1 \mu_t(dy)} = \frac{a_0 e^{\left(\theta \pi(1) - \lambda \left[ G(1) - G(0)\right] \right) t } \delta(x) + \left( 1 - a_0 \right) p^{\lambda}_{\theta}(x) dx}{a_0 e^{\left(\theta \pi(1) - \lambda \left[ G(1) - G(0)\right] \right) t } + \left( 1 - a_0 \right) }.
 \end{equation}
 Because $\lambda \left[ G(1) - G(0) \right] > \theta \pi(1)$ by assumption, we can use the expression in Equation \eqref{eq:convexzerointerior} to conclude that $\ol{\mu}_t(dx) \rightharpoonup p^{\lambda}_{\theta}(x) dx$ as $ t \to \infty$.
 \end{proof}

\subsection{Concentration of Steady-State Densities in the Limit of Infinite Between-Group Competition}
     \label{sec:steadypropconcentration}

In this section, we study the behavior of the steady state densities $p^{\lambda}_{\theta}(x)$ in the limit of infinite strength of between-group competition. We characterize the concentration of the steady state densities to measures supported upon values $x$ at which $G(x) = G(1)$, and therefore concentration at levels of cooperation that are not necessarily optimal for group-level replication.

To highlight the dependence of our steady-state density on $\lambda$, we can rewrite the expression of our steady state $p^{\lambda}_{\theta}(x)$ from Equation \eqref{pthetactilde} in the form 
 \begin{equation} \label{eq:steadyLaplace}
 p^{\lambda}_{\theta}(x) = \frac{b(x) \exp\left( \lambda h(x)\right) }{ \int_0^1 b(y) \exp\left( \lambda h(y) \right) dy },
 \end{equation}
 where $b(x)$ and $h(x)$ are given by the formulas 
 \begin{subequations} \label{eq:Laplacefunctions}
 \begin{align} 
 b(x) &:=  \left( 1 - x\right)^{\theta - 1} \frac{\pi(1)}{ x \pi(x)} \exp \left(\int_x^1 \frac{\theta \left[ \pi(1) - s \pi(s)\right]}{s \pi(s)} ds \right) \\
 h(x) &:= \int_x^1 \left[\frac{G(s) - G(1)}{s(1-s) \pi(s)}\right] ds.
 \end{align}
 \end{subequations}
 From the form of Equation \eqref{eq:steadyLaplace},  we expect $p^{\lambda}_{\theta}(x)$ to concentrate around the global maximizer of $h(x)$ as $\lambda \to \infty$. The critical points $x_c$ of $h(x)$ satisfy $G(x_c) = G(1)$, and we further compute that 
\[h''(x) \bigg|_{x = x_c} = \frac{-x(1-x) \pi(x) G'(x) - \left(x(1-x) \pi(x)\right)' \left[ G(1) - G(x) \right]}{x^2 (1-x)^2 \pi(x)^2} \bigg|_{x = x_c} = \frac{-G'(x_c)}{x_c (1-x_c) \pi(x_c)}.  \]
Because $\pi(x) > 0$ for $x \in [0,1]$ under the PD assumptions, we see that the local maxima of $h(x)$ are upcrossings of $G(1)$, and the local minima of $h(x)$ are downcrossings of $G(1)$. 

\begin{example} \label{rem:peakupcross}
As an example of the possible collective optima and upcrossings of $G(1)$ that can occur in our multilevel dynamics, we can turn to the example of the game-theoretic model of Section \eqref{sec:gamemotivation}. We recall from Equation \eqref{eq:gamepiG} that, for dynamics arising from games with the payoff matrix of Equation \eqref{eq:payoffmatrix}, that the group-level replication rate $G(x)$ is given by the quadratic function
\[G(x) = P + (S+T-2P)x + (R-S-T+P) x^2. \]
In this case, $G(x)$ is maximized by the following level of cooperation
\begin{equation} \label{eq:piecewisemax}
    x^* = \left\{
     \begin{array}{cl}
            \ds\frac{S+T-2P}{-2(R - S - T + P)} & : 2R < S + T, R - S - T + P < 0 \vspace{2mm} \\
       1 & : \mathrm{otherwise} 
     \end{array}
   \right.,
\end{equation}
so an intermediate level of cooperation can optimize group-level reproduction when the total payoff $S+T$ generated by the interaction of a cooperator and defector exceeds the total payoff $2R$ generated by two cooperators. In addition, we see from the fact that $G(x)$ is a quadratic function of $x$ in the game-theoretic setting that $G'(x)$ changes sign at most once in $[0,1]$, and that $G(x)$ experiences a single upcrossing of $G(1)$ in $[0,1]$.
This upcrossing occurs at the level of cooperation $\overline{x}$ given by
\begin{equation} \label{eq:piecewiseupcrossing}
    \overline{x} = \left\{
     \begin{array}{cl}
            \ds\frac{R-P}{-(R - S - T + P)} & : 2R < S + T, R - S - T + P < 0 \vspace{2mm} \\
       1 & : \mathrm{otherwise} 
     \end{array}
   \right.,
\end{equation}
so the upcrossing occurs in the interior for the same conditions in which the collective optimum $x^*$ is achieved by an intermediate level of cooperation. By comparing Equation \eqref{eq:piecewisemax} and \eqref{eq:piecewiseupcrossing}, we see that $\overline{x} < x^*$ whenever the collective optimum features a mix of cooperators and defectors.

For the  generalization of the PD and HD dynamics studied in this paper with continuously differentiable replication rates, $G(x)$ can either be maximized by interior levels of cooperation or by full-cooperation and can feature arbitrarily many upcrossings of $G(1)$. As a result, in the broader class of models, it is possible $h(x)$ to be maximized at $x=1$ even when $G(x)$ is maximized at an interior level of cooperation.
\end{example}

In Proposition \ref{prop:interiordelta}, we consider the case in which $h(x)$ is maximized by an interior level of cooperation $\overline{x} \in (0,1)$, and use the Laplace integration method \cite{bender1999advanced} to show that $p^{\lambda}_{\theta}(x)$ concentrates upon $\overline{x}$ as $\lambda \to \infty$. 
 
 \begin{proposition} \label{prop:interiordelta} 
Consider the steady state densities given by Equation \eqref{eq:Laplacefunctions} under the assumptions that $G(x), \pi(x) \in C^1\left([0,1]\right)$, $G(1) > G(0)$ and $\pi(x) > 0$ for $x \in [0,1]$. Suppose that $G(x)$ has an interior maximizer $x^*$ satisfying $G(x^*) > G(1)$, and that $h(x)$ has a unique maximizer $\overline{x} < 1$. Then the family of steady states $p^{\lambda}_{\theta}(x) \rightharpoonup \delta(x - \overline{x})$ as $\lambda \to \infty$. 
\end{proposition}

Furthermore, if $\overline{x}$ is the only upcrossing of $G(1)$ in $[0,1]$, then $ \overline{x} < x^*$ and the level of cooperation achieved as $\lambda \to \infty$ is less than the level that achieves the maximal group reproduction rate $G(x^*)$. In particular, we note from Example \ref{rem:peakupcross} that this is true for all PD and HD games with corresponding quadratic $G(x)$ given by Equation \eqref{eq:gamepiG}.

\begin{remark}
If, in addition to assumptions of Proposition \ref{prop:interiordelta}, $G(x)$ has a unique upcrossing $\overline{x} < 1$ of $G(1)$, then $\overline{x} < x^*$. Because $p^{\lambda}_{\theta}(x) \rightharpoonup \delta(x-\overline{x})$ as $\lambda \to \infty$, this means that the level of cooperation achieved as steady state as $\lambda \to \infty$ will be less than the optimal level of cooperation when $x^* < 1$ and $\overline{x}$ is the unique upcrossing of $G(1)$. 
\end{remark}

When $h(x)$ is maximized by $\ol{x} = 1$, the Laplace method breaks down because of the behavior of $b(x)$ at $x = 1$, so we employ a different approach to show that $p^{\lambda}_{\theta}(x)$ concentrates upon full-cooperation as $\lambda \to \infty$ under the additional assumption that $\theta > 1$. 

 \begin{proposition} \label{prop:deltaedge}
 Consider the steady state densities given by Equation \eqref{eq:Laplacefunctions} under the assumptions on $G(x)$ and $\pi(x)$ from Proposition \ref{prop:interiordelta}. Suppose $h(x)$ achieves a unique maximum at $x = 1$ and consider $\theta > 1$. Then, for our family of steady-state solutions $p^{\lambda}_{\theta}(x) \rightharpoonup \delta(x-1)$ as $\lambda \to \infty$.
 \end{proposition}
 In particular, $h(x)$ has a global maximum at $x=1$ when $G(x)$ is non-decreasing on $[0,1]$, and therefore the population concentrates  upon the optimal level of cooperation as $\lambda \to \infty$ when full-cooperation is the best possible group composition. However, there exist group reproduction functions $G(x)$ for which$h(x)$ is still maximized by full-cooperation even though intermediate levels of cooperation are collectively-optimal. In such cases,  Proposition \ref{prop:deltaedge} tells us the population can concentrate at a level of cooperation greater than is optimal for group-level reproduction, so the %
 shadow of lower-level selection can also manifest itself by promoting too much cooperation.

\begin{proof}[Proof of Proposition \ref{prop:interiordelta}]
 We start by integrating both sides of Equation \ref{eq:steadyLaplace} against a test-function $v(x)$, obtaining 
\begin{equation*} 
\int_0^1 v(x) p^{\lambda}_{\theta}(x) dx = \int_0^1 v(x)  \frac{b(x) \exp\left( \lambda h(x)\right) }{ \int_0^1 b(y) \exp\left( \lambda h(y) \right) dy } dx = \frac{\int_0^1 v(x) b(x) \exp\left( \lambda h(x) \right) dx}{ \int_0^1 b(x) \exp\left(\lambda h(x) \right) dx }.
\end{equation*}
We can further rearrange this expression to obtain 
\begin{dmath} \label{eq:Laplaceratiorearranged}
\int_0^1 v(x) p^{\lambda}_{\theta}(x) dx  
= v(\overline{x}) \left( \frac{\sqrt{\frac{2 \pi}{\lambda |h''(\overline{x})| }}b(\overline{x}) \exp\left( \lambda h(\overline{x})\right)}{\int_0^1 b(y) \exp\left(\lambda h(y) \right) dy} \right) \left( \frac{\int_0^1 v(y) b(y) \exp\left( \lambda h(y) \right) dy}{v(\overline{x}) \sqrt{\frac{2 \pi}{\lambda |h''(\overline{x})| }}b(\overline{x}) \exp\left( \lambda h(\overline{x})\right)}\right)
\end{dmath}
Because $b(x)$ is continuous and nonzero at $\overline{x} < 1$, we can use the interior critical-point case of the Laplace method \cite{bender1999advanced} to obtain the following asymptotic formulas  for the two terms in parenthesis 
\begin{align}
    \ds\lim_{\lambda \to \infty} \left[ \frac{\sqrt{\frac{2 \pi}{\lambda |h''(\overline{x})| }}b(\overline{x}) \exp\left( \lambda h(\overline{x})\right)}{\int_0^1 b(y) \exp\left(\lambda h(y) \right) dy} \right] &= 1 \: \: \mathrm{and} \: \:
    \ds\lim_{\lambda \to \infty}  \left[  \frac{ \sqrt{\frac{2 \pi}{\lambda |h''(\overline{x})| }}v(\overline{x}) b(\overline{x}) \exp\left( \lambda h(\overline{x})\right)}{\int_0^1 v(y) b(y) \exp\left(\lambda h(y) \right) dy} \right] = 1. \label{eq:asymptotics}
\end{align}
Taking the limit as $\lambda \to \infty$ in Equation \eqref{eq:Laplaceratiorearranged}, we can apply the asymptotic formulas from Equation \eqref{eq:asymptotics} to see that
\begin{equation*}
\ds\lim_{\lambda \to \infty} \int_0^1 v(x) p^{\lambda}_{\theta}(x) dx = v(\overline{x}),
\end{equation*}
and we conclude that $p_{\theta}^{\lambda}(x) \rightharpoonup \delta(x - \ol{x})$ as $\lambda \to \infty$. 
\end{proof}

\begin{proof} [Proof of Proposition \ref{prop:deltaedge}]
 We assume, for contradiction, that there is a test-function $v(x) \in C^1$ such that $\int_0^1 v(x) p^{\lambda}_{\theta}(x) dx \not\to v(1)$. In that case, there exists an $\epsilon > 0$ such that for any $\Lambda > 0$, there is a $\lambda > \Lambda$ for which
 \begin{equation}\label{eq:lambdanotconverge}
    \epsilon <  \bigg| \int_0^1 \psi(x) p^{\lambda}_{\theta}(x) dx - v(1) \bigg| \leq \int_0^1 |v(x) - v(1) | p^{\lambda}_{\theta}(x) dx,
 \end{equation}
 where the second inequality follows because $p^{\lambda}_{\theta}(x)$ is a probability density. Because $v(\cdot)$ is continuous, there is $ \delta > 0$ such that $|v(x) - v(1)| < \frac{\epsilon}{2}$ for $x \in [1-\delta,1]$, so we can further estimate that %
 \begin{equation} \label{eq:lambdanotconvergesecondestimate}
   \bigg| \int_0^1 v(x) p^{\lambda}_{\theta}(x) dx - v(1) \bigg| \leq 
   \frac{\epsilon}{2} + \int_{0}^{1-\delta} |v(x) - v(1) | p^{\lambda}_{\theta}(x) dx   < \frac{\epsilon}{2} + 2 ||v||_{\infty} \int_0^{1 - \delta} p^{\lambda}_{\theta}(x) dx .
 \end{equation}
 Combining the results of Equations \ref{eq:lambdanotconverge} and \ref{eq:lambdanotconvergesecondestimate}, we see that there are $\epsilon,\delta > 0$ such that for any $\Lambda > 0$, there exists $\lambda > \Lambda$ for which 
 \begin{equation} \label{eq:probless1minusdelta}
     \int_{0}^{1-\delta} p^{\lambda}_{\theta}(x) dx > \frac{\epsilon}{4 ||v||_{\infty}} > 0.
 \end{equation}
 For such $\lambda$, we can then consider the steady state probability $\int_0^{1-\delta} p^{\lambda}_{\theta}(x) dx$ found on the interval $[0,1-\delta]$. Using Equation \eqref{eq:steadyLaplace} and the fact that $b(x) \geq 0$, we can estimate that, for any $\Delta > 0$,
 \begin{equation} \label{eq:deltaintervalinequality}
     \int_0^{1-\delta} p^{\lambda}_{\theta}(x) = \frac{\int_0^{1-\delta} b(x) e^{\lambda h(x)} dx}{\int_0^{1} b(x) e^{\lambda h(x)} dx} \leq  \frac{\int_0^{1-\delta} b(x) e^{\lambda h(x)} dx}{\int_{1-\Delta}^{1} b(x) e^{\lambda h(x)} dx}.
 \end{equation}
 Because we have assumed that $h(x)$ has a global maximum at $x = 1$ and that $G(x)$ is a $C^1$ function, we know that $h(x)$ is locally non-decreasing as $x \to 1^-$. Therefore, choosing sufficiently small $\Delta > 0$, we see that there is an $A > 0$ such that  
 \begin{equation} \label{eq:denominatorintegral}
     \int_{1-\Delta}^1 b(x) e^{\lambda h(x)} dx \geq e^{\lambda h(1-\Delta)} \int_{1-\Delta}^1 b(x) dx \geq A  e^{\lambda h(1-\Delta)}
 \end{equation}
 Turning to the numerator of Equation \ref{eq:deltaintervalinequality}, we can use the fact that $(1-x)^{\theta - 1} \leq 1$ for $\theta > 1$ to see that the integrand in the numerator satisfies
 \begin{dmath} \label{eq:numeratorintegrandestimate}
 b(x) e^{\lambda h(x)} \leq \exp\left(\lambda h(x) + \theta \int_x^1 \frac{\left[ \pi(1) - s \pi(s) \right]}{s \pi(s)} ds \right) 
 = \exp\left( \int_x^1 \left\{ \frac{1}{s \pi(s)} \left[\frac{\lambda [G(s) - G(1)] + \theta [\pi(1) - s \pi(s)]}{1 - s} \right]  \right\} ds \right).
 \end{dmath}
 In particular, we note that when $s = 0$, the term in square brackets takes the value $\lambda \left[ G(0) - G(1) \right] + \theta \pi(1)$, which is negative when steady state densities of the form $p^{\lambda}_{\theta}(x)$ exist. Because $G(x), \pi(x) \in C^1[0,1]$, the integrand on the last line of Equation \eqref{eq:numeratorintegrandestimate} is negative for $x$ close enough to $0$, and therefore there is $d$ close enough to $1$ and $M < \infty$ such that, for $x \in [0,1-\delta]$, 
 \begin{dmath} \label{eq:numeratorintegral}
 b(x) e^{\lambda h(x)} \leq \exp\left( h(\max(x,1-d)) + \theta \int_{\max(x,1-d)}^1 \frac{\left[ \pi(1) - s \pi(s) \right]}{s \pi(s)} ds  \right) \leq M \exp\left( \max_{x \in [1-d,1-\delta]} h(x) \right).
 \end{dmath}
Further denoting $\hat{h}_{d,\delta} = \max_{x \in [1-d,1-\delta]} h(x)$, we can combine the estimates from Equations \eqref{eq:deltaintervalinequality}, \eqref{eq:denominatorintegral}, and \eqref{eq:numeratorintegral} to see that 
\begin{equation} \label{eq:prob01minusdeltaestimate}
\int_0^{1-\delta} p^{\lambda}_{\theta}(x) dx \leq M A^{-1} (1-\delta) \exp\left( \lambda \left[\hat{h}_{d,\delta} - h(1 - \Delta) \right] \right).
\end{equation}
Because $h(x)$ has a unique global maximum at $x = 1$, there is a $\Delta^* > 0$ such that, for any $\Delta < \Delta^*$, $h(1 - \Delta) > h(x)$ for $x \in [0,1-\Delta^*)$. Namely, choosing $\Delta < \min\left( \Delta^*,\delta\right)$ allows us to additionally deduce that $h(1-\Delta) > \hat{h}_{d,\delta}$ and conclude that 
\begin{equation}
    \int_0^{1-\delta} p^{\lambda}_{\theta}(x) dx \to 0 \: \: \mathrm{as} \: \: \lambda \to \infty,
\end{equation}
contradicting the existence of a subsequence of $\lambda$-values tending to infinity on which the inequality of Equation \eqref{eq:probless1minusdelta} holds. Therefore we can conclude that $p^{\lambda}_{\theta}(x) \rightharpoonup \delta(x-1)$ as $\lambda \to \infty$.
\end{proof}

\section{Application of Multipopulation Framework to Analyze Generalization of Hawk-Dove and Stag-Hunt Dynamics}
  \label{sec:HDSH}
 
 In this section, we consider the long-time behavior of solutions to Equation \eqref{multiselect} when the individual-level and group-level replication rates resemble the dynamics of Hawk-Dove (HD) or Stag-Hunt (SH) games. We do this by decomposing the distribution of groups playing these games into two subpopulations featuring levels of cooperation above and below the interior within-group equilibrium $x_{eq}$, and then study the evolution of these two conditional distributions using a two-population version of the multipopulation dynamics studied in Section \ref{sec:multiplepopulations}. In Section \eqref{sec:SHHDassumptions}, we formulate the individual and group replication rates for the two subpopulations considered for the HD and SH games, and, in Section \ref{sec:HDresults}, we present the results for the long-time behavior of the multilevel HD and SH dynamics in light of Theorem \ref{thm:longtimemultiple}.

  \subsection{Assumptions HD and SH Dynamics and Formulation as a 2-Population Scenario}
  \label{sec:SHHDassumptions}

To generalize the multilevel Hawk-Dove and Stag-Hunt dynamics, we adapt our assumptions on our replication rates $\pi(x)$ and $G(x)$ to reflect the properties of the payoff rankings for these two games. In particular, we see from Equations \ref{eq:characteristicsreplicatorgame} and \ref{eq:interiorequilibrium} that $\pi(x) > 0$ for $x \in (x_{eq},1]$ and $\pi(x) < 0$ for $x \in [0,x_{eq})$ for HD games, while the opposite signs hold for the SH game. 

For the group payoff function, we can use the rankings of payoffs for the HD and SH games and Equation \ref{eq:interiorequilibrium} to see that
\begin{equation} \label{eq:Gxeqinequality}
\begin{aligned}
    G(1) - G(x_{eq}) &= \frac{\left(R - S\right) \left(R -T\right)}{R-S-T+P} > 0 \\ 
    G(0) - G(x_{eq}) &= \frac{\left(S-P\right)\left(T-P\right)}{R-S-T+P} < 0,
    \end{aligned}
\end{equation}
and therefore we the group reproduction function satisfies $G(0) < G(x_{eq}) < G(1)$ for both games.

To reformulate the dynamics of the HD and SH games in terms of the multipopulation dynamics studied in Section \ref{sec:multiplepopulations},  
it is helpful to try to understand ODE of Equation \eqref{eq:replicatorcharacteristics} in comparison to a logistic ODE supported on one of the intervals $[0,x_{eq})$ or $(x_{eq},1]$. We do this by rewriting Equation \eqref{eq:replicatorcharacteristics} in the form
\begin{equation} \label{eq:replicatorcharacteristicsxeq}
    \dsddt{x(t)} = - x (1 - x) \left( x - x_{eq} \right) \left(\frac{\pi(x)}{x - x_{eq}} \right) \: \:, \: \: x(0) = x_0.
\end{equation}
Next, we can map the intervals $[0,x_{eq})$ or $(x_{eq},1]$ into $[0,1]$ using the rescaled variables 
\begin{equation} X_1 = \frac{x}{x_{eq}}\: \:. \: \: X_2 = \frac{x - x_{eq}}{1 - x_{eq}}. \end{equation} 
Then we can describe the respective within-group dynamics below and above the equilibrium $x_{eq}$ using the modified gain functions
\begin{equation} \label{eq:Pifunctions}
\begin{aligned}
    \Pi_1(X_1) &:= \left\{
     \begin{array}{ll}
        \ds\frac{(1-x_{eq} X_1) \pi(x_{eq} X_1)}{1 - X_1} & : X_1 \in [0,1) \vspace{2mm} \\
        -x_{eq} (1-x_{eq}) \pi'(x_{eq}) & : X_1 = 1
     \end{array}
   \right.  \\
   \Pi_2(X_2) &:= \left\{
     \begin{array}{ll}
        \ds\frac{\left(x_{eq} + (1-x_{eq}) X_2 \right) \pi\left(x_{eq} + (1-x_{eq}) X_2\right)}{X_2} & : X_2 \in (0,1] \vspace{2mm} \\
        x_{eq} (1-x_{eq}) \pi'(x_{eq}) & : X_2 = 0
     \end{array}
   \right. .
   \end{aligned} 
\end{equation}
Further introducing the measures $\ol{\mu}_t^1(dX_1) = 1_{X_1 \in [0,1]}\ol{\mu}_t(dX_1)$ and $\ol{\mu}_t^2(dX_2) = 1_{X_2 \in [0,1]}\ol{\mu}_t(dX_2)$ as well as the modified group-reproduction functions
\begin{subequations}
\label{eq:Gfunctions}
\begin{align}
G_1(X_1) &= G(x_{eq} X_1) \\
G_2(X_2) &= G\left( x_{eq} + (1 - x_{eq}) X_2 \right),
\end{align}
\end{subequations}
we can see that these measures $\ol{\mu}_t^1(dX_1)$ and $\ol{\mu}_t^2(dX_2)$ evolve according to Equation \eqref{eq:measurevaluedsubpopulation} for our choices of modified gain and group-reproduction functions. To guarantee that modified gain functions $\Pi_1(X_1), \Pi_2(X_2) \in C^1([0,1])$ in line with the assumption of Theorem \ref{thm:longtimemultiple}, we see from Equation \eqref{eq:Pifunctions} that it suffices to assume that the original gain function satisfies $\pi(x) \in C^2([0,1])$.

\subsection{Summary of Long-Time Behavior for HD and SH Dynamics} 
\label{sec:HDresults}

We start with the case of the Stag-Hunt game. Because $\pi(x) > 0$ for $x \in (0,x_{eq})$ and $\pi(x) < 0$ for $x \in (x_{eq},1)$ for the SH game, we see that $r_2^m = \lambda G_2(1) = \lambda G(1)$, while $r_1^M \leq \lambda G_1(1) = G(x_{eq})$. Therefore the two-population representation of the SH game will always satisfy the hypothesis of Theorem \ref{thm:longtimemultiple} regarding a dominant subpopulation, and so the support groups below $x_{eq}$ vanishes in the long-time limit. Because the dynamics above $x_{eq}$ resemble a PDel game, we can use an approach inspired by the proof of Proposition \ref{prop:PDeldelta} to show fixation upon full-cooperation.
\begin{proposition} \label{prop:SH}
Suppose that $G(x) \in C^1([0,1])$, $\pi(x) \in C^2([0,1])$, $G(0) < G(x_{eq}) < G(1)$, $\pi(x)$ has a single root $x_{eq} \in (0,1)$, and 
that $\pi'(x_{eq}) > 0$. We further assume that the initial distribution contains groups with levels of cooperation exceeding the within-group equilibrium, i.e. $\mu_0\left(\left(x_{eq},1\right]\right) > 0$. If $\lambda > 0$, $\ol{\mu}_t(dx) \rightharpoonup \delta(1-x)$ as $t \to \infty$.
\end{proposition}
 This generalizes a previous result for SH games with the payoff matrices of Equation \eqref{eq:payoffmatrix} \cite{cooney2020analysis}, whose proof relied on the fact that the quadratic $G(x)$ of Equation \eqref{eq:Gxgame} is increasing for $x > x_{eq}$ under the SH payoff rankings. By contrast, Proposition \ref{prop:SH} only requires the ranking of the values group-reproduction function at the equilibria of the within-group dynamics.

\myindent We now turn to the Hawk-Dove game and consider initial measures with infimum and supremum \holder exponents near $x=1$ satisfying $\infty > \overline{\theta} \geq \underline{\theta} > 0$. Because $\pi(x) < 0$ for $x \in (0,x_{eq})$ and $\pi(x) > 0$ for $x \in (x_{eq},1)$, we can find that the principal growth rates on our two intervals are given by $r_1^M = \lambda G(x_{eq})t$ and $r_2^m = \max\{\lambda G(x_{eq}), \lambda G(1) - \overline{\theta} \pi(1) \}$. We can show for either possible value of $r_2^m$ that the probability $\ol{\mu}_t\left([0,x_{eq})\right)$ vanishes as $t \to \infty$, with the case of $r_2^M = \lambda G(1) - \overline{\theta} \pi(1)$ following from Theorem \ref{thm:longtimemultiple} and the case of $r_2^M = \lambda G(x_{eq})$ requiring an argument analogous to the proof of Proposition \ref{prop:PDeldelta}. We can further characterize the long-time behavior of the multilevel HD dynamics by analyzing the distribution of groups above $x_{eq}$ in the same way we studied the PD dynamics in previous sections. %

We show in Theorem \ref{prop:deltaHD} that the population concentrates upon a delta-function $\delta(x - x_{eq})$  at the within-group equilibrium when $\lambda \left[G(1) - G(x_{eq})\right] < \underline{\theta} \pi(1)$, while we show in Theorem \ref{prop:HDpersistence} that the fraction of cooperators in the population exceeds the equilibrium level $x_{eq}$ infinitely often when $\lambda \left[G(1) - G(x_{eq})\right] > \overline{\theta} \pi(1) > \underline{\theta} \pi(1)$. Under the stronger assumption that the initial measure has a well-defined positive, finite \holder exponent and constant near $x=1$, we show in Theorem \ref{wHDlimtheorem} that the population converges to a steady state density whose support consists of groups with fractions of cooperation between $x_{eq}$ and $1$.

 \begin{theorem} \label{prop:deltaHD}
Suppose that $G(x) \in C^1([0,1])$, $\pi(x) \in C^2([0,1])$, $G(0) < G(x_{eq}) < G(1)$. We further assume that the initial distribution $\mu_0(dx)$ has supremum \holder exponent $\theta$ near $x = 1$, that $\pi(x)$ has a single root $x_{eq} \in (0,1)$, and 
that $\pi'(x_{eq}) > 0$. If $\lambda \left( G(1) - G(x_{eq})\right) < \theta \pi(1)$, then $\mu_t(dx) \rightharpoonup \delta(x - x_{eq})$.
 \end{theorem}

\begin{theorem} \label{prop:HDpersistence}
Suppose that $G(x)$ and $\pi(x)$ satisfy the assumptions of Theorem \ref{prop:deltaHD} and that the initial distribution $\overline{\mu}_0(dx)$ has nonzero infimum and supremum H{\"o}lder exponents $\overline{\theta}$ and $\underline{\theta}$ near $x=1$ with corresponding \holder constants $C_{\overline{\theta}}$ and $C_{\underline{\theta}}$ that are finite and nonzero. If $\lambda \left[G(1) - G(x_{eq}) \right] > \overline{\theta} \pi(1) > \underline{\theta} \pi(1)$, then
\begin{equation}
    \ds\limsup_{t \to \infty} \ds \int_{x_{eq}}^1 x \ol{\mu}_t(dx) > x_{eq}.
\end{equation}
\end{theorem}  
  
Now we consider the case of convergence to steady state for the multilevel HD dynamics. When the initial measure has H{\"o}lder exponent $\theta > 0$ near $x=1$ and the strength of between-group satisfies $\lambda \left[ G(1) - G(x_{eq})\right] > \theta \pi(1)$, the principal growth rates $r_1^M$ and $r_2^m$ associated with the measures $\ol{\mu}^1_t(dX_1)$ and $\ol{\mu}^2_t(dX_2)$ satisfying $r_2^m = \lambda G(1) - \theta \pi(1) > \lambda G(x_{eq}) = r_1^M$. As a consequence, we can apply Theorem \ref{thm:longtimemultiple} to conclude that the probably of groups below $x_{eq}$ satisfies $\ol{\mu}^1_t([0,1]) \to 0$ as $t \to \infty$. For the distribution of groups above $x_{eq}$, Theorem \ref{thm:longtimemultiple} tells us that $\ol{\mu}_t^2(dX_2) \rightharpoonup p_{\theta}^{\lambda,2}(X_2)$, where
$p_{\theta}^{\lambda,2}(X_2)$ is a steady state density given by
\begin{subequations} \label{eq:HDsteady2part}
\begin{align}
p_{\theta}^{\lambda,2}(X_2) &= \frac{f^{\lambda,2}_{\theta}(X_2)}{\int_0^1 f^{\lambda,2}_{\theta}(y) dy} \\
f^{\lambda,2}_{\theta}(X_2) &= X_2^{\nu_2 -1} (1-X_2)^{\theta - 1} \left(\frac{\Pi_2(1)}{ \Pi_2(X_2)}\right) \exp\left( \int_{X_2}^1 \frac{-\lambda \tilde{C}_2(s)}{\Pi_2(s)} ds \right) \\ 
\nu_2 &= \frac{1}{\Pi_2(0)} \left( \lambda \left[ G_2(1) - G_2(0) \right] - \theta \Pi_2(1) \right) \\ 
- \lambda \tilde{C}_2(s) &= \lambda \left( \frac{G_2(s) - G_2(0)}{s} \right) + 
\nu_2 \left( \frac{\Pi_2(s) - \Pi_2(0)}{s} \right) \\  
&+ \lambda \left( \frac{G_2(s) - G_2(1) }{1 - s} \right) - \theta \left( \frac{\Pi_2(s) - \Pi_2(1)}{1 - s} \right) \nonumber.
\end{align}
\end{subequations}
We can use this representation of the steady-state conditional measure $\ol{\mu}_t(dX_2)$ for group compositions featuring more than $x_{eq}$ cooperators, paired with fact that the probability of groups below $x_{eq}$ vanishes in the long-time limit, to find the long-time steady-state achieved by the multilevel dynamics in the Hawk-Dove case. Applying to Equation \eqref{eq:HDsteady2part} the change of variables $X_2 = \frac{x-x_{eq}}{1 - x_{eq}}$ and the definitions of the modified replication rates $\Pi_2(X_2)$ and $G_2(X_2)$ from Equations \eqref{eq:Pifunctions} and \eqref{eq:Gfunctions}, we find that the steady states of the multilevel HD dynamics take the following form:
 \begin{equation} \label{eq:HDofthetapiecewise}
   q^{\lambda}_{\theta}(x) = \left\{
     \begin{array}{lr}
       0 & : x < x_{eq}\\
       \ds\frac{g^{\lambda}_{\theta}(x)}{\int_{x_{eq}}^{1} g^{\lambda}_{\theta}(y) dy}  & : x \geq x_{eq},
     \end{array}
   \right.
\end{equation} 
with  $g^{\lambda}_{\theta}(x)$ given by
\begin{equation} \label{eq:HDgthetax}
    g^{\lambda}_{\theta}(x) = \left(x - x_{eq}\right)^{\nu_H - 1} \left(1 - x\right)^{\theta - 1} \left( \frac{\pi(1)}{x \pi(x)} \right) \exp\left(\int_{\frac{x-x_{eq}}{1 - x_{eq}}}^1 \frac{\left(-\lambda / (1-x_{eq}) \right) s \tilde{C}_H(s) }{(x_{eq} + (1 - x_{eq}) s) \pi(x_{eq} + (1-x_{eq}) s)} ds \right)
\end{equation}
and where $\nu_H$ and $-\lambda C_H(s)$ are given by
\begin{subequations} \label{eq:HDhelpers}
\begin{align}
\nu_H &= \frac{\lambda \left[G(1) - G(x_{eq})\right] - \theta \pi(1)}{x_{eq}(1-x_{eq}) \pi'(x_{eq})} \\
    -\lambda \tilde{C}_H(s) &= \lambda \left(\frac{G\left(x_{eq} + (1-x_{eq}) s\right) - G(x_{eq})}{s} +  \frac{G\left(x_{eq} + (1-x_{eq}) s \right) - G(1)}{1-s}\right) \\ & + \theta \left(\frac{s^{-1} \left(x_{eq} + (1-x_{eq}) s  \right) \pi\left( x_{eq} + (1-x_{eq}) s \right) - \pi(1)}{1-s} \right) \nonumber \\% 
    & + \nu_H \left(\frac{ \left[x_{eq} + (1-x_{eq}) s \right] \pi\left( x_{eq} + (1-x_{eq}) s \right) - s x_{eq} (1-x_{eq}) \pi'(x_{eq})}{s^2} \right)  \nonumber  
    \end{align}
\end{subequations}
In Theorem \ref{wHDlimtheorem}, we characterize the long-time weak convergence of solutions of the multilevel HD dynamics to steady states in the form described by Equations \eqref{eq:HDofthetapiecewise}, \eqref{eq:HDgthetax}, \eqref{eq:HDhelpers} for the case in which initial conditions have well-define H{\"o}lder data near $x=1$ and in which between-group selection is sufficiently strong. This result confirms and generalizes \cite[Conjecture 2]{cooney2020analysis}, which addresses convergence to steady-state in the case for which the replication rates arise from Hawk-Dove games with the payoff matrix of Equation \eqref{eq:payoffmatrix}. 
 \begin{theorem}\label{wHDlimtheorem}
Suppose that $G(x)$ and $\pi(x)$ satisfy the assumptions of Theorem \ref{prop:deltaHD} and that the initial distribution $\mu_0(dx)$ has \holder exponent $\theta$ near $x=1$ with corresponding positive, finite \holder constant $C_{\theta}$. If $\lambda \left[G(1) - G(x_{eq}) \right] > \theta \pi(1)$, then, for any continuous test-function $v(x)$, the solution $\overline{\mu}_t(dx)$ to Equation \eqref{multiselect} satisfies
\begin{equation}\label{vmubarlimHD}
\lim_{t\to \infty} \int_0^1 v(x)\ov{\mu}_t(dx)=\int_0^1 v(x)q^{\lambda}_{\theta}(x)dx.
\end{equation}
\end{theorem}

Using the same approach as in Section \ref{sec:steadystates}, we can characterize the threshold level of between-group competition required for an integrable steady state
\begin{equation} \label{eq:lambdastarHD}
    \lambda^*_H(\theta) := \frac{\theta \pi(1)}{G(1) - G(x_{eq})}.
\end{equation}
Then we can use this expression to see that the average payoff at steady state is given by
 \begin{equation} \label{eq:HDpayofflambdastar}
 \langle G(\cdot) \rangle_f = \left(\frac{\lambda^*_H(\theta)}{\lambda}\right) G(x_{eq}) + \left( 1 - \frac{\lambda^*_H(\theta)}{\lambda} \right) G(1),
 \end{equation}
 so average payoff interpolates between $G(x_{eq})$ at $\lambda = \lambda^*_H$ and $G(1)$ as $\lambda \to \infty$. Because $G(1) > G(x_{eq})$, we see that the average group payoff increases with $\lambda$ and that group payoff is limited by the average payoff of a full-cooperator group, even if group payoff $G(x)$ is maximized by an interior fraction of cooperators.

\end{document}